 \documentclass[11pt,a4paper]{article}
\usepackage[margin=0.6in]{geometry}
\usepackage{amsthm}
\usepackage{bm}
\usepackage[dvips]{graphics}
\usepackage{graphicx}
\usepackage{algorithm,}
\usepackage{courier}

\usepackage[allcolors=blue]{hyperref}

\usepackage{psfrag}
\usepackage{units}
\usepackage{sidecap}
\usepackage{amsfonts,amsmath,amssymb}
\hypersetup{colorlinks=true}

\usepackage{dsfont,color,setspace, cite}

\usepackage{caption,cite}
\usepackage{subcaption}
\usepackage{wrapfig}
\usepackage{comment}

\DeclareMathOperator*{\argmin}{arg\,min}

\usepackage{xspace}
\usepackage{bbm}

\def\tr{\mathop{\rm tr}\nolimits}%
\def\var{\mathop{\rm var}\nolimits}%
\def\diag{\mathop{\rm diag}\nolimits}%
\def\e{\epsilon}

\DeclareMathOperator\E{E}

\DeclareMathOperator\R{R}

\newcommand{\N}{\mathrm{N}}

\def\textiid{i.i.d.\@\xspace}
\newcommand\iid{\ifmmode\text{ i.i.d. } \else \textiid \fi}

\newtheorem{assumption}{Assumption}
\newtheorem{theorem}{Theorem}
\newtheorem{lemma}{Lemma}

\newtheorem{corollary}{Corollary}

\begin{document}
\long\def\/*#1*/{}

%
\newcommand{\blind}{1}  
\newcommand{\longer}{1} 

\if1\blind
{
  \title{\bf A scalable estimate of the out-of-sample prediction error  via approximate leave-one-out}
  \author{Kamiar Rahnama Rad  \thanks{
    K.R. is supported by the NSF DMS grant 1810888, and the Betty and Marvin Levine Fund. }
    \hspace{.2cm}\\
    Baruch College, City University of New York\\
    and \\
    Arian Maleki  \thanks{
    A.M. gratefully acknowledges NSF DMS grant 1810888. } \\
    Columbia University}
    \date{}
  \maketitle
} \fi

\if0\blind
{
  \bigskip
  \bigskip
  \bigskip
  \begin{center}
    {\LARGE\bf A scalable estimate of the out-of-sample prediction error  via approximate leave-one-out}
\end{center}
  \medskip
} \fi


\newcommand{\io}{\underline{1}}

\newcommand{\alo}{{\rm ALO}}
\newcommand{\lo}{{\rm LO}}

\newcommand{\extra}{{\rm Err}_{\rm extra}}
\newcommand{\ine}{\rm{Err}_{\rm in}}
\newcommand{\inew}{\widehat{{\rm Err}}_{\rm in}}
\newcommand{\p}{\mathds{P}}
\newcommand{\mb}{\mathbf{m}}   
\newcommand{\bb}{\mathbf{b}}
\newcommand{\bl}{\bm{\hat \beta}}
\newcommand{\blo}{\hat \beta^\circ}
\newcommand{\bli}{\bm{ \hat \beta}_{/i }}
\newcommand{\blio}{ \hat \beta_{/i }^\circ}
\newcommand{\df}{\text{df}}
\newcommand{\poly}{\rm{poly}}
\newcommand{\snr}{\text{snr}}
\newcommand{\blii}{  \bm{\tilde \beta}_{/i }}
\newcommand{\blt}{  \bm{\tilde \beta}^i}

\newcommand{\polyn}{\rm{poly}(\log n)}
\newcommand{\tXI}{{\bm{\bar X}}_{/i}}
\newcommand{\tGI}{{ \bm{\Gamma}}_{i/\bm{\tilde \delta_\ell}, \bm{\tilde \delta_r } }}
\newcommand{\tGII}{{ \bm{\Gamma}}_{ i/{\bm{\tilde \delta}_{1,\ell},\bm{\tilde \delta}_{1,r}  } }}
\newcommand{\tGIII}{{ \bm{\Gamma}}_{i/{\bm{\tilde \delta}_{2,\ell}, \bm{\tilde \delta}_{2,r} }}}
\newcommand{\zGI}{{\bm{ \Gamma}}_{\bm{\tilde \zeta} /i}}
\newcommand{\xGI}{{ \bm{\Gamma}}_{\bm{\tilde \xi} /i}}
\newcommand{\xxGI}{{ \bm{\Gamma}}_{ \bm{\xi} /i}}

\newcommand{\XI}{\bm{X}_{/i}}
\newcommand{\yi}{\bm{y}_{/i}}

\newcommand{\ld}{\dot{\ell}}
\newcommand{\ldd}{\ddot{\ell}}
\newcommand{\efd}{\dot{f}}
\newcommand{\efdd}{\ddot{f}}
\newcommand{\rd}{\dot{r}}
\newcommand{\psd}{\dot{\psi}}
\newcommand{\rdd}{\ddot{r}}
\newcommand{\rddd}{\dddot{r}}
\newcommand{\lddd}{\dddot{\ell}}
\newcommand{\dli}{\bm{\Delta}_{ \slash i}}
\newcommand{\pd}{\dot{\phi}}

\newcommand{\fd}{\dot{f}_e}
\newcommand{\fdd}{\ddot{f}_e}

\begin{abstract}
The paper considers the problem of out-of-sample risk estimation under the high dimensional settings where  standard techniques such as $K$-fold cross validation suffer from large biases. Motivated by the low bias of the leave-one-out cross validation (LO) method, we propose a computationally efficient closed-form approximate leave-one-out formula (ALO) for a large class of regularized estimators. Given the regularized estimate, calculating ALO requires minor computational overhead.  With minor assumptions about the data generating process, we obtain a finite-sample upper bound for $|\text{LO} - \text{ALO}|$. Our theoretical analysis illustrates that $|\text{LO} - \text{ALO}| \rightarrow 0$ with overwhelming probability, when $n,p \rightarrow \infty$, where the dimension $p$ of the feature vectors may be comparable with or even greater than the number of observations, $n$. Despite the high-dimensionality of the problem, our theoretical results do not require any sparsity assumption on the vector of regression coefficients. Our extensive numerical experiments  show that $|\text{LO} - \text{ALO}|$ decreases as $n,p$ increase, revealing the excellent finite sample performance of ALO. We further illustrate the usefulness of our proposed out-of-sample risk estimation method  by an example of real recordings from spatially sensitive neurons (grid cells) in the medial entorhinal cortex of a rat. 
\end{abstract}

\noindent%

\section{Introduction}\label{section:intro}


\subsection{Main objectives} 
Consider a dataset  $\mathcal{D} = \{ (y_1, \bm{x_1}), (y_2, \bm{x_2}), \ldots, (y_n, \bm{x_n})\}$ where $\bm{x_i} \in \R^p$ and $y_i \in \R$. In many applications, we model these observations as independent and identically distributed draws from  some joint distribution $q(y_i | \bm{x_i}^\top \bm{\beta}^*) p(\bm{x_i})$ where the superscript $\top$ denotes the transpose of a vector. To estimate the parameter $\bm{\beta}^*$ in such models, researchers often use the  optimization problem
\begin{eqnarray}\label{eq:ori_opt}
\bl \triangleq  \underset{\bm{\beta} \in \R^p}{\argmin}  \Bigl \{ \sum_{i=1}^n  \ell ( y_i|\bm{x_i}^\top \bm{\beta} ) + \lambda r(\bm{\beta})  \Bigr \}, \label{eq:bl}
\end{eqnarray}
where $\ell$ is called the loss function, and is typically set to $ - \log q(y_i | \bm{x_i}^\top \bm{\beta})$ when $q$ is known, and $r(\bm{\beta})$ is called the regularizer. In many applications, such as parameter tuning or model selection, one would like to estimate the \textit{out-of-sample prediction error}, defined as
\begin{equation}
\extra \triangleq \E [ \phi ( y_{\rm new},\bm{x}_{\rm new}^\top \bl ) | \mathcal{D} ], 
\end{equation}
where $(y_{\rm new},\bm{x}_{\rm new})$ is a new sample from the distribution $q(y | \bm{x}^\top \bm{\beta}^*) p(\bm{x})$ independent of $\mathcal{D}$, and $\phi$ is a function that measures the closeness of $y_{\rm new}$ to $\bm{x}_{\rm new}^\top \bl$. A standard choice for $\phi$ is $- \log q(y | \bm{x}^\top \bm{\beta})$.  However, in general we may use other functions too. Since $\extra$ depends on the rarely known  joint distribution of $(y_i,\bm{x_i})$, a core problem  in model  assessment is to estimate it from data.

 \begin{figure}
\hspace{5.2cm}
\begin{center}
        \includegraphics[width=0.8\textwidth]{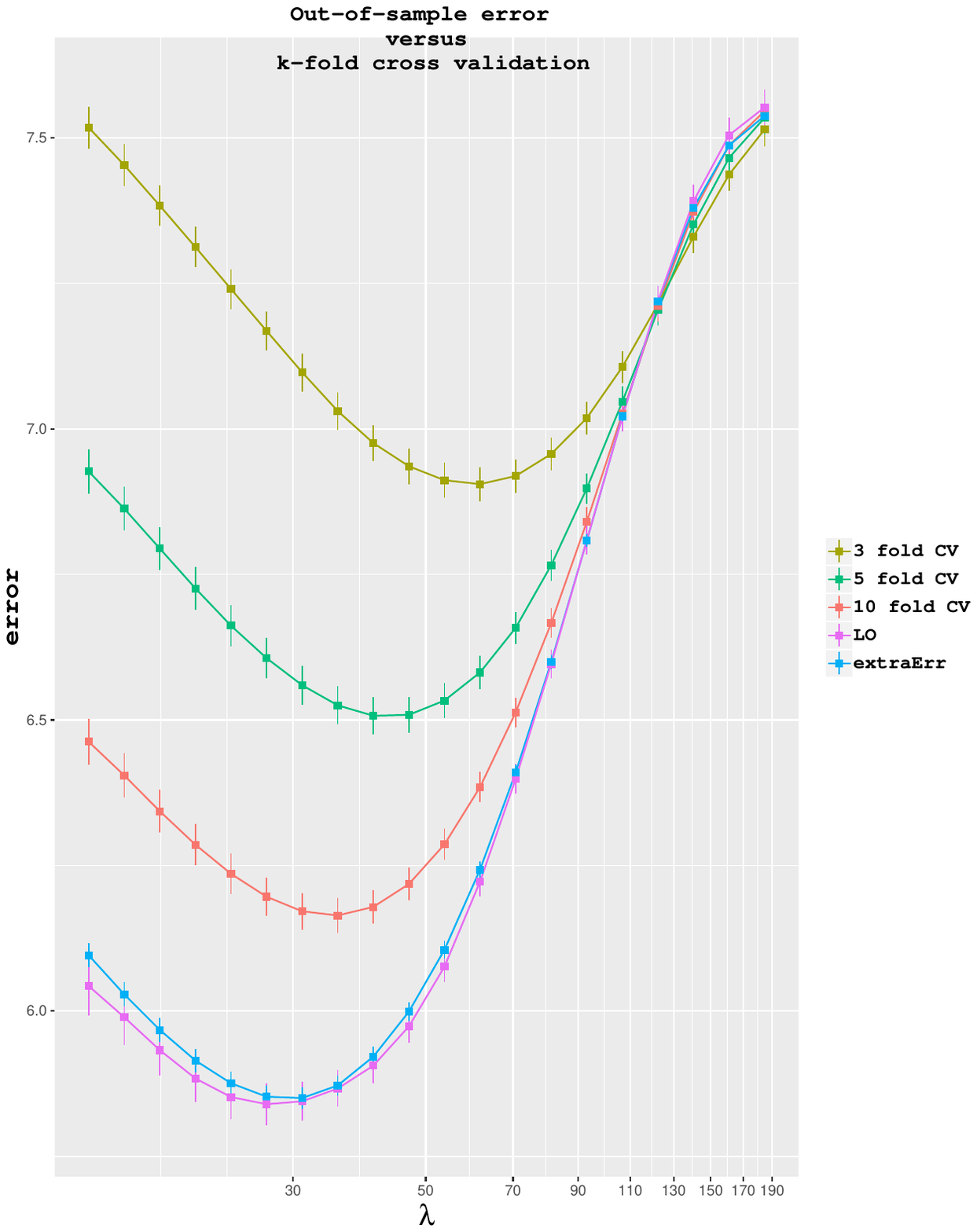}
          \caption{ Comparison of  $K$-fold cross validation (for $K=3,5,10$) and leave-one-out cross validation with the true (oracle-based)  out-of-sample  error for the LASSO problem where $\ell ( y|\bm{x}^\top \bm{\beta} )=\frac{1}{2}(y-\bm{x}^\top \bm{\beta})^2$ and $r(\bm{\beta})=\| \bm{\beta} \|_1$. In high-dimensional settings the upward bias of $K$-fold CV clearly decreases as number of folds increase. Data is $\bm{y} \sim \N(\bm{X \beta}^*,\sigma^2\bm{I})$ where $\bm{X} \in \R^{p \times n}$. The number of nonzero elements of the true $\bm{\beta}^*$ is set to $k$ and their values is set to $1/3$. Dimensions are $(p,n,k)=\bigl(1000,250,50\bigr)$ and  $\sigma=2$. The rows of $\bm{X}$ are independent $\N(\bm{0}, \bm{I})$. Extra-sample test data is $y_{\rm new} \sim \N(\bm{x}_{\rm new}^\top \bm{\beta}^*, \sigma^2)$ where  $\bm{x}_{\rm new} \sim \N(\bm{0}, \bm{I})$. The true (oracle-based) out-of-sample prediction error is $\extra =  \E [  ( y_{\rm new}-\bm{x}_{\rm new}^\top \bl )^2 | \bm{y,X} ]= \sigma^2 + \|\bl-\bm{\beta}^* \|_2^2$.  All depicted quantities are averages based on 500 random independent samples, and error bars depict one standard error. } 
          \label{fig:loovscv5}
          \end{center}
\end{figure}

 
 This paper considers a computationally efficient approach to the problem of  estimating  $\extra$ under the high-dimensional setting, where both $n$ and $p$ are large, but $n/p$ is a fixed number, possibly less than one. This high dimensional setting has received a lot of attention \cite{karoui2013asymptotic,EBBLY13,bean2013optimal,donoho2016high,NR16,SBC17,DW18}. But the problem of estimating $\extra$  has not been carefully studied in generality, and as a result the issues of the existing techniques and their remedies have not been explored. For instance, a popular technique in practice is the $K$-fold cross validation, where $K$ is a small number, e.g. $3$ or $5$. Figure \ref{fig:loovscv5} compares the performance of the $K$-fold cross validation for $4$ different values of $K$ on a  LASSO linear regression problem. This figure implies that in high-dimensional settings, $K$-fold cross validation suffers from a large bias, unless $K$ is a large number. This bias is due to the fact that in high-dimensional settings the fold that is removed in the training phase, may have a major effect on the solution of \eqref{eq:ori_opt}. This claim can be easily seen for LASSO linear regression with an IID design matrix using  phase transition diagrams \cite{DMM11}.  To summarize, as the number of folds increases, the bias of the estimates reduces at the expense of a higher computational complexity.

 In this paper, we consider the most extreme form of cross validation, namely leave-one-out cross-validation (LO), which according to Figure \ref{fig:loovscv5} is the least biased cross validation based estimate of the out-of-sample error.  We will use the fact that both $n$ and $p$ are large numbers to approximate LO for both smooth and non-smooth regularizers. Our estimate, called approximate leave-one-out (ALO), requires solving the optimization problem \eqref{eq:ori_opt} once.  Then, it uses $\bl$ to approximate $\lo$ without solving the optimization problem again. In addition to obtaining $\bl$, $\alo$ requires a matrix inversion and two matrix-matrix multiplications. Despite these extra steps $\alo$ offers a significant computational saving compared to $\lo$.  This point is illustrated in Figure \ref{fig:time}  by comparing the computational complexity of $\alo$ with that of $\lo$ and a single fit as both $n$ and $p$ increase for various data shapes, that is $n>p, \ n=p$, and $n<p$. Details of this simulation are given in Section \ref{sec:time}. 
 
The main algorithmic and theoretical contributions of this paper are as follows. First, our computational complexity comparison between $\lo$ and $\alo$, confirmed by extensive numerical experiments, show that $\alo$ offers a major reduction in the computational complexity of estimating the out-of-sample risk. Moreover,  with minor assumptions about the data generating process, we obtain a finite-sample upper bound for $|\text{LO} - \text{ALO}|$, proving that under the high-dimensional settings ALO presents a sensible approximation of LO for a large class of regularized estimation problems in the generalized linear family. Finally, we provide readily usable R implementation of ALO  online; see \texttt{https://github.com/Francis-Hsu/alocv}, and we illustrate the usefulness of our proposed out-of-sample risk estimation in unexpected scenarios that fail to satisfy the assumptions of our theoretical framework. Specifically, we present a novel neuroscience example about the computationally efficient tuning of the spatial scale in estimating an inhomogeneous spatial point process.


 \begin{figure} 
    \centering
    \begin{subfigure}[b]{0.32\textwidth}
        \includegraphics[width=\textwidth]{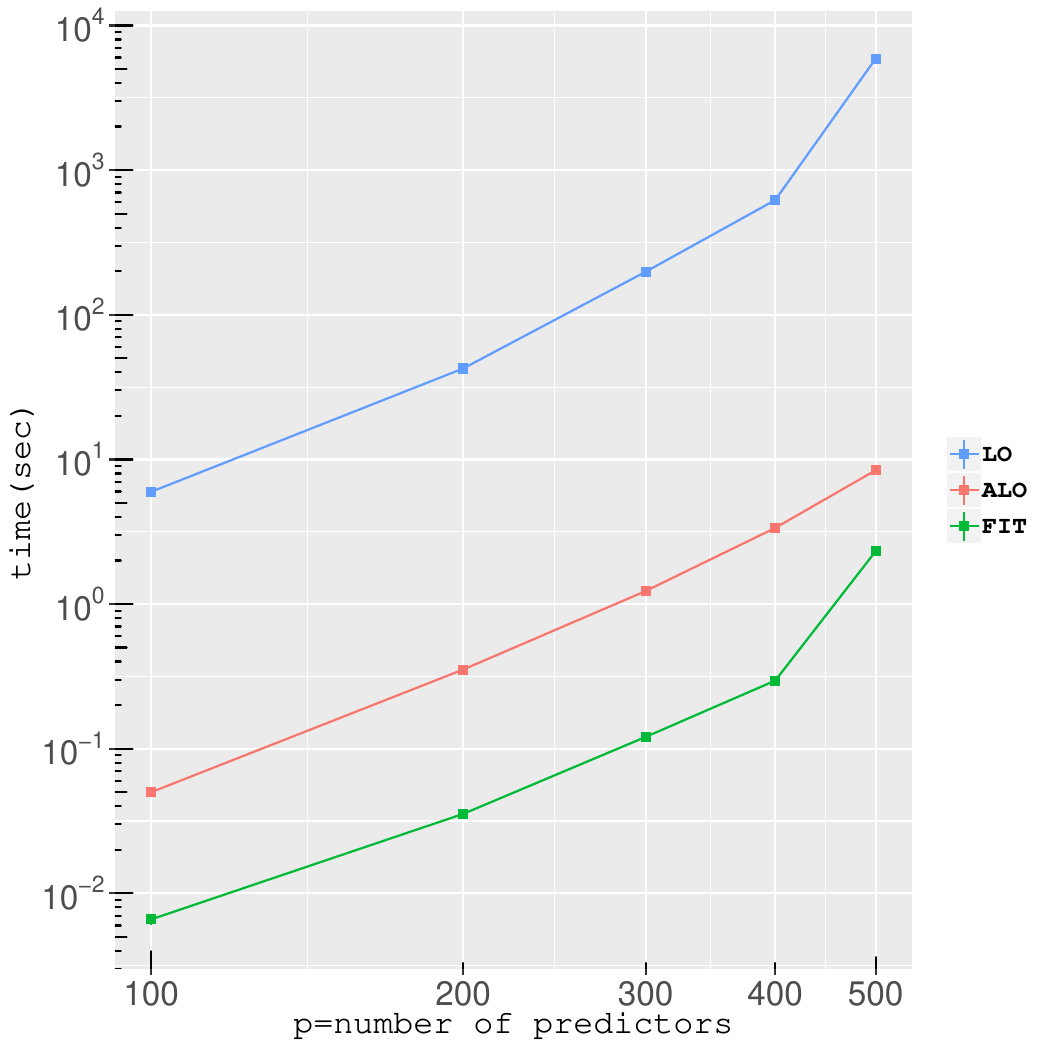}
        \caption{Elastic-net linear regression (section \ref{sec:num:linear}) for $\frac{n}{p}=5$.}
        \label{fig:linear_time}
    \end{subfigure}    
    \begin{subfigure}[b]{0.32\textwidth}
        \vspace{1cm}
        \includegraphics[width=\textwidth]{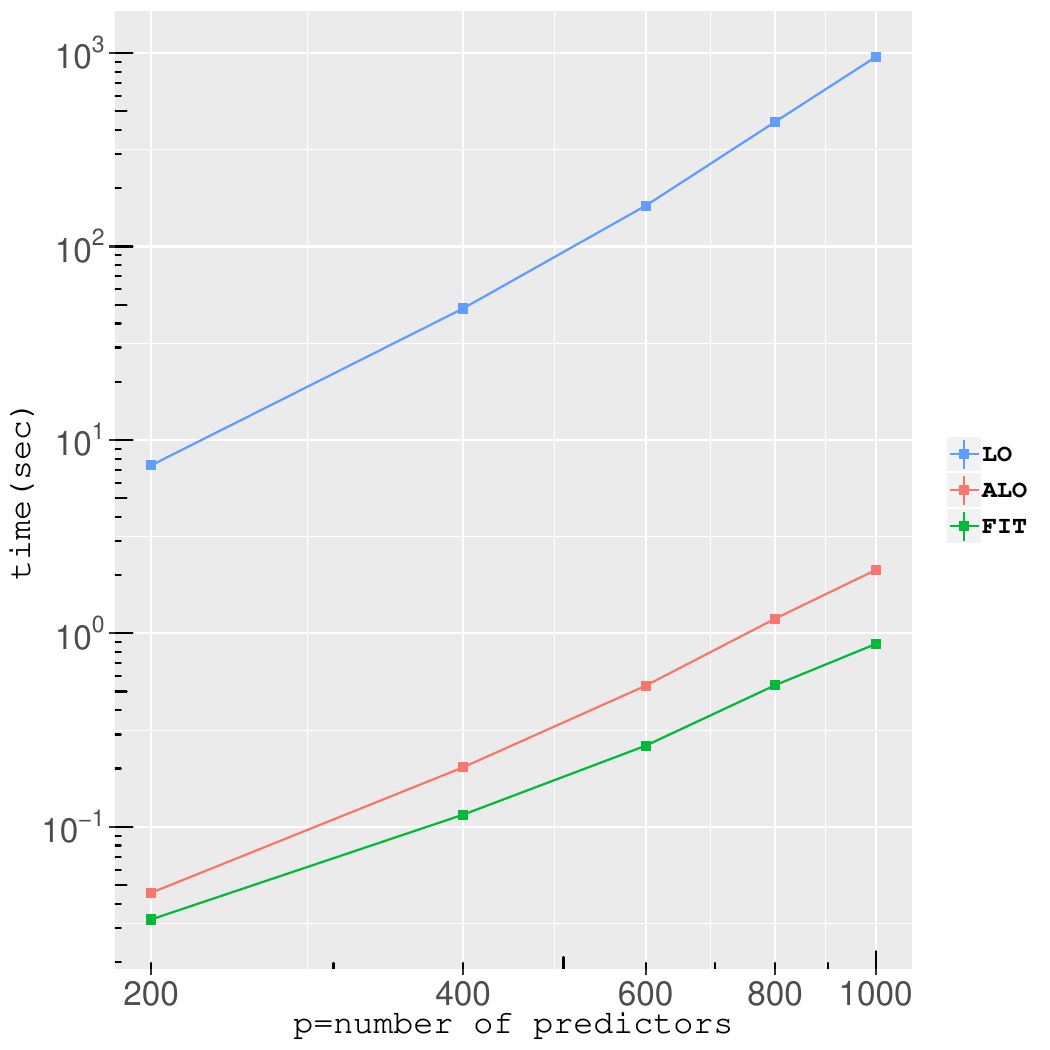}
        \caption{LASSO logistic  regression (section \ref{sec:num:logistic}) for $\frac{n}{p} =1$. }
        \label{fig:logistic_time}
    \end{subfigure}
        ~ 
    \begin{subfigure}[b]{0.32\textwidth}
        \includegraphics[width=\textwidth]{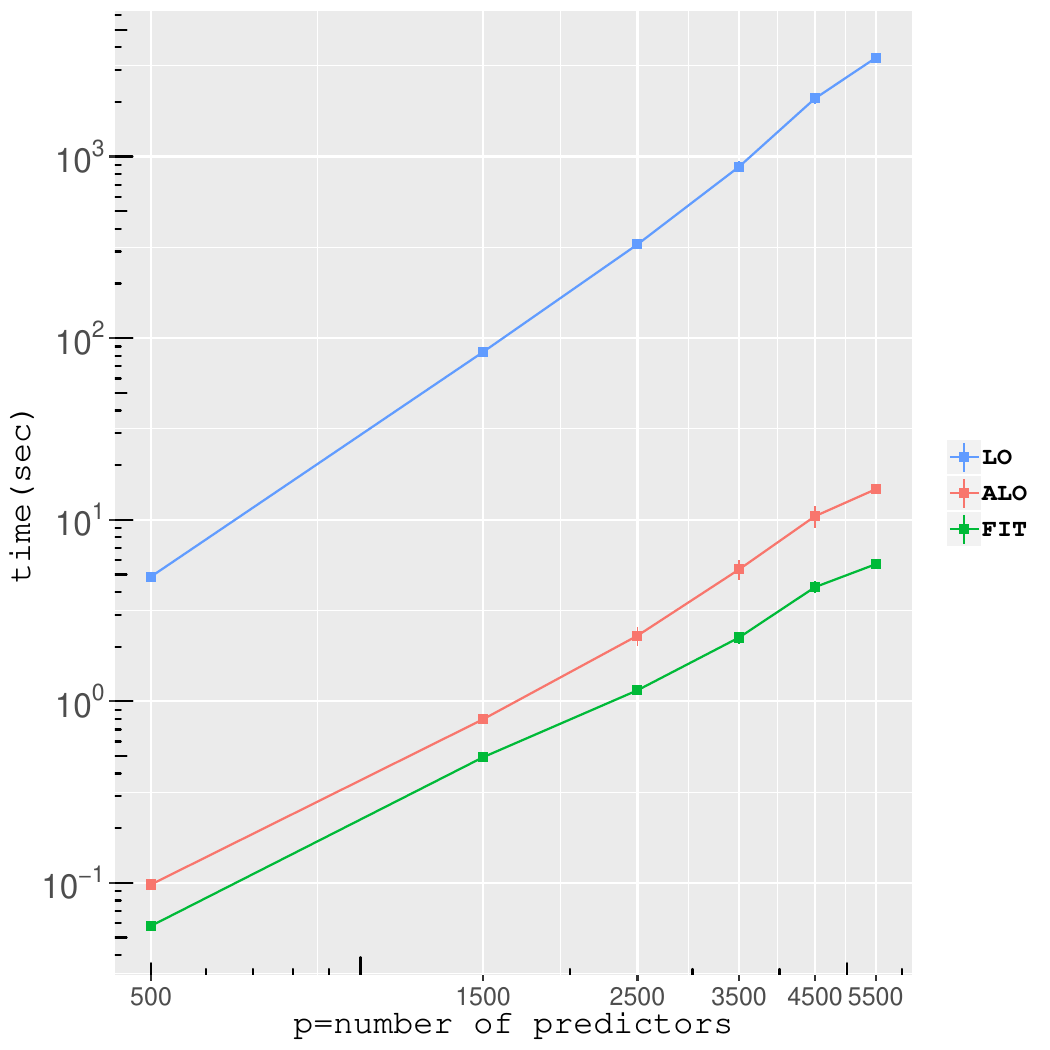}
        \caption{Elastic-net Poisson  regression  (section \ref{sec:num:poisson}) for $\frac{n}{p}=\frac{1}{10}$.}
        \label{fig:poisson_time}
    \end{subfigure}
    \caption{The time to compute $\alo$ and $\lo$. FIT refers to the time to fit $\bl$ and the $\alo$ time includes computing $\bl$.  
          Calculating $\lo$ takes orders of magnitude longer than $\alo$. }\label{fig:time}
\end{figure}



\subsection{Relevant work}
 The problem of estimating $\extra$ from $\mathcal{D}$ has been  studied for (at least) the past 50 years. Methods such as cross validation (CV) \cite{S74,G75}, Allen's PRESS statistic \cite{A74}, generalized cross validation (GCV) \cite{CW79,GHW79}, and bootstrap \cite{E83} have been proposed for this purpose. In the high dimensional setting, employing LO or bootstrap is computationally expensive and the less computationally complex approaches such as  5-fold (or 10-fold) CV suffer from high bias as illustrated in Figure \ref{fig:loovscv5}. 

 As for the computationally efficient approaches, extensions of Allen's PRESS \cite{A74}, and generalized cross validation (GCV) \cite{CW79,GHW79} to non-linear models and classifiers with ridge penalty are well known: smoothing splines for generalized linear models in \cite{OYR86}, spline estimation of generalized additive models \cite{B90}, ridge estimators in logistic regression in \cite{CH92}, smoothing splines with non-Gaussian data using various extensions of GCV in \cite{G92,XW96,GX01}, support vector machines \cite{OW00}, kernel logistic regression in \cite{CT08}, and Cox's proportional hazard model with a ridge penalty in \cite{MG13}. Moreover, leave-one-out approximations for posterior means of Bayesian models with Gaussian process priors using the Laplace approximation and Expectation Propagation were introduced in \cite{VMTSW16}, and extended in \cite{VGG17}. Despite the existence of this vast literature, the performance of such approximations in high-dimensional settings is unknown except for the straightforward linear ridge regression framework. Moreover, past heuristic approaches have only considered the ridge regularizer. The results of this paper include a much broader set of regularizers; examples include but are not limited to LASSO \cite{T96}, elastic net \cite{ZH05} and bridge \cite{FF93}, just to name a few.

 

 
More recently, a few  papers have studied the problem of estimating $\extra$ under high-dimensional settings \cite{mousavi2013asymptotic, OK16}. The approximate message passing framework introduced in \cite{MalekiThesis,DoMaMo09} was used in \cite{mousavi2013asymptotic}  to obtain an estimate of $\extra$ for LASSO linear regression. In another related paper, \cite{OK16} obtained similar results using approximations popular in statistical physics. The results of \cite{mousavi2013asymptotic} and \cite{OK16} are only valid for cases where the design matrix has IID entries and the empirical distribution of the regression coefficients converges weakly to a distribution with a bounded second moment. In this paper, our theoretical analysis includes correlated design matrices, and regularized estimators beyond  LASSO linear regression. 

In addition to these approaches, another  contribution has been to study GCV and $\extra$ for restricted least-squares estimators of submodels of the overall model without regularization \cite{BF83, leeb2008evaluation,L09}.  In \cite{leeb2008evaluation} it was  shown that a variant of GCV  converges  to $\extra$ uniformly over a collection of candidate models provided that there are not too many candidate models, ruling out complete subset selection. Moreover, since restricted least-squares estimators are studied, the conclusions  exclude the regularized problems considered in this paper.

Finally, it is worth mentioning that in another line of work,  strategies have been proposed  to obtain unbiased estimates of the in-sample error. In contrast to the out-of-sample error, the in-sample error is about the prediction  of new responses for the same explanatory variables as in the training data. The literature of in-sample error estimation is too vast to be reviewed here. Mallow's $C_p$ \cite{M73}, Akaike's Information Criterion (AIC) \cite{Ak74,HT89}, Stein's Unbiased Risk Estimate (SURE) \cite{S81,ZHT07,TT12} and Efron's Covariance Penalty \cite{E86}  are seminal examples of in-sample error estimators. When $n$ is much larger than $p$,  the in-sample prediction error is expected to be close to the out-of-sample prediction error. The problem is that in high-dimensional settings, where $n$ is of the same order as (or even smaller than) $p$, the in-sample and out-of-sample errors are different.

The rest of the paper is organized as follows. After introducing the notations, we first present the approximate leave-one-out formula ($\alo$) for twice differentiable regularizers in Section \ref{ssec:newtonmethod}. In Section \ref{ssec:non-diff}, we show how $\alo$ can be extended to nonsmooth regularizers such as LASSO using Theorem \ref{thm:lasso_approx} and Theorem \ref{thm:lassobounderror}. In Section \ref{sec:computations}, we compare the computational complexity and memory requirements of ALO and LO.  In Section \ref{sec:res}, we present Theorem \ref{theo:main}, illustrating with minor assumptions about the data generating process that $|\text{LO} - \text{ALO}| \rightarrow 0$ with overwhelming probability, when $n,p \rightarrow \infty$, where $p$ may be comparable with or even greater than $n$. The numerical examples in Section \ref{sec:numsim} study the statistical accuracy and computational efficiency of the approximate leave-one-out approach. To illustrate the accuracy and computational efficiency of $\alo$ we apply it to synthetic and real data in Section \ref{sec:numsim}. We generate synthetic data, and compare $\alo$ and $\lo$ for elastic-net linear regression in Section \ref{sec:num:linear}, LASSO logistic regression in Section \ref{sec:num:logistic}, and elastic-net Poisson regression in Section \ref{sec:num:poisson}. For real data we apply LASSO, elastic-net and ridge logistic regression to sonar returns from two undersea targets in Section \ref{sec:sonar}, and we apply LASSO Poisson regression to real recordings from spatially sensitive neurons (grid cells) in Section \ref{sec:num:grid}. Our synthetic and real data examples  cover various data shapes, that is $n > p$, $n=p$ and $n<p$. 
In Section \ref{sec:conc} we discuss directions for future work. Technical proofs are collected in Section \ref{sec:proof}, the appendix.




  \subsection{Notation}
  
 We first review the notations that will be used in the rest of the paper.    Let $\bm{x_i}^\top \in \R^{1 \times p}$ stand for the $i$th row of $\bm{X} \in \R^{n \times p}$. $\bm{y}_{/i} \in \R^{(n-1) \times 1}$ and $\XI \in \R^{(n-1) \times p}$ stand for $\bm{y}$ and $\bm{X}$, excluding the $i$th entry $y_i$ and the $i$th row $\bm{x_i}^\top$, respectively.  The vector $\bm{a} \odot \bm{b}$ stands for the entry-wise product of two vectors $\bm{a}$ and $\bm{b}$. For two vectors $\bm{a}$ and $\bm{b}$, we use $\bm{a} < \bm{b}$ to indicate element-wise inequalities. Moreover, $|\bm{a}|$ stands for the vector obtained by applying the element-wise absolute value to every element of $\bm{a}$. For a set $S \subset \{1,2, 3, \ldots, p\}$, let $\bm{X}_S$ stands for the submatrix of $\bm{X}$ restricted to \textit{columns} indexed by $S$. Likewise, we let $\bm{x}_{i,S} \in \R^{|S| \times 1}$ stand for for subvector of $\bm{x}_i$ restricted to the entries indexed by $S$. For a vector $\bm{a}$, depending on which notation is easier to read, we may use $[\bm{a}]_i$ or $a_i$ to denote the $i$th entry of $\bm{a}$. The diagonal matrix with elements of the vector $\bm{a}$ is referred to as $\text{diag}[\bm{a}]$.  Moreover, define 
\begin{eqnarray*}
\pd(y,z) &\triangleq& \frac{\partial \phi (y,z)  }{\partial z}, \quad \ld_i(\bm{\beta}) \triangleq \frac{\partial \ell (y_i|z)  }{\partial z} |_{z =  \bm{x_i}^\top \bm{\beta}}, \quad  \ldd_i(\bm{\beta}) \triangleq \frac{\partial ^2\ell (y_i|z)  }{\partial z^2} |_{z = \bm{x_i}^\top \bm{\beta}} \\
\bm{ \ld}_{/ i}(.) &\triangleq& [\ld_1(.),\cdots,\ld_{i-1}(.),\ld_{i+1}(.),\cdots,\ld_n(.)]^\top, \\
  \bm{\ldd}_{/ i}(.) &\triangleq&  [\ldd_1(.),\cdots,\ldd_{i-1}(.),\ldd_{i+1}(.),\cdots,\ldd_n(.)]^\top.
\end{eqnarray*}
The notation $\poly\log n$ denotes  polynomial of $\log n$ with a finite degree. Finally,  let $\sigma_{\max}(\bm{A})$ and $\sigma_{\min}(\bm{A})$ stand for the largest and smallest singular values of $\bm{A}$, respectively. 
%
 %
 %
 \section{Approximate leave-one-out}\label{ssec:aloformula}
 \subsection{Twice differentiable losses and regularizers}\label{ssec:newtonmethod}
 The leave-one-out cross validation estimate is defined through the following formula:
  \begin{eqnarray}
\lo &\triangleq& \frac{1}{n} \sum_{i=1}^n \phi(y_i, \bm{x}_i^\top \bli),\label{eq:loo}
\end{eqnarray}
\vspace{-0.5cm}
where 
\vspace{-.5cm}
\begin{eqnarray}\label{eq:optimization_lo}
\bli \triangleq  \underset{\bm{\beta} \in \R^p}{\argmin}  \Bigl \{ \sum_{j\neq i}  \ell ( y_j|\bm{x}_j^\top \bm{\beta} ) + \lambda r(\bm{\beta})  \Bigr \} \label{eq:bli},
\end{eqnarray}
is the leave-$i$-out estimate. If done naively, the calculation of $\lo$ asks for the optimization problem \eqref{eq:bli}  to be solved $n$ times, a computationally demanding task  when $p$ and $n$ are  large. To resolve this issue, we use the following simple strategy: Instead of solving \eqref{eq:bli} accurately, we use one step of the Newton method for solving \eqref{eq:bli} with initialization $\bl$. Note that this step requires both $\ell$ and $r$ to be twice differentiable. We will explain how this limitation can be lifted in the next section. The Newton step leads to the following simple approximation of $\bli$:\footnote{Note that in the rest of the paper for notational simplicity of our theoretical results we have assumed that $r(\bm{\beta}) = \sum_{i=1}^p r(\beta_i)$. However, the extension to non-separable regularizers is straightforward.}
\[
\boldsymbol{\tilde{\beta}}_{ \slash i} = \bl + \Big( \sum_{j\neq i} \bm{x_j} \bm{x_j}^\top \ldd ( y_j|\bm{x}_j^\top \bl ) + \lambda \diag[\bm{\rdd}(\bl)] \Big)^{-1} \bm{x_i} \ld ( y_i|\bm{x}_i^\top \bl ),
\]
 where $\bl$ is defined in \eqref{eq:ori_opt}. Note that $\sum_{j\neq i} \bm{x_j} \bm{x_j}^\top \ldd ( y_j|\bm{x}_j^\top \bl ) + \lambda \diag[ \bm{\rdd}(\bl)]$ is still dependent on the observation that is removed. Hence, the process of computing the inverse (or solving a linear equation) must be repeated $n$ times. Standard methods for calculating inverses (or solving linear equations) require cubic time and quadratic space (see Appendix C.3 in \cite{CONV04}), rendering them
impractical for high-dimensional applications when repeated $n$ times\footnote{A natural idea for reducing the computational burden involves exploiting structures (such as sparsity and banded-ness) of the involved matrices. However, in this paper we do not make any assumption regarding the structure of $\bm{X}$.}. We use the Woodburry lemma to reduce the computational cost:
\begin{eqnarray}
\Big( \sum_{j\neq i} \bm{x_j} \bm{x_j}^\top \ldd ( y_j|\bm{x}_j^\top \bl ) + \lambda \diag[{\bm \rdd}(\bl)] \Big)^{-1} =  \bm{J}^{-1}
+ \frac{\bm{J}^{-1} \bm{x_i} \ldd ( y_i|\bm{x}_i^\top \bl ) \bm{x_i}^\top \bm{J}^{-1}}{1- \bm{x_i}^\top \bm{J}^{-1} \bm{x_i} \ldd ( y_i|\bm{x}_i^\top \bl ) },
\end{eqnarray}
where $\bm{J} = ( \sum_{j=1}^n \bm{x_j} \bm{x_j}^\top \ldd ( y_j|\bm{x}_j^\top \bl )+\lambda \diag[{\bm \rdd}(\bl) ])$. 
Following this approach we define $\alo$ as
\begin{eqnarray}\label{eq:aloformulafinal}
\alo &\triangleq& \frac{1}{n} \sum_{i=1}^n \phi \left (y_i,  \bm{x_i}^\top  \boldsymbol{\tilde{\beta}}_{ \slash i}       \right)=  \frac{1}{n} \sum_{i=1}^n \phi \left (y_i,  \bm{x_i}^\top \bl +\left(\frac{\ld_i(\bl)}{\ldd_i(\bl)}\right) \left(  \frac{H_{ii}}{1 - H_{ii}} \right)     \right),
\end{eqnarray}
where
\begin{eqnarray}
 \bm{H} &\triangleq&  \bm{X} \left (\lambda \diag[\bm{\rdd}(\bl)] + \bm{X}^\top \diag[\bm{\ldd}(\bl)] \bm{X} \right)^{-1}  \bm{X}^\top \diag[\bm{\ldd}(\bm{\bl})]. 
 \end{eqnarray}
 Algorithm 1 summarizes how one should obtain an $\alo$ estimate of the $\extra$. We will show that under the high-dimensional settings one Newton step is sufficient for obtaining a good approximation of $\bli$, and the difference $|\alo-\lo|$ is small when either $n$ or both $n,p$ are large. However, before that we resolve the differentiability issue of the approach we discussed above. 

 \begin{algorithm}\label{alg:1}
\caption{Risk estimation with ALO for twice differentiable losses and regularizers}\label{euclid}
\textbf{Input.} $(\bm{x_1}, y_1), (\bm{x_2}, y_2), \ldots, (\bm{x_n}, y_n)$.  \\
\textbf{Output.} $\extra$ estimate.
\begin{enumerate}
\item Calculate  $\bl =  \underset{\bm{\beta} \in \R^p}{\argmin}  \Bigl \{ \sum_{i=1}^n  \ell ( y_i|\bm{x_i}^\top \bm{\beta} ) + \lambda r(\bm{\beta})  \Bigr \}.$

\item Obtain  $\bm{H} =  \bm{X} \left (\lambda \diag[\bm{\rdd}(\bl)] + \bm{X}^\top \diag[\bm{\ldd}(\bl)] \bm{X} \right)^{-1}  \bm{X}^\top \diag[\bm{\ldd}(\bm{\bl})].$ 

\item The estimate of $\extra$ is given by $\frac{1}{n} \sum_{i=1}^n \phi \left (y_i,  \bm{x_i}^\top \bl +\left(\frac{\ld_i(\bl)}{\ldd_i(\bl)}\right) \left(  \frac{H_{ii}}{1 - H_{ii}} \right)     \right)$.

\end{enumerate}
\end{algorithm}
 
 \subsection{Nonsmooth regularizers}\label{ssec:non-diff}
 
The Newton step, used in the derivation of $\alo$, requires the twice differentiability of the loss function and regularizer. However, in many modern applications non-smooth regularizers, such as LASSO, are preferable. In this section, we explain how $\alo$ can be used for non-smooth regularizers. We start with the $\ell_1$-regularizer, and then extend it to the other bridge estimators. A similar approach can be used for other non-smooth regularizers. Consider
 \begin{eqnarray}\label{eq:ori_optalpha}
\bl \triangleq  \underset{\bm{\beta} \in \R^p}{\argmin}  \Bigl \{ \sum_{i=1}^n  \ell ( y_i|\bm{x_i}^\top \bm{\beta} ) + \lambda \|\bm{\beta}\|_1  \Bigr \}. 
\end{eqnarray}
Let  $\bm{\hat g}$ be a  subgradient of $\| \bm{\beta} \|_1$ at $ \bl$, denoted by $\bm{\hat g} \in \partial \| \bl \|_1$. Then, the pair $(\bl, \bm{\hat g})$ must satisfy the zero-subgradient condition
\[
\sum_{i=1}^n \bm{x_i} \ld(y_i| \bm{x_i}^\top {\bl}) + \lambda \bm{\hat g}=0.
\]
As a starting point we use a smooth approximation of  the function $\|\bm{\beta}\|_1$ in our $\alo$ formula. For instance, we can use the following approximation  introduced in \cite{schmidt2007fast}:
\[
r^\alpha(\bm{\beta}) = \sum_{i=1}^p \frac{1}{\alpha} \Big( \log(1+ e^{\alpha \beta_i}) + \log(1+ e^{-\alpha \beta_i}) \Big). 
\] 
Since $\lim_{\alpha \rightarrow \infty} r^{\alpha} (\bm{\beta}) = \|\bm{\beta}\|_1$, we can use
\begin{eqnarray}
 \bl^{\alpha} &\triangleq&  \underset{\bm{\beta} \in \R^p}{\argmin}  \Bigl \{ \sum_{i=1}^n  \ell ( y_i|\bm{x_i}^\top \bm{\beta} ) + \lambda \sum_{i=1}^p r^{\alpha} (\beta_i)  \Bigr \}, \label{eq:smoothlasso}
 \end{eqnarray}
to obtain the following formula for $\alo$: 
\begin{eqnarray}
\alo^{\alpha} &\triangleq& \frac{1}{n} \sum_{i=1}^n \phi \left (y_i,  \bm{x_i}^\top \bl^{\alpha} +\left(\frac{\ld_i(\bl^{\alpha})}{\ldd_i(\bl^{\alpha})}\right) \left(  \frac{H^{\alpha}_{ii}}{1 - H^{\alpha}_{ii}} \right)     \right) \label{eq:alosmooth}
\end{eqnarray}
where $ \bm{H}^{\alpha} \triangleq  \bm{X} \left (\lambda \diag[\bm{\rdd}(\bl^{\alpha})] + \bm{X}^\top \diag[\bm{\ldd}(\bl^{\alpha})] \bm{X} \right)^{-1}  \bm{X}^\top \diag[\bm{\ldd}(\bl^{\alpha})].$ Note that $\| \bl^{\alpha} - \bl \|_2 \rightarrow 0$ as $\alpha \rightarrow \infty$, according to Lemma \ref{thm:firstbdlasso} in Section \ref{ssec:proofthmnonsmooth}. Therefore, we take the $\alpha \rightarrow \infty$ limit in \eqref{eq:alosmooth}, yielding  a simplification of $\alo^{\alpha}$ in this limit. To prove this claim, we denote the active set of $\bl$ with $S$, and we suppose the following: 

\begin{assumption} \label{A1}
$\bl$ is the unique global minimizer of \eqref{eq:bl}.  
\end{assumption}

\begin{assumption} \label{A3} 
$\bl^{\alpha}$ is the unique global minimizer of  \eqref{eq:smoothlasso} for every value of $\alpha$.
\end{assumption}

\begin{assumption} \label{A4}
$\ldd(y|\bm{x}^\top\bm{\beta})$ is a continuous function of $\bm{\beta}$.
\end{assumption}

\begin{assumption} \label{A2}
\label{ass:A.2} 
The strict dual feasibility condition $ \| \bm{\hat g}_{ S^c}  \|_{\infty} < 1$ holds.
\end{assumption}


%
%
%
%

\begin{theorem}\label{thm:lasso_approx}
If Assumptions \ref{A1},\ref{A3},\ref{A4} and \ref{A2}  hold, then 
\begin{equation}\label{eq:LASSO_ALO}
\lim_{\alpha \rightarrow \infty} \alo^{\alpha} = \frac{1}{n} \sum_{i=1}^n \phi \left (y_i,  \bm{x_i}^\top \bl +\left(\frac{\ld_i(\bl)}{\ldd_i(\bl)}\right) \left(  \frac{H_{ii}}{1 - H_{ii}} \right)     \right),
\end{equation}
where $ \bm{H}=  \bm{X}_S \left ( \bm{X}_S^\top \diag[\bm{\ldd} (\bl)]   \bm{X}_S  \right)^{-1} \bm{X}_S^\top \diag[\bm{\ldd} (\bl)].$
\end{theorem} 
The proof of this theorem is presented in Section \ref{ssec:proofthmnonsmooth}. For the rest of the paper,  the right hand side of \eqref{eq:LASSO_ALO} is the $\alo$ formula  we use as an approximation of $\lo$ for LASSO problems. In the simulation section, we show that the formula  we obtain in Theorem \ref{thm:lasso_approx} offers an accurate estimate of the out-of-sample prediction error. For instance, in the standard LASSO problem, where $\ell (u,v) = (u-v)^2/2$ and $r(\bm{\beta})=\|\bm{\beta} \|_1$, Theorem \ref{thm:lasso_approx} gives the following estimate of the out-of-sample prediction error:
\begin{eqnarray}
 \lim_{\alpha \rightarrow \infty} \alo^{\alpha} =  \frac{1}{n}\sum_{i=1}^n \frac{(y_i - \bm{x}_i^\top \bl)^2}{(1- H_{ii})^2},\label{eq:alo-lasso}
\end{eqnarray}
where $ \bm{H}=  \bm{X}_S \left ( \bm{X}_S^\top  \bm{X}_S  \right)^{-1} \bm{X}_S^\top$. Figure \ref{fig:extra-alo} compares this estimate with the oracle estimate of the out-of-sample prediction error on a LASSO example. More extensive simulations are reported in Section \ref{sec:numsim}.

\begin{figure}
\hspace{2.5cm}
\includegraphics[width=0.6\textwidth]{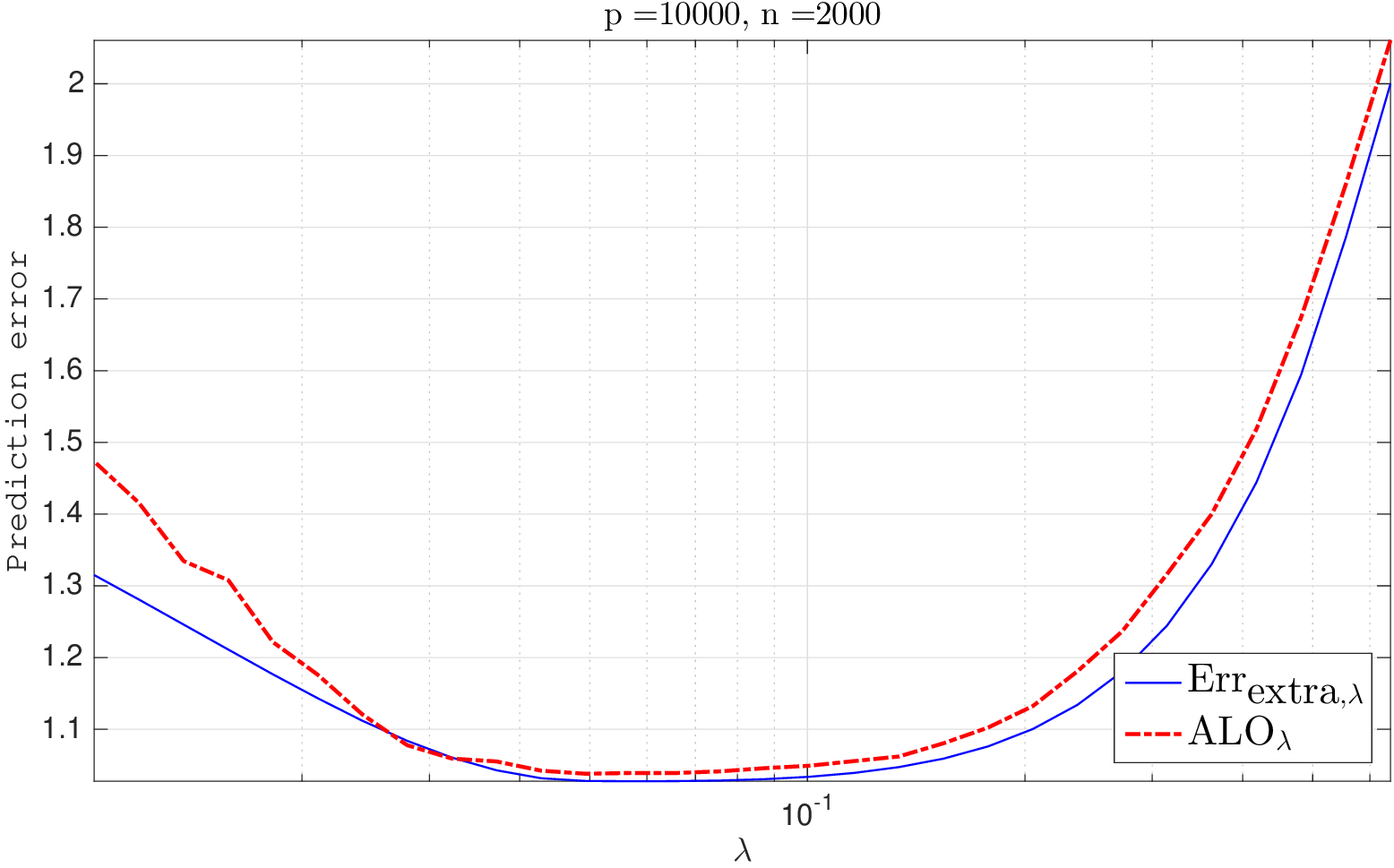} \caption{ Out-of-sample prediction  error  versus $\alo$. Data is $\bm{y} \sim \N(\bm{X \beta}^*,\sigma^2\bm{I})$ where $\sigma^2=1$ and $\bm{X} \in \R^{p \times n}$ with $p=10000$ and $n=2000$. The number of nonzero elements of the true $\bm{\beta}^*$ is set to $k=400$ and their values is set to 1. The rows $\bm{x_i}^\top$ of the predictor matrix are  generated randomly as $\N(0,\bm{\Sigma})$ with correlation structure ${\rm cor}(X_{ij},X_{ij'})=0.3$  for all $i=1,\ldots,n$ and $j,j'=1,\cdots,p$. The covariance matrix $\bm{\Sigma}$ is scaled such the signal variance $\var(\bm{x}^\top \bm{\beta}^*) =1$.  Out-of-sample test data is $y_{\rm new} \sim \N(\bm{x}_{\rm new}^\top \bm{\beta}^*, \sigma^2)$ where $\bm{x}_{\rm new} \sim \N(0,\bm{\Sigma})$. Out-of-sample error is calculated as $\E_{(y_{\rm new}, \bm{x}_{\rm new})} [  ( y_{\rm new}-\bm{x}_{\rm new}^\top \bl )^2 | \bm{y},\bm{X} ]=\sigma^2 +  \| \bm{\Sigma}^{1/2}(\bl-\bm{\beta}^* \|_2^2$ and $\alo$ is calculated using equation \eqref{eq:alo-lasso}.
}\label{fig:extra-alo}
\end{figure}

Note that Assumptions \ref{A1},\ref{A3} and \ref{A4} hold for most of the practical problems. For instance, to study the conditions under which Assumption \ref{A1} holds refer to \cite{tibshirani2013lasso}. Moreover, for $\ell (u,v) = (u-v)^2/2$, Assumption \ref{A1} is a consequence of Assumption \ref{A2} \cite{W09}.  Assumption  \ref{A2} also holds in many cases with probability one with respect to the randomness of the dataset \cite{W09,TT12}. Even if this assumption is violated in a specific problem (note that checking this assumption is straightforward), we can use the following theorem to evaluate the accuracy of the ALO formula in Theorem \ref{thm:lasso_approx}. 

\begin{theorem}\label{thm:lassobounderror}
Let $S$ and $T$ denote the active set of $\bl$, and the set of zero coefficients at which the subgradient vector is equal to $1$  or $-1$. Then, 
\begin{eqnarray*}\label{eq:subgradient}
&&\bm{x}_{i,S}^\top \left( \bm{X}_S^\top \diag[\bm{\ldd} (\bl)]   \bm{X}_S  \right)^{-1} \bm{x}_{i,S} \ldd_i(\bl)< \lim\inf_{\alpha \rightarrow \infty} H^{\alpha}_{ii}  \nonumber \\
&& \lim\sup_{\alpha \rightarrow \infty} H_{ii}^\alpha< \bm{x}_{i,S \cup T}^\top \left( \bm{X}_{S\cup T}^\top \diag[\bm{\ldd} (\bl)]   \bm{X}_{S \cup T}  \right)^{-1} \bm{x}_{i,S\cup T} \ldd_i(\bl)
\end{eqnarray*}
\end{theorem}

This Theorem is proved in \ref{ssec:proofthmapprox}.  A simple implication of this theorem is that
\begin{equation}\label{alo:lasso:low}
\lim\sup_{\alpha \rightarrow \infty} \alo^{\alpha} \leq \frac{1}{n} \sum_{i=1}^n \phi \left (y_i,  \bm{x_i}^\top \bl +\left(\frac{\ld_i(\bl)}{\ldd_i(\bl)}\right) \left(  \frac{H^h_{ii}}{1 - H^h_{ii}} \right)     \right),
\end{equation}
and
\begin{equation}\label{alo:lasso:high}
\lim\inf_{\alpha \rightarrow \infty} \alo^{\alpha} \geq \frac{1}{n} \sum_{i=1}^n \phi \left (y_i,  \bm{x_i}^\top \bl +\left(\frac{\ld_i(\bl)}{\ldd_i(\bl)}\right) \left(  \frac{H^l_{ii}}{1 - H^l_{ii}} \right)     \right),
\end{equation}
where
\begin{eqnarray}
 \bm{H}^l&=&  \bm{X}_S \left ( \bm{X}_S^\top \diag[\bm{\ldd} (\bl)]   \bm{X}_S  \right)^{-1} \bm{X}_S^\top \diag[\bm{\ldd} (\bl)], \nonumber \\
  \bm{H}^h&=&  \bm{X}_{S \cup T} \left ( \bm{X}_{S \cup T} ^\top \diag[\bm{\ldd} (\bl)]   \bm{X}_{S \cup T}  \right)^{-1} \bm{X}_{S \cup T}^\top \diag[\bm{\ldd} (\bl)].
\end{eqnarray}
By comparing \eqref{alo:lasso:low} and \eqref{alo:lasso:high} we can evaluate the error in our simple formula of the risk, presented in Theorem \ref{thm:lasso_approx}. The approach we proposed above can be extended to other non-differentiable regularizers too. Below we consider two other popular classes of estimators: (i) bridge and (ii) elastic net, and show how we can derive ALO formulas for each estimator.

\noindent{\em Bridge estimators:} Consider the class of bridge estimators
 \begin{eqnarray}\label{eq:ori_optalpha}
\bl \triangleq  \underset{\bm{\beta} \in \R^p}{\argmin}  \Bigl \{ \sum_{i=1}^n  \ell ( y_i|\bm{x_i}^\top \bm{\beta} ) + \lambda \|\bm{\beta}\|_q^q  \Bigr \}, 
\end{eqnarray}
where $q$ is a number between $(1,2)$. Note that these regularizers are only one time differentiable at zero. Hence, the Newton method introduced in Section \ref{ssec:newtonmethod} is not directly applicable. One can argue intuitively that since the regularizer is differentiable at zero, none of the regression coefficients will be zero. Hence, the regularizer is locally twice differentiable and formula \eqref{eq:aloformulafinal} works well. While this argument is often correct, we can again use the idea introduced above for LASSO to obtain the following $\alo$ formula that can be used even when an estimate of $0$ is observed:
\begin{equation}\label{eq:aloformulabridge}
 \frac{1}{n} \sum_{i=1}^n \phi \left (y_i,  \bm{x_i}^\top \bl +\left(\frac{\ld_i(\bl)}{\ldd_i(\bl)}\right) \left(  \frac{H_{ii}}{1 - H_{ii}} \right)     \right) ,
\end{equation}
where if we define $S \triangleq \{i: \beta_i \neq 0\}$ and for $u \neq 0$, $\rdd^q (u) \triangleq q(q-1) |u|^{q-2}$, then
\begin{eqnarray}\label{eq:bridgeHalo}
 \bm{H}=  \bm{X}_S \left ( \bm{X}_S^\top \diag[\bm{\ldd}(\bl)]\bm{X}_S + \lambda {\rm diag}[\bm{ \rdd}^q_{S}(\bl) ]      \right)^{-1} \bm{X}_S^\top \diag[\bm{\ldd} (\bl)].
\end{eqnarray}
This formula is derived in Section \ref{ssec:bridgederivation}.

\noindent {\em Elastic-net:} Finally, we consider the  following elastic-net estimator
 \begin{eqnarray}\label{eq:ori_optalpha}
\bm{\hat{\beta}} \triangleq  \underset{\bm{\beta} \in \R^p}{\argmin}  \Bigl \{ \sum_{i=1}^n  \ell ( y_i|\bm{x_i}^\top \bm{\beta} ) +\lambda_1 \|\bm{\beta}\|_2^2+ \lambda_2 \|\bm{\beta}\|_1  \Bigr \}. 
\end{eqnarray}
Again by smoothing the $\ell_1$-regularizer (similar to what we did for LASSO) we obtain the following ${\rm ALO}$ formula for the out-of-sample predictor error:
\[
\frac{1}{n} \sum_{i=1}^n \phi \left (y_i,  \bm{x_i}^\top \bm{\hat{\beta}} +\left(\frac{\ld_i(\bm{\hat{\beta}})}{\ldd_i(\bm{\hat{\beta}})}\right) \left(  \frac{H_{ii}}{1 - H_{ii}} \right)     \right),
\]
where $S= \{ i : {\hat{\beta}}_i \neq 0\}$, and  
\begin{equation}
 \bm{H}=  \bm{X}_S \left ( \bm{X}_S^\top \diag[\bm{\ldd} (\bm{\hat{\beta}})]   \bm{X}_S + 2\lambda_1 \bm{I } \right)^{-1} \bm{X}_S^\top \diag[\bm{\ldd} (\bm{\hat{\beta}})]. \label{eq:elnet}
\end{equation}
We do not derive this formula, since it follows exactly the same lines as those of LASSO and bridge. Algorithm 2 summarizes all the calculations required for the calculation of $\alo$ for elastic-net.

\begin{algorithm}\label{alg:2}
\caption{Risk estimation with ALO for elastic-net regularizer}
\textbf{Input.} $(\bm{x_1}, y_1), (\bm{x_2}, y_2), \ldots, (\bm{x_n}, y_n)$.  \\
\textbf{Output.} $\extra$ estimate.
\begin{enumerate}
\item Calculate  $\bl =  \underset{\bm{\beta} \in \R^p}{\argmin}  \Bigl \{ \sum_{i=1}^n  \ell ( y_i|\bm{x_i}^\top \bm{\beta} ) + \lambda_1 \|\bm{\beta}\|_2^2+ \lambda_2 \|\bm{\beta}\|_1 \Bigr \}.$

\item Calculate $S= \{ i : {\hat{\beta}}_i \neq 0\}$.

\item Obtain  $\bm{H}= \bm{X}_S \left ( \bm{X}_S^\top \diag[\bm{\ldd} (\bm{\hat{\beta}})]   \bm{X}_S + 2\lambda_1 \bm{I } \right)^{-1} \bm{X}_S^\top \diag[\bm{\ldd} (\bm{\hat{\beta}})]$, where $\bm{X}_S$ only includes the columns of $\bm{X}$ that are in $S$.  

\item The estimate of $\extra$ is given by $\frac{1}{n} \sum_{i=1}^n \phi \left (y_i,  \bm{x_i}^\top \bl +\left(\frac{\ld_i(\bl)}{\ldd_i(\bl)}\right) \left(  \frac{H_{ii}}{1 - H_{ii}} \right)     \right)$.

\end{enumerate}
\end{algorithm}
 
\section{Computational complexity and memory requirements of ALO}\label{sec:computations}
 Counting the number of floating point operations algorithms require is a standard approach for comparing their computational complexities. In this section, we calculate and compare the number of operations required by $\alo$ and $\lo$. We first start with Algorithm 1 and then discuss Algorithm 2.

\subsubsection*{Algorithm 1} Before we start the calculations, we should warn the reader that in many cases the specific structure of the loss and/or the regularizer enables more efficient implementation of the formulas. However, here we consider the worst case scenario. Furthermore, the calculations below are concerned with the implementation of ALO and LO on a single computer, and we have not explored their parallel or distributed implementations.  

The first step of Algorithm 1 requires solving an optimization problem. Several different methods exist for solving this optimization problem. Here, we discuss the interior point method and the accelerated gradient descent algorithm. Suppose that our goal is to reach accuracy $\epsilon$. Then, interior point method requires $O(\log(1/\epsilon))$ iterations to reach this accuracy, while accelerated gradient descent requires $O(\frac{1}{\sqrt{\epsilon}})$ iterations \cite{nesterov2013introductory}. Furthermore, each iteration of the accelerated gradient descent requires $O(np)$ operations, while each iteration of the interior point method requires $O(p^3)$ operations. 

Regarding the memory usage of these two algorithms, note that in the accelerated gradient descent algorithm the memory is mainly used for storing matrix $\bm{X}$. Hence, the amount of memory that is required by this algorithm is $O(np)$. On the other hand, interior point method uses $O(p^3)$ of memory. 

The second step of Algorithm 1 is to calculate the matrix $\bm{H}$. This requires inverting the matrix $\left (\lambda \diag[\bm{\rdd}(\bl)] + \bm{X}^\top \diag[\bm{\ldd}(\bl)] \bm{X} \right)^{-1}$. In general, this inversion requires $O(p^3)$ (e.g. by using Cholesky factorization). However, if $n$ is much smaller than $p$, then one can use a better trick for performing the matrix inversion; suppose that both $\ell$ and $r$ are strongly convex at $\bl$ and define $\bm{\Gamma} \triangleq (\diag[\bm{\ldd}(\bl)])^{\frac{1}{2}}$, and $\bm{\Lambda} \triangleq \lambda {\rm diag}[\bm{ \rdd}(\bl)].$ Then, from the matrix inversion lemma we have
\begin{equation}\label{eq:inversiontrick1}
\bm{X}(\bm{X}^\top \bm{\Gamma}^2 \bm{X} + \bm{\Lambda})^{-1} \bm{X}^\top = \bm{X} \bm{\Lambda}^{-1}\bm{X}^\top -  \bm{X} \bm{\Lambda}^{-1} \bm{X}^\top \bm{\Gamma} (\bm{I} +  \bm{\Gamma} \bm{X}\bm{\Lambda}^{-1}  \bm{X}^\top \bm{\Gamma}  )^{-1} \bm{\Gamma} \bm{X}\bm{\Lambda}^{-1}  \bm{X}^\top. 
\end{equation}
The inversion $(\bm{I} +  \bm{\Gamma} \bm{X}\bm{\Lambda}^{-1}  \bm{X}^\top \bm{\Gamma}  )^{-1}$ requires $O(n^3)$ operations and $O(np)$ of memory (the main memory usage is for storing $\bm{X}$). Also, the other matrix-matrix multiplications require $O(n^2 p+n^3)$ operations. Hence, overall if we use the matrix inversion lemma, then the calculation of $\bm{H}$ requires $O(n^3+ n^2 p)$ operations. In summary,  the calculation of $\bm{H}$ requires $O(\min (p^3+n p^2, n^3+ n^2p))$. Also, the amount of memory that is required by the algorithm is $O(np)$. 
 The last step of $\alo$, i.e. Step 3 in Algorithm 1, requires only $O(np)$ operations. Hence, the calculations of $\alo$  in Algorithm 1 requires 
 \begin{enumerate}
 \item Through interior point method: $O(\min (p^3 \log (1/\epsilon)+p^3+n p^2, p^3 \log (1/\epsilon) + n^3+ n^2p))$
 \item Through accelerated gradient descent: $O(\min (np \frac{1}{\sqrt{\epsilon}} + p^3+n p^2, np \frac{1}{\sqrt{\epsilon}} + n^3+ n^2p))$
 \end{enumerate}
 Similarly, the calculation of the $\lo$ requires solving $n$ optimization problem of the form \eqref{eq:optimization_lo}. Hence, the number of floating point operations that are required for $\lo$ are:
  \begin{enumerate}
 \item Through interior point method: $O(np^3 \log (1/\epsilon))$.
 \item Through accelerated gradient descent: $O(n^2 p \frac{1}{\sqrt{\epsilon}})$.
 \end{enumerate}
  
  \subsubsection*{Algorithm 2}
  Note that in Algorithm 2, we have used the specific form of the regularizer and simplified the form of $\bm{H}$. Hence, this allows for faster calculation of $\bm{H}$ and equivalently faster calculation of the $\alo$ estimate. Again the first step of calculating $\alo$ is to solve the optimization problem. Solving this optimization problem by the interior point method or accelerated proximal gradient descent requires $O(p^3 \log (1/\epsilon))$ and $O(np\frac{1}{ \sqrt{\epsilon}})$ floating point operations respectively. 
  The next step is to calculate $\bm{H}$. If $\bl$ is $s$-sparse, i.e., has only $s$ non-zero coefficients, then the calculation of $\bm{H}$ requires $O(s^3 +ns^2)$ floating point operations. Also, the amount of memory required for this inversion is $O(s^2)$. Finally, the last step requires $O(np)$ operations. Hence, calculating an $\alo$ estimate of the risk requires:
  \begin{enumerate}
 \item Through interior point method: $O(p^3 \log (1/\epsilon) +s^3 +ns^2 + np)$.
 \item Through accelerated proximal gradient descent: $O( np \frac{1}{\sqrt{\epsilon}} +s^3 +ns^2 + np)$
 \end{enumerate}

The calculations of $\lo$ in the worst case is similar to what we had in the previous section:\footnote{It is known that after a finite number of iterations the estimates of proximal gradient descent becomes sparse, and hence the iterations require less operations. Hence, in practice the sparsity can reduce the computational complexity of calculating $\lo$ even though this gain is not captured in the worst case analysis of this section. }
  \begin{enumerate}
 \item Through interior point method: $O(np^3 \log (1/\epsilon))$.
 \item Through accelerated proximal gradient descent: $O(n^2 p \frac{1}{\sqrt{\epsilon}})$.
 \end{enumerate}
 
In this section, we used the number of floating point operations to compare the computational complexity of $\alo$ and $\lo$ . However, since this approach is based on the worst case scenarios and is not capable of capturing the constants, it is less accurate than comparing the timing of algorithms through simulations. Hence, Section \ref{sec:numsim}  compares the performance of $\alo$ and $\lo$ through simulations.   
 
 \subsection*{Memory usage}

First, we discuss Algorithm 1. We only consider the accelerated gradient descent algorithm. As discussed above, the amount of memory that is required for Step 1 of ALO is  $O(np)$ (the main memory usage is for storing matrix $X$). For the second step,   
 direct inversion of $\left (\lambda \diag[\bm{\rdd}(\bl)] + \bm{X}^\top \diag[\bm{\ldd}(\bl)] \bm{X} \right)^{-1}$ requires $O(p^2)$ of memory. However,  by using the formula derived in \eqref{eq:inversiontrick1} the memory usage reduces to $O(n^2)$ (for inverting $(\bm{I} +  \bm{\Gamma} \bm{X}\bm{\Lambda}^{-1}  \bm{X}^\top \bm{\Gamma}  )^{-1}$). Hence, the total amount of memory required for the second step of Algorithm 1 is $O(\min(np + n^2, np+p^2))$: $np$ for storing $\bm{X}$ and $n^2$ or $p^2$ for calculating  $\left (\lambda \diag[\bm{\rdd}(\bl)] + \bm{X}^\top \diag[\bm{\ldd}(\bl)] \bm{X} \right)^{-1}$. The last step of ALO requires negligible amount of memory.  Hence, the total amount of memory $\alo$ requires especially when $n < p$, is $O(np+n^2)$, which is the same as $O(np)$. Note that the amount of memory required by LO is also $O(np)$, since it requires to store $\bm{X}$. 
 
 The situation is even more favorable for ALO in Algorithm 2; all the memory requirements are the same as before, except that the amount of memory that is required for the calculation and storing of $\left ( \bm{X}_S^\top \diag[\bm{\ldd} (\bm{\hat{\beta}})]   \bm{X}_S + 2\lambda_1 \bm{I } \right)^{-1}$ is $O(s^2)$.


\section{Theoretical Results in High Dimensions} \label{sec:res}

\subsection{Assumptions} \label{sec:ass}
In this section, we introduce  assumptions later used in our theoretical results. The assumptions and  theoretical results that follow are presented for finite sample sizes. However, the final conclusions of this paper are focused on the high-dimensional asymptotic setting in which $n, p \rightarrow \infty$ and $n/p \rightarrow \delta_o$, where $\delta_o$ is a finite number bounded away from zero. Hence, if we write a constant as $c(n)$, it may be the case that the constant depends on both $n$ and $p$, but since $p \sim n/\delta_o$, we drop the dependance on $p$. We use this simplification for the sake of brevity and clarity of presentation. Since our major theorem involves finite sample sizes it is straightforward to go beyond this high-dimensional asymptotic setting and obtain more general results useful for other asymptotic settings.

\begin{assumption} \label{ass:0}
 The rows of $\bm{X} \in \R^{n \times p}$ are independent zero mean Gaussian vectors with covariance $\bm{\Sigma}$. Let $\rho_{\max}$  denote the largest eigenvalue of $\bm{\Sigma}$.
\end{assumption}

As we mentioned earlier, in our asymptotic setting, we assume that $n/p \rightarrow \delta_o$ for some $\delta_o$ bounded away from zero. Furthermore, we assume that the rows of $\bm{X}$ are scaled in a way that $\rho_{max} = \Theta(\frac{1}{n})$  to ensure that  $\bm{x}_i^\top \bm{\beta} = O_p(1)$ and $\bm{\beta}^\top \bm{\Sigma} \bm{\beta}= O(1)$, assuming that each $\beta_i$ is $O(1)$. Under this scaling the signal-to-noise ratio in each observation remains fixed as $n,p$ grow.\footnote{Furthermore, under this scaling of the optimal value of $\lambda$ will be $O_p(1)$ \cite{mousavi2013asymptotic}. }  For more information on this asymptotic setting and scaling, the reader may refer to \cite{karoui2013asymptotic, donoho2016high, DMM11, bayati2012lasso, weng2016overcoming,DW18}. 

\begin{assumption} \label{ass:1}
There exist finite constants $c_1(n)$ and $c_2(n)$, and $q_n \rightarrow 0$ all  functions of $n$,  such that with probability at least $1-q_n$ for all $i=1,\ldots,n$
\begin{eqnarray}
c_1(n) &>&  \| \bm{\ld}(\bl)\|_{\infty}, \label{eq:c1def}
 \\
 c_2(n) &>&  \sup_{t \in [0,1]} \frac{  \|  \bm{\ldd}_{/ i}( (1-t) \bli + t \bl ) - \bm{\ldd}_{/ i}(\bl  ) \|_2}{ \|   \bli - \bl   \|_2}. \label{eq:c2def}
 \\
 c_2(n) &>&   \sup_{t \in [0,1]} \frac{  \|  \bm{\rdd}( (1-t) \bli + t \bl ) - \bm{\rdd}(\bl  ) \|_2}{ \|   \bli - \bl   \|_2}.\label{eq:c3def}
 \end{eqnarray}
\end{assumption}

In what follows, for various regularizers and regression methods, by explicitly quantifying constants $c_1(n)$ and $c_2(n)$, we discuss conditions \eqref{eq:c1def}, \eqref{eq:c2def}, and \eqref{eq:c3def} in Assumption \ref{ass:1}. We consider the ridge regularizer in Lemma \ref{lem:ridgereg} and the smoothed-$\ell_1$ (and elastic-net) regularizer in Lemma \ref{lem:smoothLASSO}. Concerning various regression methods, we consider logistic (Lemma \ref{lem:logistic1}), robust regression (Lemma \ref{lem:psudoHuber}), least-squares  (Lemmas \ref{lem:ridge1} and \ref{cor:ridge}), and Poisson (Lemmas \ref{lem:poisson1} and \ref{lem:poisson2}) regression. The results below show that under mild assumptions, for the cases mentioned above, $c_1(n)$ and $c_2(n)$ are polynomial functions of $\log n$, a result that plays a key role in our main theoretical result presented in Section \ref{sec:main}.

\begin{lemma}\label{lem:ridgereg}
For the ridge regularizer $r(z) = z^2$, we have
\begin{eqnarray*}
 \sup_{t \in [0,1]} \frac{  \|  \bm{\rdd}( (1-t) \bli + t \bl ) - \bm{\rdd}(\bl  ) \|_2}{ \|   \bli - \bl   \|_2}=0.
\end{eqnarray*}
\end{lemma}
Due to simplicity we skip the proof. As mentioned in Section \ref{ssec:non-diff}, a standard smooth approximation of the $\ell_1$-norm is given by
\begin{eqnarray*}
r^\alpha(z) = \sum_{i=1}^p \frac{1}{\alpha} \Big( \log(1+ e^{\alpha z}) + \log(1+ e^{-\alpha z}) \Big).
\end{eqnarray*} 
\begin{lemma}\label{lem:smoothLASSO}
For the smoothed-$\ell_1$ regularizer we have
\begin{eqnarray*}
 \sup_{t \in [0,1]} \frac{  \|  \bm{\rdd}( (1-t) \bli + t \bl ) - \bm{\rdd}(\bl  ) \|_2}{ \|   \bli - \bl   \|_2} \leq 4 \alpha^2.
\end{eqnarray*}
\end{lemma}
We present the proof of this result  in Section \ref{ssec:lem:smoothLASSO1}. Note that as a consequence of Lemma \ref{lem:smoothLASSO},  for the smoothed elastic-net regularizer, defined as $r(z) = \gamma z^2 + (1-\gamma) r^\alpha(z)$ for $\gamma \in [0,1]$, we have
\begin{eqnarray*}
 \sup_{t \in [0,1]} \frac{  \|  \bm{\rdd}( (1-t) \bli + t \bl ) - \bm{\rdd}(\bl  ) \|_2}{ \|   \bli - \bl   \|_2} \leq 4 (1-\gamma) \alpha^2.
\end{eqnarray*} 
\begin{lemma} \label{lem:logistic1}
In the generalized linear model family, for the negative logistic regression log-likelihood $\ell ( y|\bm{x}^\top \bm{\beta} )=-y\bm{x}^\top \bm{ \beta} +\log (1 + e^{\bm{x}^\top \bm{ \beta}})$, where  $y \in \{0,1\}$, we have
\begin{eqnarray*}
 \sup_{t \in [0,1]} \frac{  \|  \bm{\ldd}_{/ i}( (1-t) \bli + t \bl ) - \bm{\ldd}_{/ i}(\bl  ) \|_2}{ \|   \bli - \bl   \|_2} &\leq& \sqrt{\sigma_{\max} (\bm{X}^\top \bm{X})}, 
 \\
  \| \ld(\bm{\beta}) \|_{\infty}  &\leq& 1.
\end{eqnarray*}
\end{lemma}
We present the proof of this result in Section \ref{ssec:prooflogistic1}. Our next example is about a smooth approximation of the Huber loss used in robust estimation, known as the pseudo-Huber loss:
\[
f_H(z) = \gamma^2 \left( \sqrt{1+ (\frac{z}{\gamma})^2 }-1\right),
\]
where $\gamma>0$ is a fixed number. 
\begin{lemma}\label{lem:psudoHuber}
For the pseudo-Huber loss function $\ell ( y|\bm{x}^\top \bm{\beta} )= f_H(y-\bm{x}^\top \bm{ \beta})$, we have
\begin{eqnarray*}
 \sup_{t \in [0,1]} \frac{  \|  \bm{\ldd}_{/ i}( (1-t) \bli + t \bl ) - \bm{\ldd}_{/ i}(\bl  ) \|_2}{ \|   \bli - \bl   \|_2} &\leq& \frac{3}{\gamma} \sqrt{ \sigma_{\max} (\bm{X}^\top \bm{X})}, \\
  \| \ld(\bm{\beta}) \|_{\infty}  &\leq& \gamma.
\end{eqnarray*}
\end{lemma}
The proof of this result is presented in Section \ref{ssec:proofpseudoHuber}. 
\begin{lemma} \label{lem:X}
If Assumption \ref{ass:0} holds with $\rho_{\max} = c/n$, and $\delta_0=n/p$, then
\[
\Pr \left(\sigma_{max} (\bm{X}^\top \bm{X}) \geq c \Big(1+ 3 {\frac{1}{\sqrt{\delta_0}}}\Big)^2\right) \leq  {\rm e}^{-p}. 
\]
\end{lemma}
The proof of this Lemma presented in Section \ref{sec:proof}. Putting together Lemmas \ref{lem:ridgereg}, \ref{lem:smoothLASSO}, \ref{lem:logistic1}, \ref{lem:psudoHuber} and \ref{lem:X},  we conclude that  for ridge/smoothed-$\ell_1$ regularized robust/logistic regression  we have $c_1(n)=O(1)$ and $c_2(n)=O(1)$.

\begin{lemma}\label{lem:ridge1}
For the loss function $\ell ( y|\bm{x}^\top \bm{\beta} )= \frac{1}{2}(y-\bm{x}^\top \bm{\beta})^2$, we have
\begin{eqnarray*}
 \sup_{t \in [0,1]} \frac{  \|  \bm{\ldd}_{/ i}( (1-t) \bli + t \bl ) - \bm{\ldd}_{/ i}(\bl  ) \|_2}{ \|   \bli - \bl   \|_2} &=& 0, \\
  \| \ld(\bl) \|_{\infty}  &\leq& \|\bm{y} - \bm{X} \bl \|_\infty.
\end{eqnarray*}
\end{lemma}
We skip the proof of this lemma because it is straightforward. 
\begin{lemma}\label{cor:ridge}
Assume $\bm{y}  \sim N( \bm{X} \bm{\beta}^* , \sigma_{\e}^2 \bm{I})$, and $\ell ( y|\bm{x}^\top \bm{\beta} ) = \frac{1}{2}(y-\bm{x}^\top \bm{\beta})^2$. Let  Assumption \ref{ass:0} hold with $\rho_{\max} = c/n$.  Finally, let $n/p = \delta_0$ and $ \frac{1}{n}\|\bm{\beta}^*\|_2^2 = \tilde{c}$. If $r(\beta) = \gamma \beta^2 + (1-\gamma) r^\alpha (\beta)$, and $ 0<\gamma <1$, then 
\[
\Pr \left(\|\bm{y} - \bm{X} \bl \|_\infty >  \tilde \zeta \sqrt{\log n} \right) \leq \frac{10}{n} + 2n {\rm e}^{-n+1} + n {\rm e}^{-p},
\]
where $\tilde{\zeta}$ is a constant that only depends on $\sigma_\epsilon, \alpha, c, \tilde{c}, \lambda, \delta_0$ and $\gamma$ (and is free of $n$ and $p$).
\end{lemma}
We present the proof of this result in Section \ref{ssec:proofcorridge}. Putting together Lemmas \ref{lem:ridgereg}, \ref{lem:smoothLASSO}, \ref{lem:ridge1}, and \ref{cor:ridge}, we conclude that for smoothed elastic-net regularized least squares regression  we have $c_1(n) = O(\sqrt{\log n})$ and $c_2(n)=O(1)$.

\begin{lemma}\label{lem:poisson1}
 In the generalized linear model family, for the negative Poisson regression log-likelihood $\ell ( y|\bm{x}^\top \bm{\beta} )= -f(\bm{x}^\top \bm{ \beta}) + y \log f(\bm{x}^\top \bm{ \beta}) - \log y!$  with the conditional mean $\E[y| \bm{x}, \bm{\beta}] = f(\bm{x}^\top \bm{\beta})$ where $f(z)=\log(1+e^z)$ (known as a soft-rectifying nonlinearity\footnote{The ``soft-rectifying'' nonlinearity  $f(z)=\log(1+e^z)$  behaves linearly for large $z$, and decays exponentially on its left tail. Owing to the convexity and log-concavity of this nonlinearity the log-likelihood is concave \cite{PAN04c}, leading to a convex estimation problem. Since the actual nonlinearity of neural systems is often sub-exponential, the ``soft-rectifying'' nonlinearity is popular in analyzing neural data (see \cite{P07,PWHP14,WP17,ZP18} and references therein).
}), we have
\begin{eqnarray*}
 \sup_{t \in [0,1]} \frac{  \|  \bm{\ldd}_{/ i}( (1-t) \bli + t \bl ) - \bm{\ldd}_{/ i}(\bl  ) \|_2}{ \|   \bli - \bl   \|_2} &\leq& (1 + 6 \| \bm{y} \|_{\infty}) \sqrt{\sigma_{max}({\bm{X}^\top \bm{X}})} \\
  \| \ld(\bm{\beta}) \|_{\infty}  &\leq& 1 +\| \bm{y} \|_{\infty}.
\end{eqnarray*}
\end{lemma}
We present the proof of this result in Section \ref{ssec:proofpoisson1}. 
\begin{lemma}\label{lem:poisson2} 
Assume that $y_i \sim {\rm Poisson} \left(f(\bm{x}_i^\top \bm{\beta^*})  \right)$  where $f(z)=\log(1+e^z)$. Let  Assumption \ref{ass:0} hold with $\rho_{\max} = c/n$. Finally, let $n/p = \delta_0$ and $\bm{{\beta^*}^\top \Sigma \beta^*} = \tilde{c}$. Then, for large enough $n$, we have
\begin{eqnarray*}
\Pr \left( (1 + 6 \| \bm{y} \|_{\infty}) \sqrt{\sigma_{max} ({\bm{X}^\top \bm{X}})} \geq {\zeta_1} \log^{3/2} n  \right) &\leq& n^{1- \log \log n}+ \frac{2}{n} + {\rm e}^{- n \log(\frac{1}{\mathbb{P} (Z \leq 1)})} + {\rm e}^{-p} \\
\Pr \left (\| \bm{y} \|_{\infty} \geq 6 \sqrt{\tilde{c}} \log^{3/2} n \right) &\leq&  n^{1- \log \log n}+ \frac{2}{n} + {\rm e}^{- n \log(\frac{1}{\mathbb{P} (Z \leq 1)})}
\end{eqnarray*}
where $Z \sim N(0, \tilde{c})$ {and  $\zeta_1$ is a constant that only depends on $c, \tilde{c}$, and $\delta_0$ (and is free of $n$ and $p$).} 
\end{lemma}
The proof of this result is presented in Section \ref{ssec:proofpoisson2}. Putting together Lemmas \ref{lem:ridgereg}, \ref{lem:smoothLASSO}, \ref{lem:poisson1}, and \ref{lem:poisson2}, we conclude that for ridge/smoothed elastic-net regularized Poisson regression we have $c_2(n) = O(\log^{3/2}(n))$ and $c_1(n) = O(\log^{3/2}(n))$.  

In summary, in the high-dimensional asymptotic setting, for all the examples we have discussed so far, $c_1(n) = O(\log^{3/2}(n))$ and $c_2(n) = O(\log^{3/2}(n))$. Hence, in the results that we will see in the next section we assume that both $c_1(n)$ and $c_2(n)$ are polynomial functions of $\log (n)$. Finally, we assume that the curvatures of the optimization problems involved in \eqref{eq:bl} and \eqref{eq:bli} have a lower bound:

\begin{assumption} \label{ass:2} There exists a constant $\nu > 0$, and a sequence $\tilde{q}_n \rightarrow 0$ such that for all $i=1,\ldots,n$
\begin{eqnarray}\label{eq:minEigComb}
 \inf_{ t \in [0,1] }  \sigma_{\min} \left (  \lambda \diag[ \bm{\rdd}(t \bl+ (1-t) \bli)] + \XI^\top \diag[\bm{\ldd}_{/ i}(t \bl + (1-t) \bli)]  \XI  \right)  \geq \nu
\end{eqnarray}
with probability at least $1-\tilde{q}_n$. Here, $\sigma_{\min}(\bm{A})$ stands for the  smallest singular value of $\bm{A}$.
\end{assumption}
Assumption \ref{ass:2} means that optimization problems \eqref{eq:bl} and \eqref{eq:bli} are strongly convex, and strong convexity is a standard assumption made in the analysis of high dimensional problems, eg.  \cite{V08,N12}. Moreover, if $r(\beta) = \gamma \beta^2 + (1-\gamma) r^\alpha (\beta)$, and $ 0<\gamma <1$, then $\nu=2\gamma$.

Before we mention our main result, we should also mention that Assumptions \ref{ass:2}, \ref{ass:0}, and \ref{ass:1} can be weakened at the expense of making our final result look more complicated. For instance, the Gaussianity of the rows of $\bm{X}$ can be replaced with the subgaussianity assumption with minor changes in our final result. We expect our results (or slightly weaker ones) to hold even when the rows of $\bm{X}$ have heavier tails. However, for the sake of brevity we do not study such matrices in the current paper. Furthermore, the smoothness of the second derivatives of the loss function and the regularizer that is assumed in \eqref{eq:c2def} and \eqref{eq:c3def} can be weakened at the expense of slower convergence in Theorem \ref{theo:main}. We will clarify this point in a footnote after \eqref{eq:referencefromproof} in the proof. 

\subsection{Main theoretical result} \label{sec:main}

Now based on these results we bound the difference $|\alo-\lo|$. The proof is given in Section \ref{sec:theo-1}.
\begin{theorem} \label{theo:main} 
Let $n/p = \delta_0$ and Assumption \ref{ass:0} hold with $\rho_{\max} = c/p$. Moreover, suppose that Assumptions \ref{ass:1}, and \ref{ass:2}  are satisfied, and that $n$ is large enough such that $q_n + \tilde{q}_n < 0.5$. Then with probability at least $1-4n e^{-p} - \frac{8n}{p^3} - \frac{8n}{(n-1)^3}- q_n- \tilde{q}_n$ the following bound is valid:
\begin{eqnarray}
\max_{1 \leq i \leq n} \left| \bm{x_i}^\top \bli - \bm{x_i}^\top \bl - \left(\frac{\ld_i(\bl)}{\ldd_i(\bl)}\right) \left(  \frac{H_{ii}}{1 - H_{ii}} \right)     \right|
 \leq    \frac{C_o }{\sqrt{p}}, \label{eq:loo_approximation}
\end{eqnarray} 
where 
\begin{eqnarray}
C_o &\triangleq& ( \frac{ 72 c^{3/2}}{\nu^3} ) \left(1 + \sqrt{\delta_0} (\sqrt{\delta_0} + 3 )^2 \frac{ c\log n}{\log p}\right)
 \left( c_1^2(n) c_2(n)  +  c_1^3(n) c_2^2(n)  \frac{5   (c^{1/2}+ c^{3/2}(\sqrt{\delta_0} + 3  )^2  )  }{\nu^2 }       \right). \label{eq:Co}
\end{eqnarray}
\end{theorem}
Recall that  in Section \ref{sec:ass} we proved that for many regularized regression problems in the generalized linear family both $c_1(n) = O({\rm PolyLog} (n))$ and $c_2(n)= O({\rm PolyLog} (n))$, where the notation ${\rm PolyLog} (n)$ denotes a polynomial in $\log (n)$. These examples included ridge and smoothed-$\ell_1$ (and elastic-net) regularizers and logistic, robust, least-squares, and Poisson regression. More specifically, the maximum degree we observed for the logarithm was $3/2$, which happened for the Poisson regression. Furthermore, as mentioned in the last section, in the high-dimensional asymptotic setting in which $n, p \rightarrow \infty$ and $n/p \rightarrow \delta_o$, where $\delta_o$ is a finite number bounded away from zero, to keep the signal-to-noise ratio fixed in each observation (as $p$ and $n$ grow), we considered the scaling that $n \rho_{\max} = O(1)$. Combining these, it is straightforward to see that $C_0(n) = O(c_1^3(n)c_2^2(n))=O({\rm PolyLog} (n))$. Therefore, the difference $\max_{1 \leq i \leq n} \left| \bm{x_i}^\top \bli - \bm{x_i}^\top \bl - \left(\frac{\ld_i(\bl)}{\ldd_i(\bl)}\right) \left(  \frac{H_{ii}}{1 - H_{ii}} \right)     \right|= O_p(\frac{{\rm PolyLog} (n)}{\sqrt{n}})$. Theorem \ref{theo:main} proves the accuracy of the approximation of the leave-one-out estimate of the regression coefficients. As a simple corollary of this result we can also prove the accuracy of our approximation of $\lo$.




\if0\longer
{
Note that in the $p$ fixed, $n \rightarrow \infty$ regime, Theorem \ref{theo:main} fails to yield  $\left|\alo - \lo \right| = o_p(1)$. This is just an artifact of our proof. In a full version of the paper that will be posted on arXiv (it is uploaded with the current paper on JASA's website) we present  a simple analysis to show that under mild regularity conditions, the error between $\alo$ and $\lo$ is $o_p(1/n)$ when $n\rightarrow \infty$ and $p$ is fixed. Due to space limitation this argument is removed from the current paper. 
  } \fi

\begin{corollary}\label{cor:aloVSlo}
Suppose that all the assumptions used in Theorem \ref{theo:main} hold. Moreover, suppose that $$\max_{i=1,2, \ldots, n} \sup_{|b_i| <  \frac{C_o}{\sqrt{p} }}  \left| \pd \left (y_i,  \bm{x_i}^\top \bli + b_i   \right) \right| \leq c_3(n)$$ with probability $r_n$. Then, with probability  at least $1-4n e^{-p} - \frac{8n}{p^3} - \frac{8n}{(n-1)^3}- q_n- \tilde{q}_n-r_n$
\begin{eqnarray}\label{eq:finalcompth}
\left | \alo - \lo \right | \leq \frac{c_3(n)C_o }{\sqrt{p}},
\end{eqnarray}
where $C_o$ is the constant defined in Theorem \ref{theo:main}. 
\end{corollary}
The proof of this result can be found in Section \ref{proof:corofMain}. 
As we discussed before, in all the examples we have seen so far $\frac{C_o}{\sqrt{p}}$ is $O(\frac{{\rm PolyLog} (n)}{\sqrt{n}})$. Hence, to obtain the convergence rate of $\alo$ to $\lo$ we only need to find an upper bound for $c_3(n)$. Note that usually the loss function $\ell$ that is used in the optimization problem is also used as the function $\phi$ to measure the prediction error. Hence, assuming $\phi(\cdot, \cdot) = \ell (\cdot, \cdot)$, we study the value of $c_3(n)$ for the examples we discussed in Section \ref{sec:ass}.

\begin{enumerate}
\item If $\phi$ is the loss function of Lemma \ref{lem:logistic1}, then $\left| \pd \left (y_i,  \bm{x_i}^\top \beta  \right) \right| \leq 2$, leading to $c_3(n)=2$.

\item If $\phi$ is the loss function of Lemma  \ref{lem:poisson1}, then $\left| \pd \left (y_i,  \bm{x_i}^\top \beta  \right) \right| \leq 1+ \|\bm{y}\|_\infty$. Furthermore, we proved in Lemma \ref{lem:poisson2}, that under the data generating mechanism described there, with high probability $\|\bm{y}\|_\infty < 6\sqrt{\tilde{c} \log^3 (n)}$, leading to $c_3(n)=1+6\sqrt{\tilde{c} \log^3 (n)}$. 

\item For the pseudo-Huber loss described in Lemma \ref{lem:psudoHuber}, we have $\left| \pd \left (y_i,  \bm{x_i}^\top \beta  \right) \right| \leq \gamma$, leading to $c_3(n)=\gamma$.

\item For the square loss  $\left| \pd \left (y_i,  \bm{x_i}^\top \bli + b_i   \right) \right| \leq |y_i - \bm{x_i}^\top \bli| + |b_i| \leq |y_i - \bm{x_i}^\top \bli|+ \frac{C_o}{\sqrt{p}}$. Hence, in order to obtain a proper upper bound we require more information about the estimate $\bli$. Suppose that our estimates are obtained from the optimization problem we discussed in Lemma \ref{cor:ridge}. Then, based on  \eqref{eq:upperyinfty} and \eqref{eq:xblithirdterm} in the proof of Lemma \ref{cor:ridge} in Appendix \ref{ssec:proofcorridge}
\[
\max_{i}  |y_i - \bm{x_i}^\top \bli| \leq \max_{i}  |y_i| + \max_{i} |\bm{x_i}^\top \bli| \leq 2 \sqrt{(c\tilde{c} + \sigma_\epsilon^2) \log n} + \sqrt{\frac{10c (c\tilde{c} + \sigma_\epsilon^2) \log n}{\lambda \gamma}}.
\]
with  probability at most $\frac{4}{n} + n{\rm e}^{-n+1}$, leading to
$
c_3(n) = 2 \sqrt{(c\tilde{c} + \sigma_\epsilon^2) \log n} + \sqrt{\frac{20c (c\tilde{c} + \sigma_\epsilon^2) \log n}{\lambda \gamma}}+ \frac{C_o}{\sqrt{p}}. 
$
\end{enumerate}
 
 In summary, in the high-dimensional asymptotic setting, for regularized regression methods introduced in Section \ref{sec:ass}, such as least-squares, logistic, Poisson and robust regression, with $r(\beta) = \gamma \beta^2 + (1-\gamma) r^\alpha (\beta)$, and $ 0<\gamma <1$, and assuming $\phi(\cdot, \cdot) = \ell (\cdot, \cdot)$, we have  $c_3(n) = O(\rm{PolyLog} (n))$, leading to  $\left | \alo - \lo \right | = O_p(\frac{\rm{PolyLog} (n)}{\sqrt{n}})$. In short, these examples show that $\alo$ offers a consistent estimate of $\lo$. 

\if1\longer
{
  Finally, note that in the $p$ fixed, $n \rightarrow \infty$ regime, Theorem \ref{theo:main} fails to yield  $\left|\alo - \lo \right| = o_p(1)$. This is just an artifact of our proof. In Theorem \ref{thm:alo_lo_largen}, presented in Section \ref{sec:largen_asymptot} we prove that under mild regularity conditions, error between $\alo$ and $\lo$ is $o_p(1/n)$ when $n\rightarrow \infty$ and $p$ is fixed. For the sake of brevity details are presented in Section \ref{sec:largen_asymptot}. 
} \fi

\section{Numerical Experiments}\label{sec:numsim}
\subsection{Summary}

To illustrate the accuracy and computational efficiency of $\alo$ we apply it to synthetic and real data. We generate synthetic data, and compare $\alo$ and $\lo$ for elastic-net linear regression in Section \ref{sec:num:linear}, LASSO logistic regression in Section \ref{sec:num:logistic}, and elastic-net Poisson regression in Section \ref{sec:num:poisson}. We should emphasize that our simulations are performed on a single personal computer, and we have not considered the impact of parallelization on the performance of $\alo$ and $\lo$. In other words, the simulation results reported for $\lo$ are based on its sequential implementation on a single personal computer. 
For real data, we apply LASSO, elastic-net and ridge logistic regression to sonar returns from two undersea targets in Section \ref{sec:sonar}, and we apply LASSO Poisson regression to real recordings from spatially sensitive neurons in Section \ref{sec:num:grid}. Our synthetic and real data examples  cover various data shapes  where $n > p$, $n=p$ and $n<p$. 


Figures \ref{fig:gaussian:elnet:cor}, \ref{fig:logistic:lasso:cor}, \ref{fig:poisson:elnet:cor}, \ref{fig:sonar}, and the middle-lower panel of Figure \ref{fig:T9C3}  reveal that $\alo$ offers a reasonably accurate estimate of $\lo$ for a large range of $\lambda$. These figures show that $\alo$  deteriorates for extremely small values of $\lambda$, specially when $p > n$. This is not a serious issue because the $\lambda$s minimizing  $\lo$ and $\alo$ tend to be far from those small values. 

The real data example in Section \ref{sec:sonar}, illustrating $\alo$ and $\lo$ in Figure \ref{fig:sonar}, is about  classifying sonar returns from two undersea targets using penalized logistic regression. The neuroscience example in Section \ref{sec:num:grid}  is about estimating an inhomogeneous spatial point process using an over-complete basis  from a sparsely sampled two-dimensional space. Given the spatial nature of the problem, the design matrix $\bm{X}$ is very sparse, which fails to satisfy the dense Gaussian design assumption we made in Theorem \ref{theo:main}. Nevertheless, the lower middle panel of Figure \ref{fig:T9C3} illustrates  the excellent performance of $\alo$ in approximating $\lo$ in an example where $p=10000$ and $n=3133$.

Figure  \ref{fig:time} compares the computational complexity(time) of a single fit, $\alo$ and $\lo$, as we increase $p$ while we keep the ratio $\frac{n}{p}$ fixed. We consider various data shapes, models, and penalties. Figure \ref{fig:linear_time}  shows time versus $p$ for elastic-net  linear regression when $\frac{n}{p}=5$. Figure \ref{fig:logistic_time}  shows time versus $p$ for LASSO logistic regression  when $\frac{n}{p}=1$. Figure \ref{fig:poisson_time}  shows time versus $p$ for elastic-net Poisson regression when $\frac{n}{p}=\frac{1}{10}$. Finally, the middle-lower panel of Figure \ref{fig:T9C3} shows that for the neuroscience example   $\alo$ takes 7 seconds in comparison to the 60428 seconds required by $\lo$.  All these numerical experiments illustrate the significant computational saving offered by $\alo$. As it pertains to the reported run times, all fittings in this paper were performed using a 3.1 GHz Intel Core i7 MacBook Pro with 16 GB of memory. All the codes for the figures presented in this paper are  available here  \texttt{https://github.com/RahnamaRad/ALO}.




\subsection{Simulations}\label{sec:simulation}
In all the examples in this section (\ref{sec:num:linear}, \ref{sec:num:logistic}, \ref{sec:num:poisson} and \ref{sec:time}), we let the true unknown parameter vector $\bm{\beta}^* \in \R^p$  to have $k =n/10$ non-zero coefficients. The $k$ non-zero coefficients are randomly selected, and their values are independently drawn from a zero mean unit variance Laplace distribution. The rows $\bm{x_1}^\top,\cdots,\bm{x_n}^\top $ of the design matrix $\bm{X}$ are independently drawn from $\N(0,\bm{\Sigma})$. We consider two correlation structures: 1) {\it Spiked}: $cor(X_{ij},X_{ij'})=0.5$, and 2) {\it Toeplitz}: $cor(X_{ij},X_{ij'})=0.9^{|j'-j|}$. $\bm{\Sigma}$ is scaled such that  the signal variance $\text{var}(\bm{x_i}^\top \bm{\beta}^*) = 1$ regardless of the problem dimension.  In this section, all the fittings and calculations of $\lo$ (and the  one standard error interval of $\lo$) were computed using the \texttt{glmnet} package in R \cite{FHT10}, and $\alo$ was computed using the \texttt{alocv} package in R \cite{ALOCV}.  

  
\subsubsection{Linear regression with elastic-net penalty}\label{sec:num:linear}
We set  $\ell ( y|\bm{x}^\top \bm{\beta} )=\frac{1}{2}(y-\bm{x}^\top \bm{\beta})^2$, $r(\bm{\beta})=\frac{ (1-\alpha)}{2} \| \bm{\beta} \|_2^2+ \alpha \|  \bm{\beta} \|_1$ and $\alpha=0.5$. We let the rows $\bm{x_1}^\top,\cdots,\bm{x_n}^\top$ of $\bm{X}$ to have  a {\it Spiked} covariance and to generate data, we sample  $\bm{y} \sim \N(\bm{X} \bm{\beta}^*, \bm{I}  )$. Moreover, $ \phi(y,\bm{x}^\top \bm{\beta}) = \left ( y - \bm{x}^\top \bm{\beta} \right)^2$ so that $\alo = \frac{1}{n} \sum_{i=1}^n \left( \frac{y_i - \bm{x_i}^\top \bl}{1 - H_{ii}} \right)^2$ with $ \bm{H}=   \bm{X}_S \left (  \bm{X}_S^\top  \bm{X}_S + \lambda (1-\alpha) \bm{I }  \right)^{-1}  \bm{X_S}^\top$. For various data shapes, that is  $\frac{n}{p} \in \{5,\ 1,\ \frac{1}{10} \}$, we depict results in Figure \ref{fig:gaussian:elnet:cor} where reported times  refer to the required time to fit the model, compute $\alo$ and $\lo$ for a sequence of 30 logarithmically spaced tuning parameters from $1$ to $100$. 




\begin{figure} \vspace{-2.5cm}
    \centering
    \begin{subfigure}[b]{0.3\textwidth}
        \includegraphics[width=\textwidth]{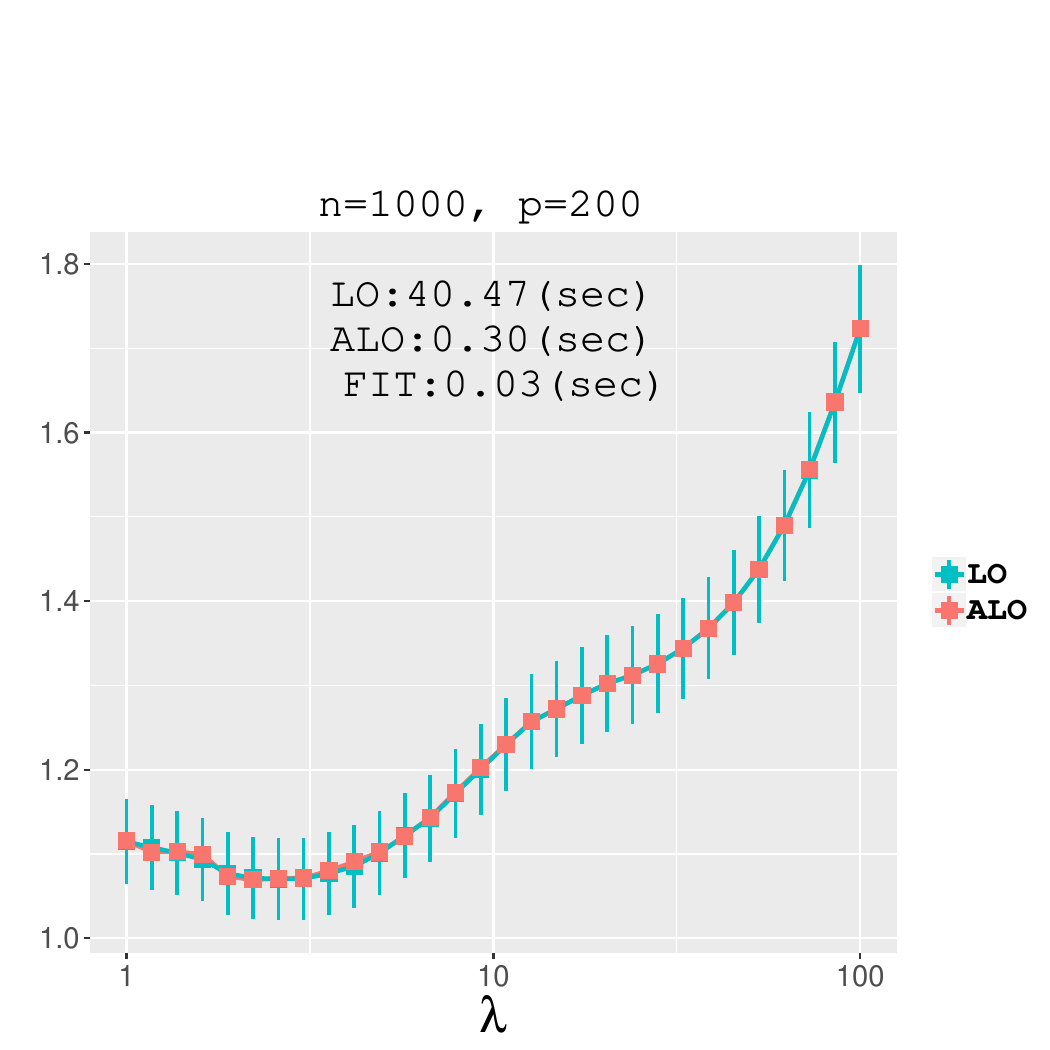}
        \caption{$n>p$}
    \end{subfigure}
    ~ 
    \begin{subfigure}[b]{0.3\textwidth}
        \includegraphics[width=\textwidth]{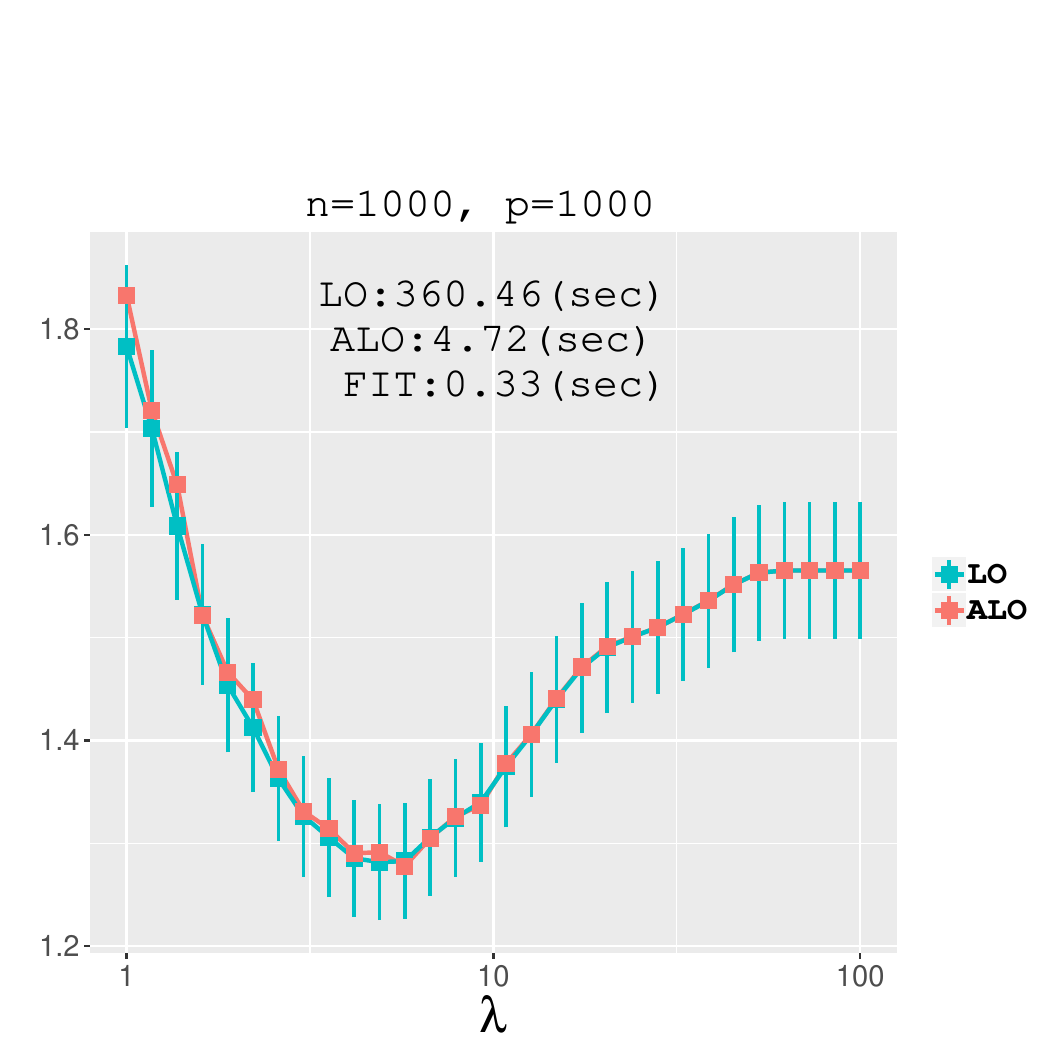}
        \caption{$n = p$}
    \end{subfigure}
    \begin{subfigure}[b]{0.3\textwidth}
        \vspace{1cm}
        \includegraphics[width=\textwidth]{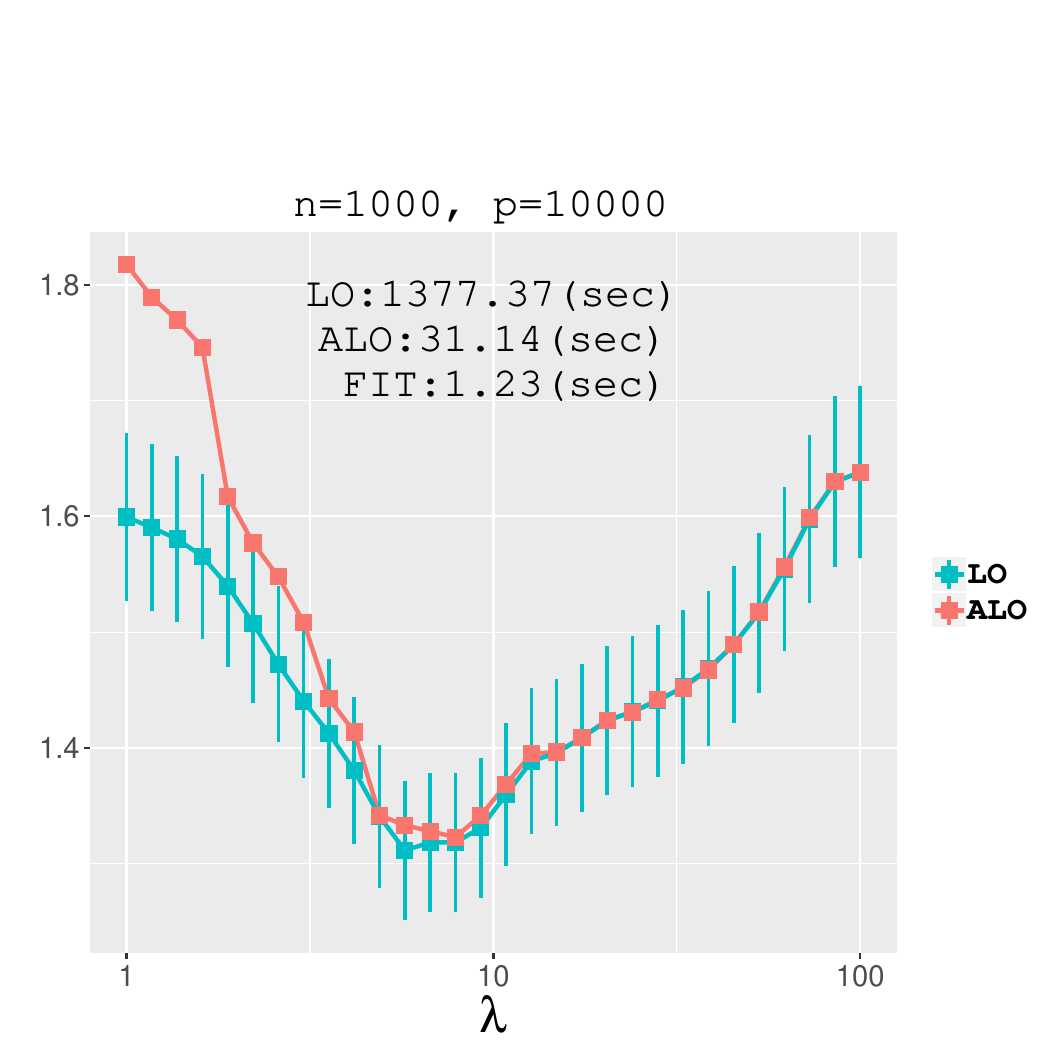}
        \caption{$n < p $}
    \end{subfigure}
    \caption{The $\alo$ and $\lo$ mean square error for elastic-net linear regression. The red error bars identify the one standard error interval of $\lo$.   
  }\label{fig:gaussian:elnet:cor}
\end{figure}

\subsubsection{Logistic regression with LASSO penalty}\label{sec:num:logistic}

\begin{figure} \vspace{-1.5cm}
    \centering
    \begin{subfigure}[b]{0.3\textwidth}
        \includegraphics[width=\textwidth]{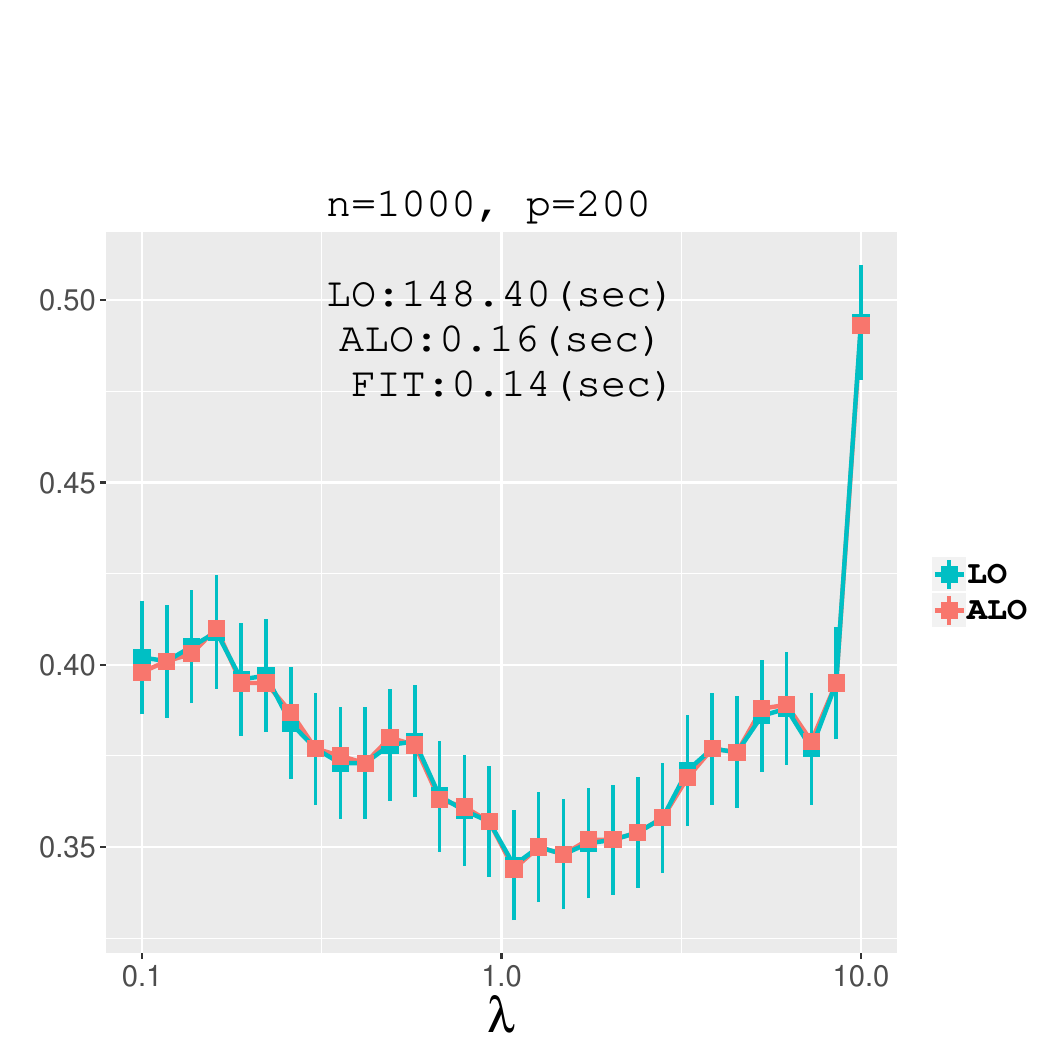}
        \caption{$n>p$}
    \end{subfigure}    
    \begin{subfigure}[b]{0.3\textwidth}
        \vspace{1cm}
        \includegraphics[width=\textwidth]{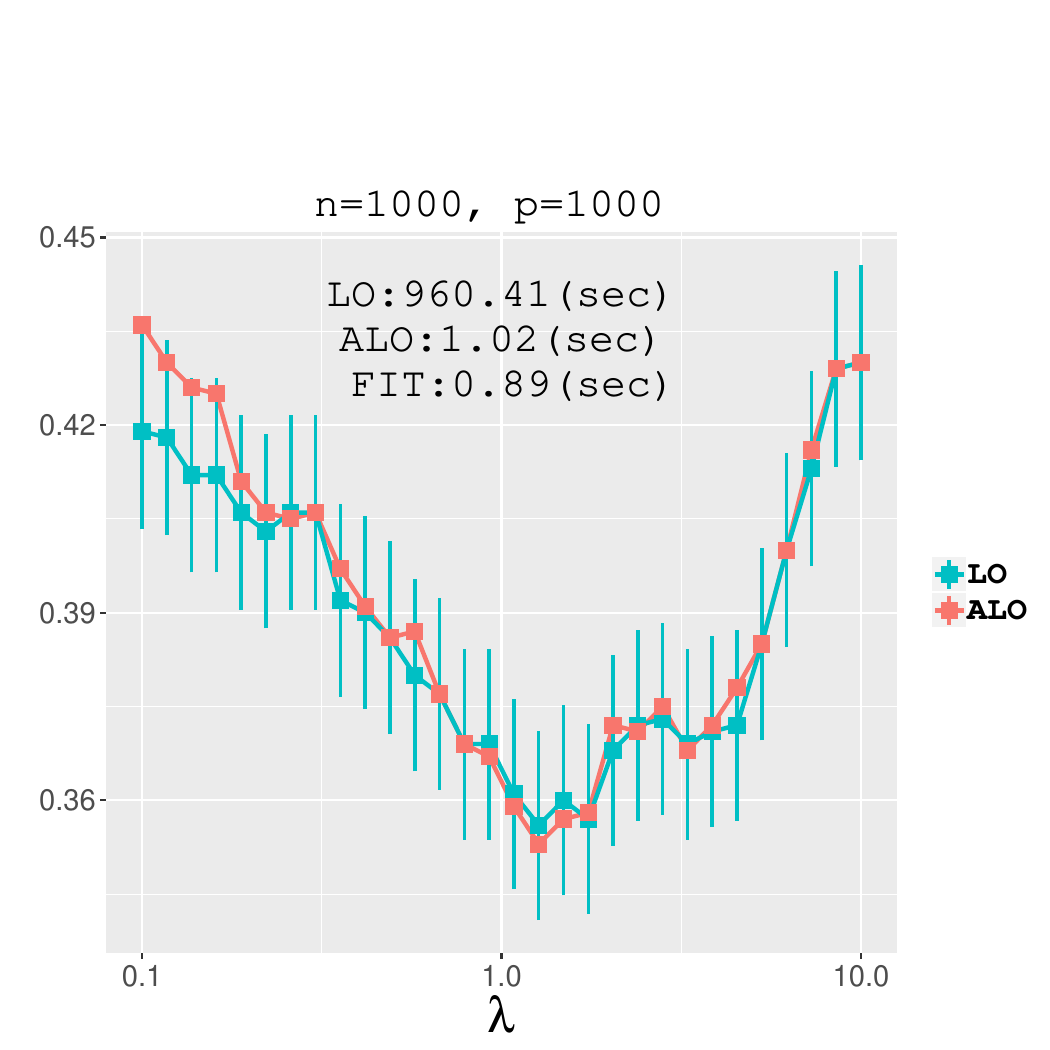}
        \caption{$n = p$}
    \end{subfigure}
        ~ 
    \begin{subfigure}[b]{0.3\textwidth}
        \includegraphics[width=\textwidth]{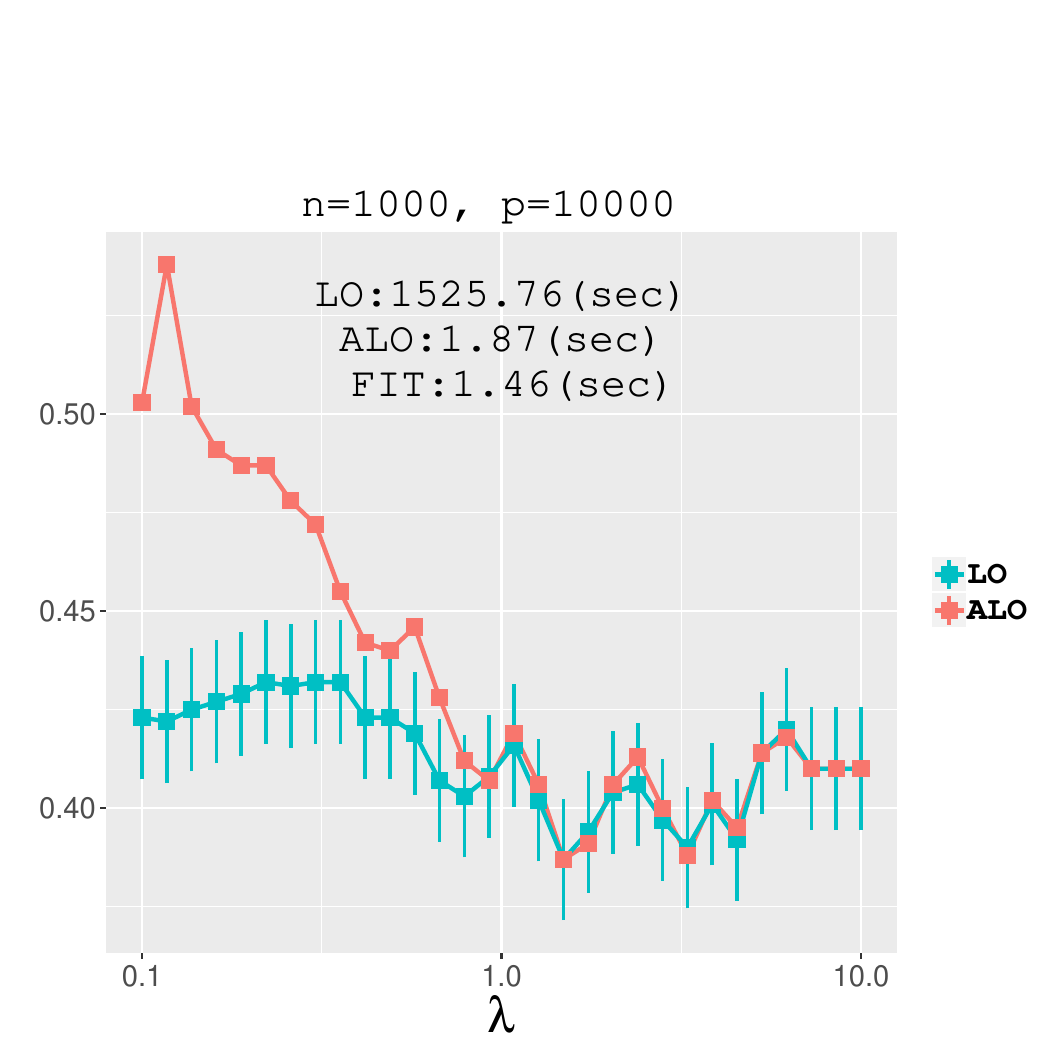}
        \caption{$ n < p$}
    \end{subfigure}
    \caption{The $\alo$ and $\lo$ misclassification errors (as a function of $\lambda$) for LASSO logistic regression. The red error bars identify the  one standard error interval of $\lo$. }\label{fig:logistic:lasso:cor}
\end{figure}


We set $\ell ( y|\bm{x}^\top \bm{\beta} )=-y\bm{x}^\top \bm{ \beta} +\log (1 + e^{\bm{x}^\top \bm{ \beta}})$ (the negative logistic log-likelihood) and $r(\bm{\beta})=\|  \bm{\beta} \|_1$. We let the rows $\bm{x_1}^\top,\cdots,\bm{x_n}^\top$ of $\bm{X}$ to have  a {\it Toeplitz} covariance and to generate data, we sample    $y_i \sim Binomial \left(\frac{e^{ \bm{x_i}^\top \bm{\beta}^*}}{1 + e^{\bm{x_i}^\top \bm{\beta}^*}}\right)$. We take the misclassification rate as our measure of error, and $1_{ \{ \bm{x}^\top \bm{\beta} > 0 \} }$ as prediction, where $1_{ \{ \cdot \}}$ is the indicator function, so that 
 \begin{eqnarray*}
 \alo = \frac{1}{n} \sum_{i=1}^n   \left |  y_i - 1_{ \{\bm{x_i}^\top \bm{\hat \beta} + \frac{\ld_i(\bm{\hat \beta})}{\ldd_i(\bm{\hat \beta})} \frac{H_{ii}}{1-H_{ii}}  > 0  \}} \right |
 \end{eqnarray*}  
 where $\bm{H}= \bm{X}_S \left ( \bm{X}_S^\top \diag[\bm{\ldd} (\bm{\hat{\beta}})]   \bm{X}_S  \right)^{-1} \bm{X}_S^\top \diag[\bm{\ldd} (\bm{\hat{\beta}})]
$, $\ld_i(\bl) =  \left(1 + e^{-\bm{x_i}^\top \bm{\hat \beta}}\right)^{-1} - y_i$ and $\ldd_i(\bl) = e^{\bm{x_i}^\top \bm{\hat \beta}}\left(1 + e^{\bm{x_i}^\top \bm{\hat \beta}} \right)^{-2}$. For various data shapes, that is  $\frac{n}{p} \in \{5,\ 1,\ \frac{1}{10} \}$, we depict results in Figure \ref{fig:logistic:lasso:cor}
where reported times  refer to the required time to fit the model, compute $\alo$ and $\lo$ for a sequence of 30 logarithmically spaced tuning parameters from $0.1$ to $10$.





\subsubsection{Poisson regression with elastic-net penalty}\label{sec:num:poisson}
We set $\ell ( y|\bm{x}^\top \bm{\beta} )=e^{y\bm{x}^\top \bm{ \beta}} - y \bm{x}^\top \bm{ \beta}$ (the negative Poisson log-likelihood), $r(\bm{\beta})=\frac{ (1-\alpha)}{2} \| \bm{\beta} \|_2^2+ \alpha \|  \bm{\beta} \|_1$ and $\alpha=0.5$. We let the rows $\bm{x_1}^\top,\cdots,\bm{x_n}^\top$ of $\bm{X}$ to have  a {\it Spiked} covariance and to generate data, we sample $y_i \sim Poisson \left(e^{ \bm{x_i}^\top \bm{\beta^*} }\right)$. We use the mean absolute error as our measure of error, and $e^{\bm{x}^\top \bm{\beta}}$ as prediction,  so that 
\begin{eqnarray*}
\alo &=& \frac{1}{n} \sum_{i=1}^n   \bigl |y_i - e^{\bm{x_i}^\top \bm{\hat \beta} + \frac{\ld_i(\bm{\hat \beta})}{\ldd_i(\bm{\hat \beta})} \frac{H_{ii}}{1-H_{ii}}} \bigr|
\end{eqnarray*}
  where $\bm{H}= \bm{X}_S \left ( \bm{X}_S^\top \diag[\bm{\ldd} (\bm{\hat{\beta}})]   \bm{X}_S + \lambda (1-\alpha) \bm{I } \right)^{-1} \bm{X}_S^\top \diag[\bm{\ldd} (\bm{\hat{\beta}})]$, $\ld_i(\bl) =  e^{\bm{x_i}^\top \bm{\hat \beta}} - y_i$, and $\ldd_i(\bl) = e^{\bm{x_i}^\top \bm{\hat \beta}}$.
For various data shapes, that is  $\frac{n}{p} \in \{5,\ 1,\ \frac{1}{10} \}$, we depict  results in Figure \ref{fig:poisson:elnet:cor} where reported times  refer to the required time to fit the model, compute $\alo$ and $\lo$ for a sequence of 30 logarithmically spaced tuning parameters from $1$ to $100$. 


%

\begin{figure} \vspace{-2cm}
    \centering
    \begin{subfigure}[b]{0.3\textwidth}
        \includegraphics[width=\textwidth]{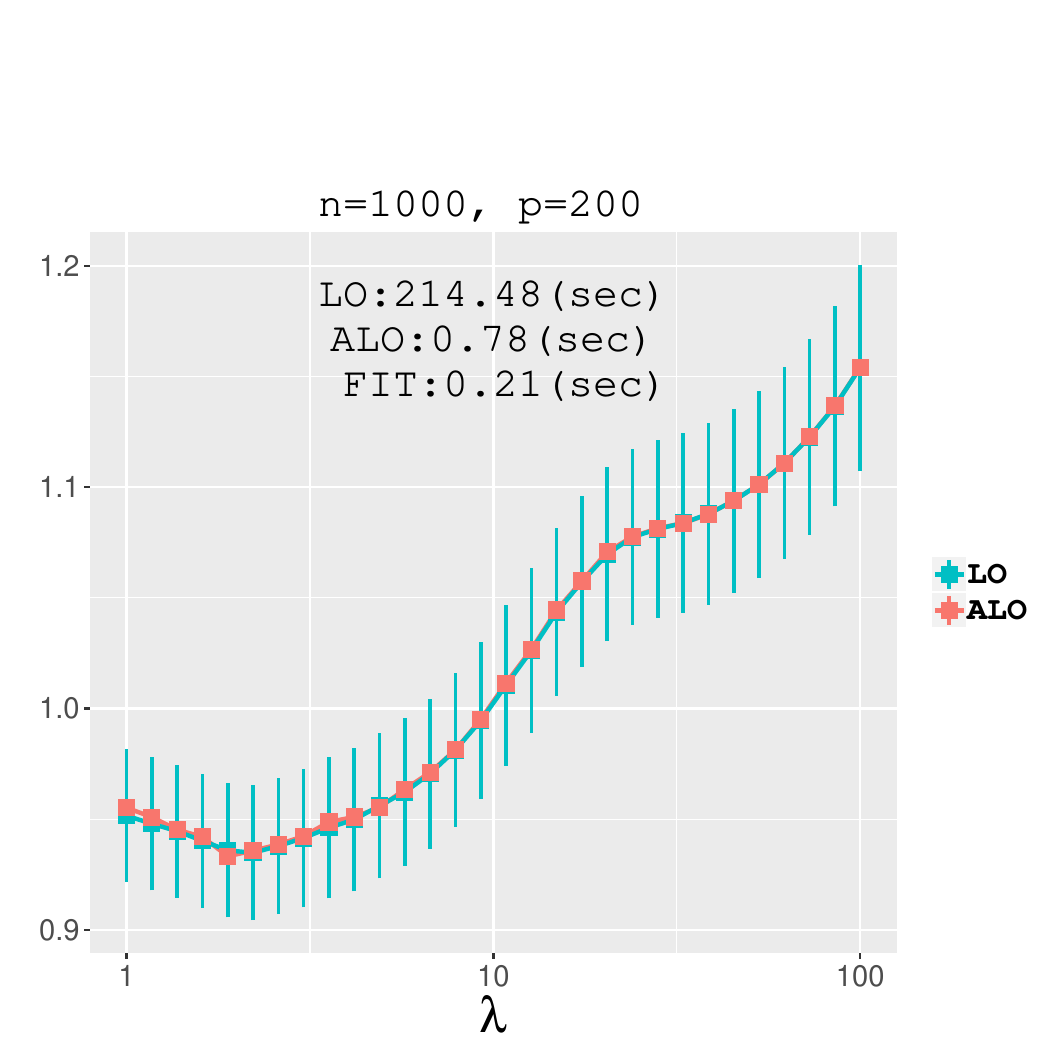}
        \caption{$n>p$}
    \end{subfigure}    
    \begin{subfigure}[b]{0.3\textwidth}
        \vspace{1cm}
        \includegraphics[width=\textwidth]{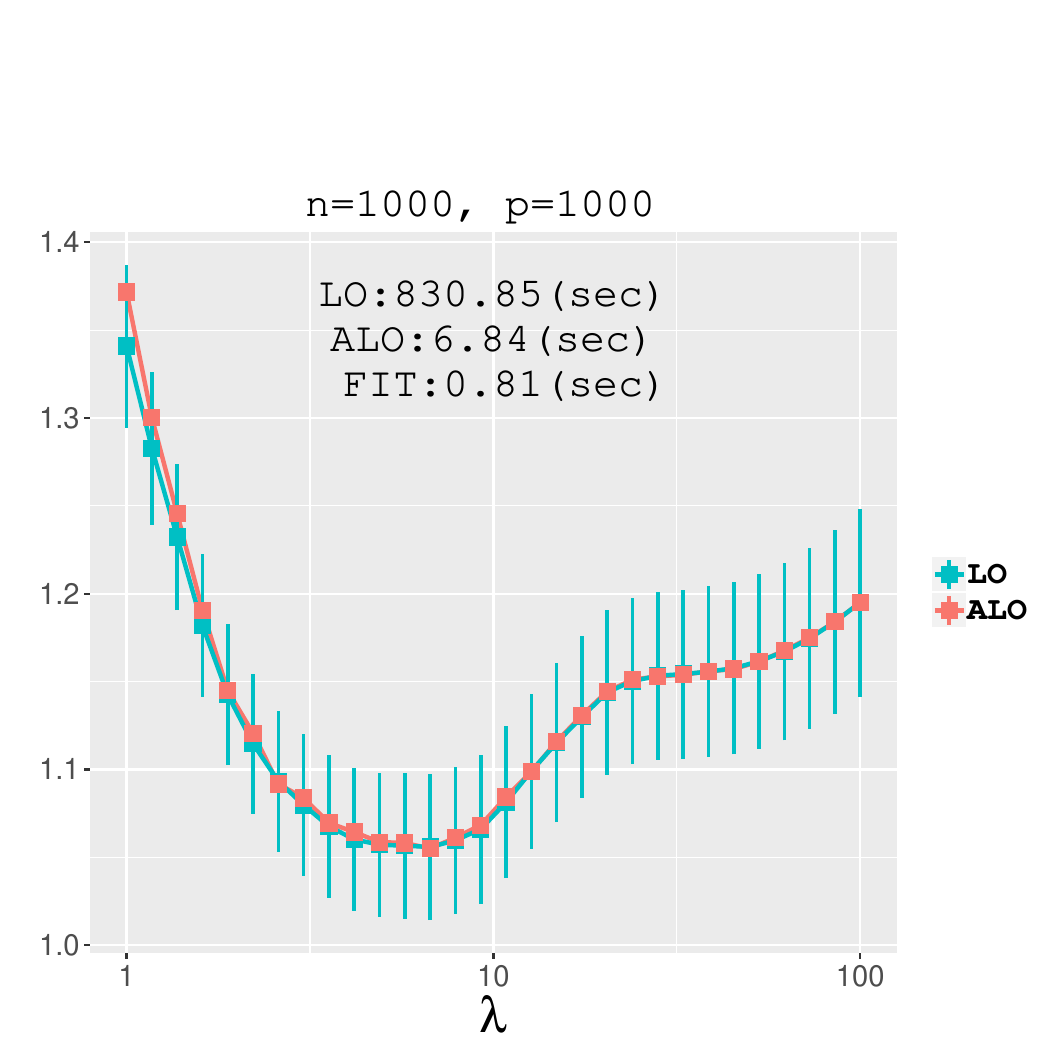}
        \caption{$n =p$}
    \end{subfigure}
        ~ 
    \begin{subfigure}[b]{0.3\textwidth}
        \includegraphics[width=\textwidth]{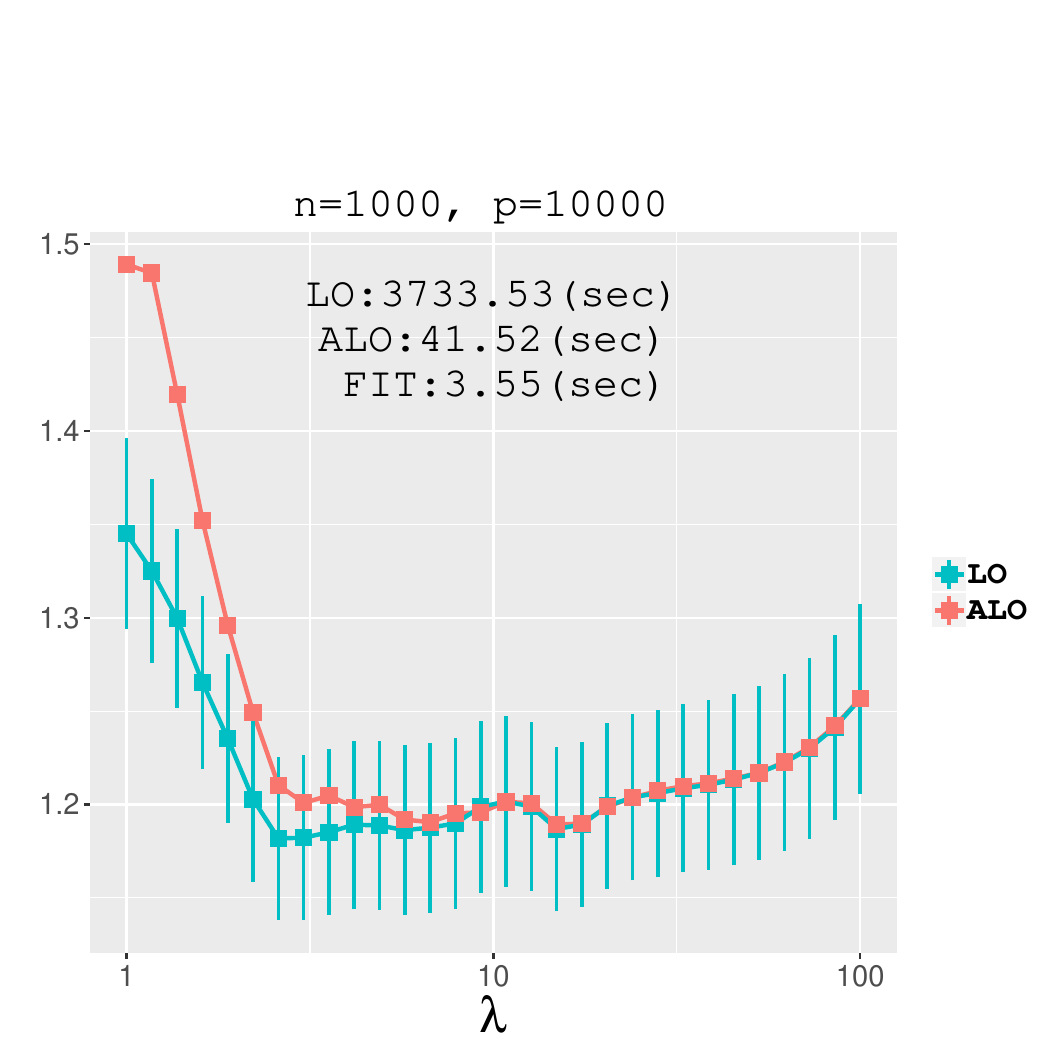}
        \caption{$ n < p$}
    \end{subfigure}
    \caption{The $\alo$ and $\lo$ mean absolute errors (as a function of $\lambda$) for elastic-net Poisson regression. The red error bars identify the  one standard error interval of $\lo$. }\label{fig:poisson:elnet:cor}
\end{figure}



\subsubsection{Timing simulations}\label{sec:time}
To compare the timing of $\alo$ with that of $\lo$, we consider the following scenarios:
\begin{itemize}
\item  Elastic-net linear regression, with rows of the design matrix having a  \textit{Spiked} covariance,  data generated as described in Sections \ref{sec:simulation} and \ref{sec:num:linear}, and considered  for a sequence of 10 logarithmically spaced tuning parameters from $1$ to $100$. We let $\frac{n}{p} = 5$.
\item LASSO logistic regression, with rows of the design matrix having a  \textit{Toeplitz} covariance,   data generated as described in Sections \ref{sec:simulation} and \ref{sec:num:logistic}, and considered  for a sequence of 10 logarithmically spaced tuning parameters from $0.1$ to $10$. We let $\frac{n}{p} = 1$.
\item Elastic-net Poisson regression, with rows of the design matrix having a  \textit{Spiked} covariance,  data generated as described in Sections \ref{sec:simulation} and \ref{sec:num:poisson}, and considered  for a sequence of 10 logarithmically spaced tuning parameters from $1$ to $100$. We let $\frac{n}{p} = \frac{1}{10}$.
\end{itemize}
The timings of a single fit,  ALO and  LO versus model complexity $p$ are
illustrated in Figure \ref{fig:time}. The reported timings are obtained by recording the time required to find a single fit  and LO using the \texttt{glmnet} package in R \cite{FHT10}, and to find ALO using  the \texttt{alocv} package in R \cite{ALOCV}, all along the tuning parameters above. This process is repeated 5 times to obtain the average timing.


\subsection{Real Data}

\subsubsection{Sonar data}\label{sec:sonar}

Here we use ridge, elastic-net and LASSO logistic regression to classify sonar returns collected from a metal cylinder and a
cylindrically shaped rock positioned on a sandy ocean floor. The data consists of a set of $n=208$ returns, 111 cylinder returns and 97 rock returns, and $p=60$ spectral features extracted from the returning signals \cite{GS88}. We use the misclassification rate as our measure of error. Numerical results comparing $\alo$ and $\lo$ for ridge, elastic-net and LASSO logistic regression are depicted in Figure \ref{fig:sonar}. The single fit and $\lo$ (and the  one standard error interval of $\lo$) were computed using the \texttt{glmnet} package in R \cite{FHT10}, and $\alo$ was computed using the \texttt{alocv} package in R \cite{ALOCV}. The values of the tuning parameters are a sequence of 30 logarithmically spaced tuning parameters between two value automatically selected  by the \texttt{glmnet} package.

\begin{figure} 
    \centering
    \begin{subfigure}[b]{0.3\textwidth}
        \includegraphics[width=\textwidth]{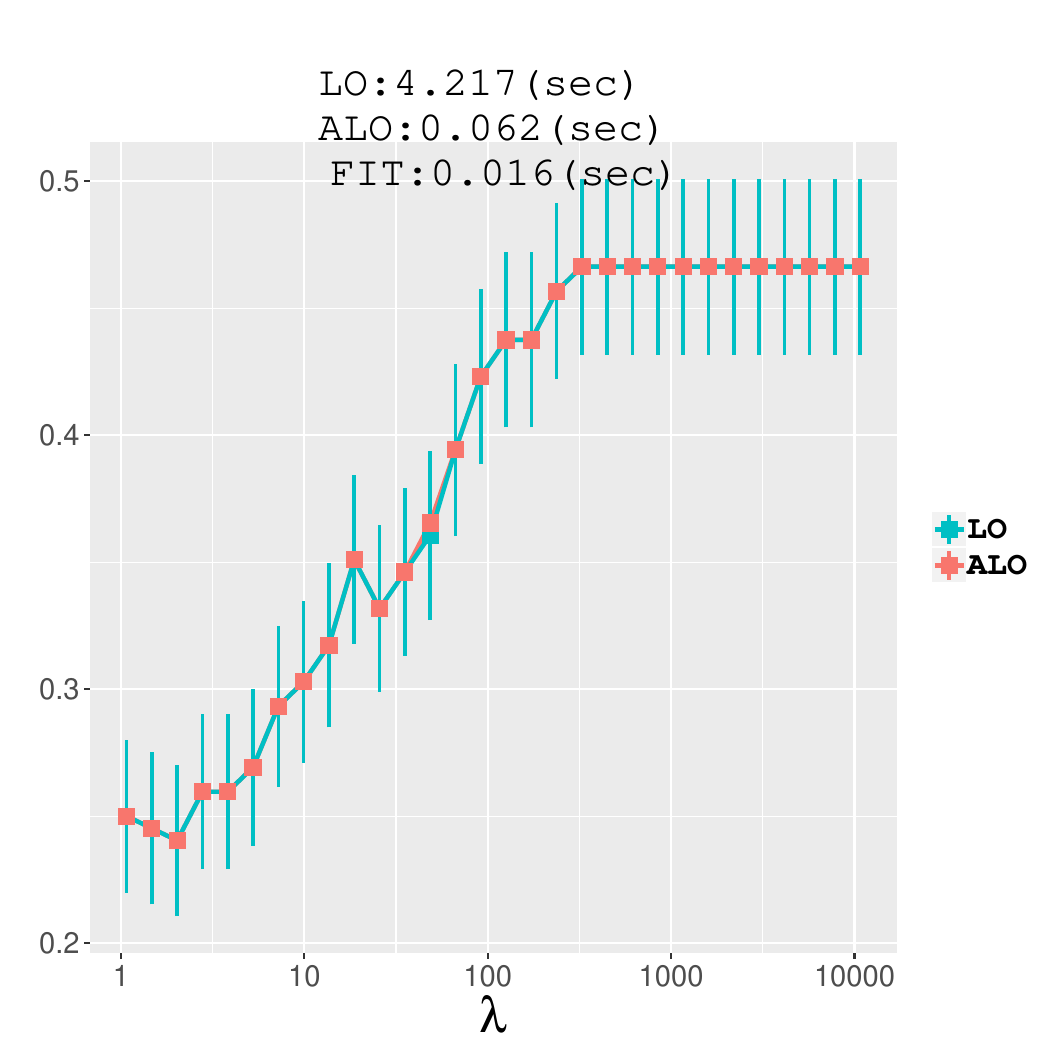}
        \caption{Ridge}
    \end{subfigure}    
    \begin{subfigure}[b]{0.3\textwidth}
        \vspace{1cm}
        \includegraphics[width=\textwidth]{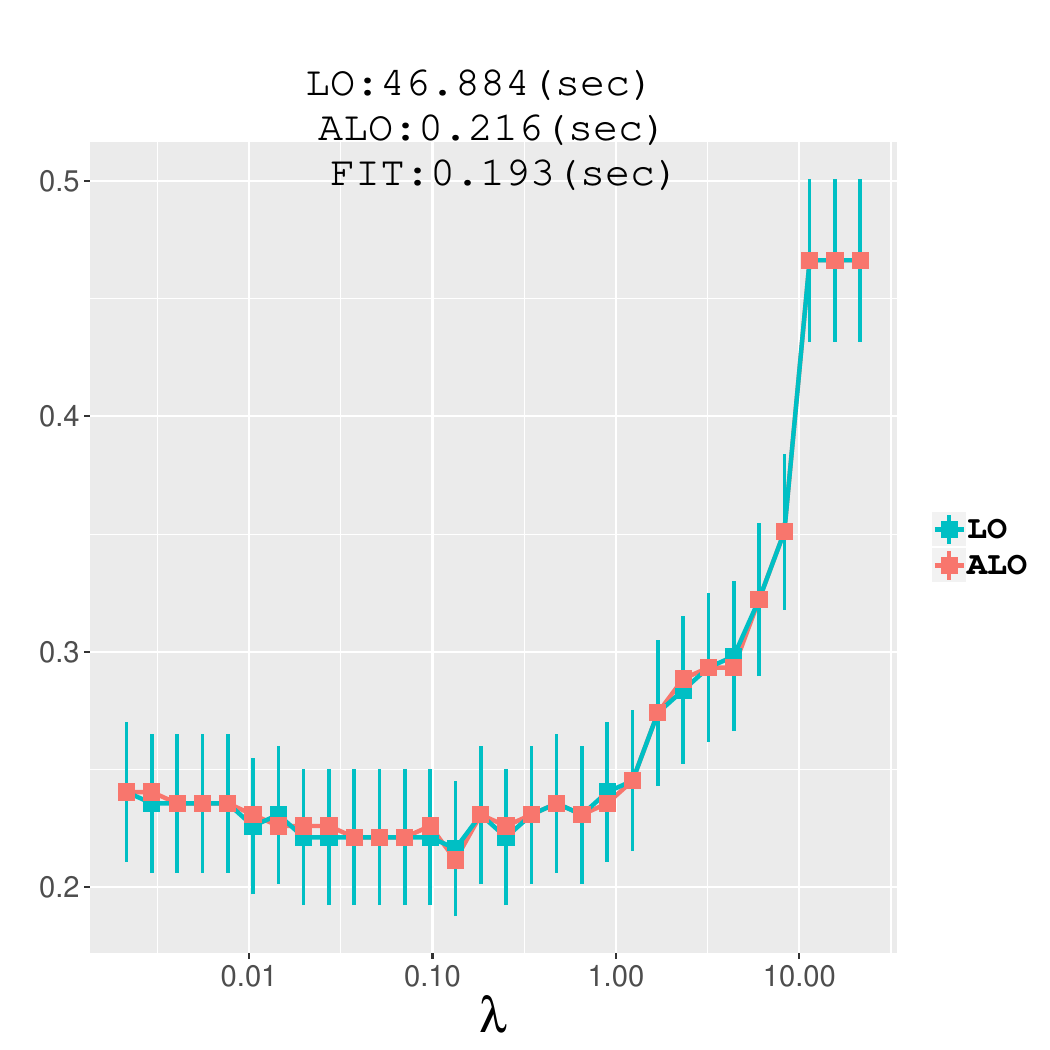}
        \caption{Elastic-net}
    \end{subfigure}
        ~ 
    \begin{subfigure}[b]{0.3\textwidth}
        \includegraphics[width=\textwidth]{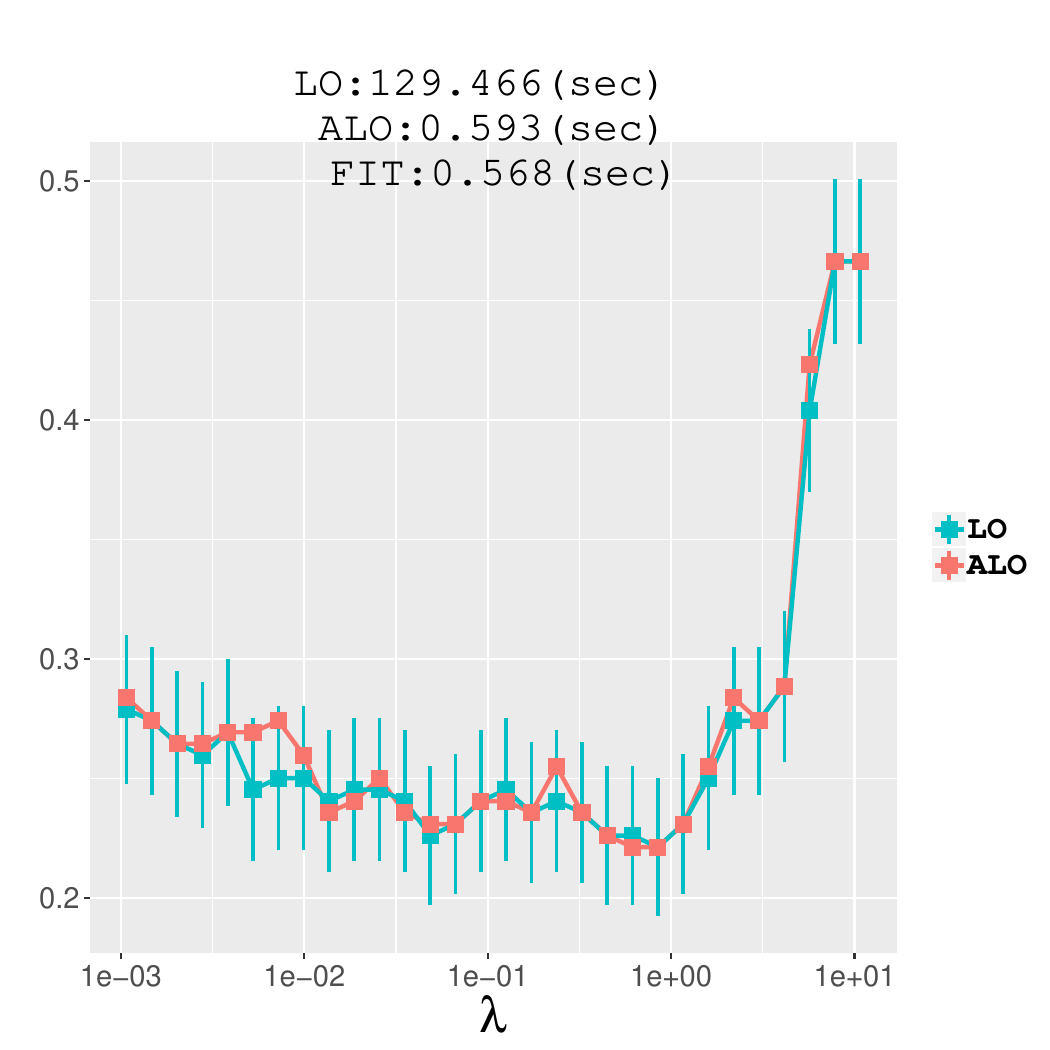}
        \caption{LASSO}
    \end{subfigure}
    \caption{The $\alo$ and $\lo$ deviances (as a function of $\lambda$) for penalized logistic regression applied to the sonar data (Section \ref{sec:sonar}) where $n=208$ and $p=60$. The red error bars identify the  one standard error interval of $\lo$. }\label{fig:sonar}
\end{figure}

\subsubsection{Spatial point process smoothing of grid cells: a neuroscience application}\label{sec:num:grid}

In this section, we compare $\alo$ with $\lo$ on a real dataset. This dataset includes electrical recordings  of single neurons in the entorhinal cortex, an area in the brain found to be particularly responsible for the navigation and perception of space in mammals \cite{MKM08}. The entorhinal cortex  is also one of the areas pathologically affected in the early stages of Alzheimer's disease, causing symptoms of spatial disorientation \cite{K14}. Moreover, the entorhinal cortex  provides input to another area, the Hippocampus, which is  involved in the cognition of space and the formation of episodic memory \cite{BM13}. 

Electrical recordings  of single neurons in the medial domain of the entorhinal cortex (MEC) of freely moving rodents have revealed spatially modulated neurons, called grid cells, firing action potentials only around the vertices of two dimensional hexagonal lattices covering the environment in which the animal navigates. The hexagonal firing pattern of a single grid cell is illustrated in the left panel of Figure \ref{fig:grid_example}. These grid cells can be categorized according to the orientation of their triangular grid, the wavelength (distance between the vertices ), and the phase (shift of the whole lattice). See the right panel of Figure \ref{fig:grid_example} for an illustration of the orientation and wavelength of a single grid cell. 
 

The data we analyze here consists of extra cellular recordings of several grid cells, and the simultaneously recorded location of the rat within a 300cm $\times$ 300cm box for roughly 20 minutes\footnote{The source of the data is \cite{SSSKMM12}. For a video of a single grid cell recorded in the MEC see the clip  {https://www.youtube.com/watch?v=i9GiLBXWAHI}.}. Since the number of spikes fired by a grid cell depends mainly on the location of the animal, regardless of the animal's speed and running direction \cite{HFMMM05}, it is reasonable to summarize this spatial dependency in terms of a rate map $\eta(\bm{r})$, where $\eta(\bm{r}) dt$ is the expected number of spikes emitted by the grid cell in a fixed time interval $dt$, given that the animal is located at position $\bm{r}$ during this time interval \cite{Kamiar08,PRHP14,DMR15}.  In other words, if the rat passes the same location again, we again expect the grid cell to fire at more or less the same rate\footnote{It is known that these rate maps can in some cases change with time but in most cases it is reasonable to assume them to be constant. Moreover, the two dimensional surface represented by  $\eta(\bm{r})$ is not the same for different grid cells.}, specifically according to a Poisson distribution with mean $\eta(\bm{r}) dt$. For each grid cell, the estimation of the rate map $\eta(\bm{r})$  is a first step toward  understanding the cortical circuitry underlying spatial cognition \cite{RRMM16}. Consequently, the estimation of firing fields without contamination from measurement noise or bias from overs-smoothing will help to clarify important questions about neuronal algorithms underlying navigation in real and mental spaces \cite{BM13}.

\begin{figure}
\vspace{-7cm}
\hspace{1.7cm}
        \includegraphics[scale=0.65]{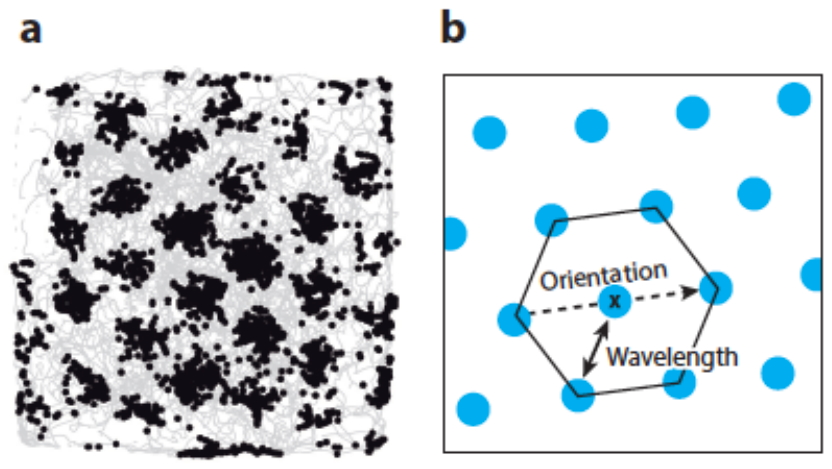}        \vspace{-6cm}
        \caption{Left: Spike locations (black) are superimposed on the animal's trajectory(grey). Firing fields are
areas covered by a cluster of action potentials. Right: The firing fields of a grid cell form a periodic triangular matrix tiling
the entire environment available to the animal. Figure is adapted from \cite{MMR14}.}\label{fig:grid_example}
\end{figure}

To be concrete, we discretize the two dimensional space into an $m \times m$ grid, and discretize time into bins with width $dt$. In this example, $dt$ is 0.4 seconds and $m$ is 50. The experiment is 1252.9 seconds long, and therefore we have $ \left \lceil{\frac{1252.9}{0.4}}\right \rceil  =3133$ time bins. In other words, $n=3133$. We use $y_i \in \{ 0,1,2,3,\cdots \}$ to denote the number of action potentials observed in  time interval $[(i-1)dt, i dt)$, where $i=1,\cdots,n$. Moreover, we use $\bm{r_i} \in \R^{m^2}$ to denote a vector composed of zeros except for a single +1 at the entry corresponding to the animal's location within the $m \times m$ grid during the time interval $[(i-1)dt, i dt)$. We assume a log-linear model $\log \eta(\bm{r}) = \bm{r}^\top \bm{z}$, relating the firing rate at location $\bm{r} \in \R^{m^2}$ to the latent vector $\bm{z}$ where the $m \times m$ latent spatial process responsible for the observed spiking activity is unraveled into $\bm{z} \in \R^{m^2}$. The firing rate can be written as $ \eta(\bm{r_i}) =\exp \left(\bm{ r_i}^\top \bm{z} \right)$.  Due to this notation, $\bm{r_i}^\top \bm{z}$ is the value of $\bm{z}$ at the animal's location during the time interval $[(i-1)dt, i dt)$. In this vein, the distribution of observed spiking activity can be written as
 \begin{eqnarray}
 p(y_i | \bm{r_i} ) &=& \frac{e^{-\eta(\bm{r_i})} \eta(\bm{r_i})^{y_i}}{y_i !}. \label{eq:poiss}
\end{eqnarray}
\begin{figure}\hspace{1cm}
        \includegraphics[scale=0.4]{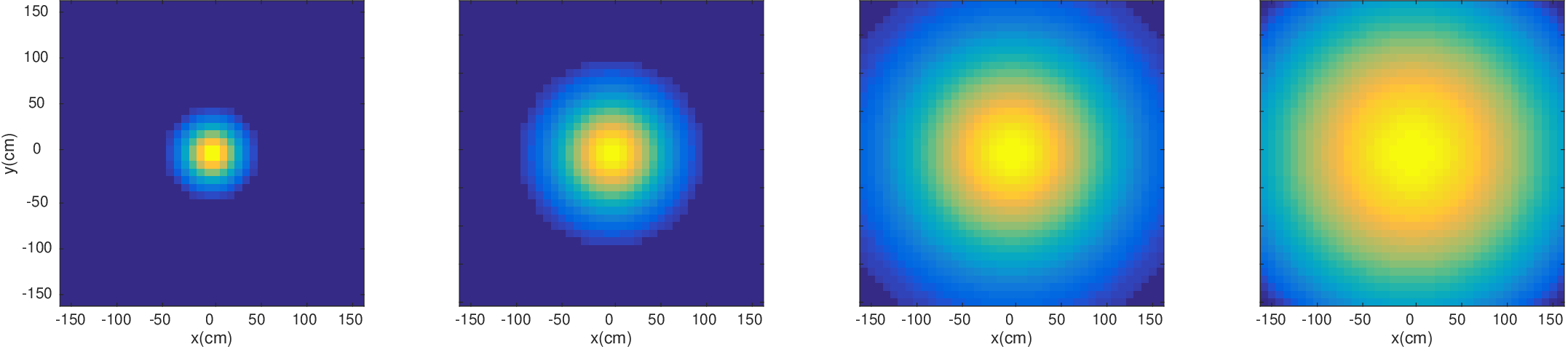}        
        \caption{The four truncated Gaussian bumps   }\label{fig:basis}
\end{figure}
As mentioned earlier, the main goal is to estimate the two dimensional rate map $\eta(\cdot)$, and a large body of work has addressed the problem of estimating a smooth rate map from neural data \cite{DGK01,GAO02,Kass05,Cunningham07,CZANNER05,CGRS09,PAFKRVVW10,Kamiar08,MGWKB11,PRHP14}. Here we employ an over-complete basis to account for the spatially localized sensitivity of grid cells. Since it is known that the rate map of any single grid cell consists of bumps of elevated firing rates,  located at various points in the two dimensional space, as illustrated in the left panel of Figure \ref{fig:grid_example}, it is reasonable to represent  $\bm{z}$ as a linear combination of $\{\bm{\psi_1},\ldots,\bm{\psi_p} \}$, an over-complete basis in $\R^p$\cite{BROW01,PRHP14,DMR15}. We compose the over-complete basis using truncated Gaussian bumps with various scales, distributed at all pixels. The four basic Gaussian bumps we use are depicted in Figure \ref{fig:basis}. Since we use four truncated Gaussian bumps for each pixel, in this example, we have a total of $p=4m^2=10000$ basis functions. We employ the  truncated Gaussian bumps $ e^{- \frac{1}{2\sigma^2} (u_x^2 + u_y^2)   }1_{ \left\{ \exp \left ( - \frac{1}{2\sigma^2} (u_x^2 + u_y^2)   \right) > 0.05 \right\} }$ where $u_x$ and $u_y$ are the horizontal and vertical coordinates. Define  $\bm{\Psi} \in \R^{m^2 \times p}$ as a matrix composed of columns $\{\bm{\psi_1},\ldots,\bm{\psi_p} \}$. Furthermore, define $ \bm{\tilde x_i} \in \R^p$ as  $ \bm{\tilde x_i} \triangleq \bm{\Psi}^\top \bm{r_i}$, and define $\bm{\tilde X} \in \R^{n \times p}$ as a matrix composed of rows $\{\bm{\tilde x_1}^\top, \ldots,\bm{\tilde x_n}^\top \}$. We normalize the columns of $\bm{\tilde X}$, calling the resulting matrix  $\bm{X}$. The columns of $\bm{X} \in \R^{n \times p}$ are unit normed. Formally, $\bm{X} = \bm{\tilde X \Gamma}^{-1}$ where $\bm{\Gamma} \in \R^{p \times p}$ is a diagonal matrix filled with the column-norms of $\bm{\tilde X}$. We use $\{\bm{x_1}^\top,\ldots, \bm{x_n}^\top \}$ to refer to the rows of $\bm{X}$, yielding $\eta(\bm{r_i}) = \exp\left( \bm{x_i}^\top \bm{\beta}\right)$. Note that due to the above mentioned rescaling, we have the following relationship between the latent map $\bm{z}$ and $\bm{\beta}$:
$\bm{z} = \bm{\Psi \Gamma \beta}$.
Sparsity of $\bm{\beta}$ refers to our prior understanding that the rate map of a grid cells consists of bumps of elevated firing rates,  located at various points in the two dimensional space, and therefore, our estimation problem is as follows:  
 \begin{eqnarray*}
 \bl &\triangleq&  \underset{\bm{\beta} \in \R^p}{\argmin}  \Bigl \{ \sum_{i=1}^n  \left [  \eta(\bm{r_i}) - y_i \log \eta(\bm{r_i})  \right]    + \lambda \|  \bm{\beta} \|_1  \Bigr \},
\\
  &=&  \underset{\bm{\beta} \in \R^p}{\argmin}  \Bigl \{  \sum_{i=1}^n \left [ \exp(\bm{x_i}^\top  \bm{\beta})- y_i \bm{x_i}^\top  \bm{\beta}   \right]   + \lambda \|  \bm{\beta} \|_1  \Bigr \}.
 \end{eqnarray*}
 Here we use the negative log-likelihood in equation \eqref{eq:poiss} as the cost function, that is, $\phi(y,\bm{x}^\top \bm{\beta})= y \bm{x}^\top  \bm{\beta}-\exp(\bm{x}^\top  \bm{\beta})+\log y!$. We remind the reader that we will use $\alo$ formula that was obtained in Theorem \ref{thm:lasso_approx}. Figures \ref{fig:T9C3}  illustrate that $\alo$ is reasonable approximation of $\lo$, allowing computationally efficient tuning of $\lambda$. To see the effect of $\lambda$ of the rate map, we also present the maps resulting from small and large values of $\lambda$, leading to under and over smooth rate maps, respectively. As it pertains to the reported run times, all fittings in this section were performed using the \texttt{glmnet} package \cite{QHFTS13} in MATLAB.


\begin{figure}\vspace{-0.8cm}
\hspace{1.5cm}
        \includegraphics[scale=0.4]{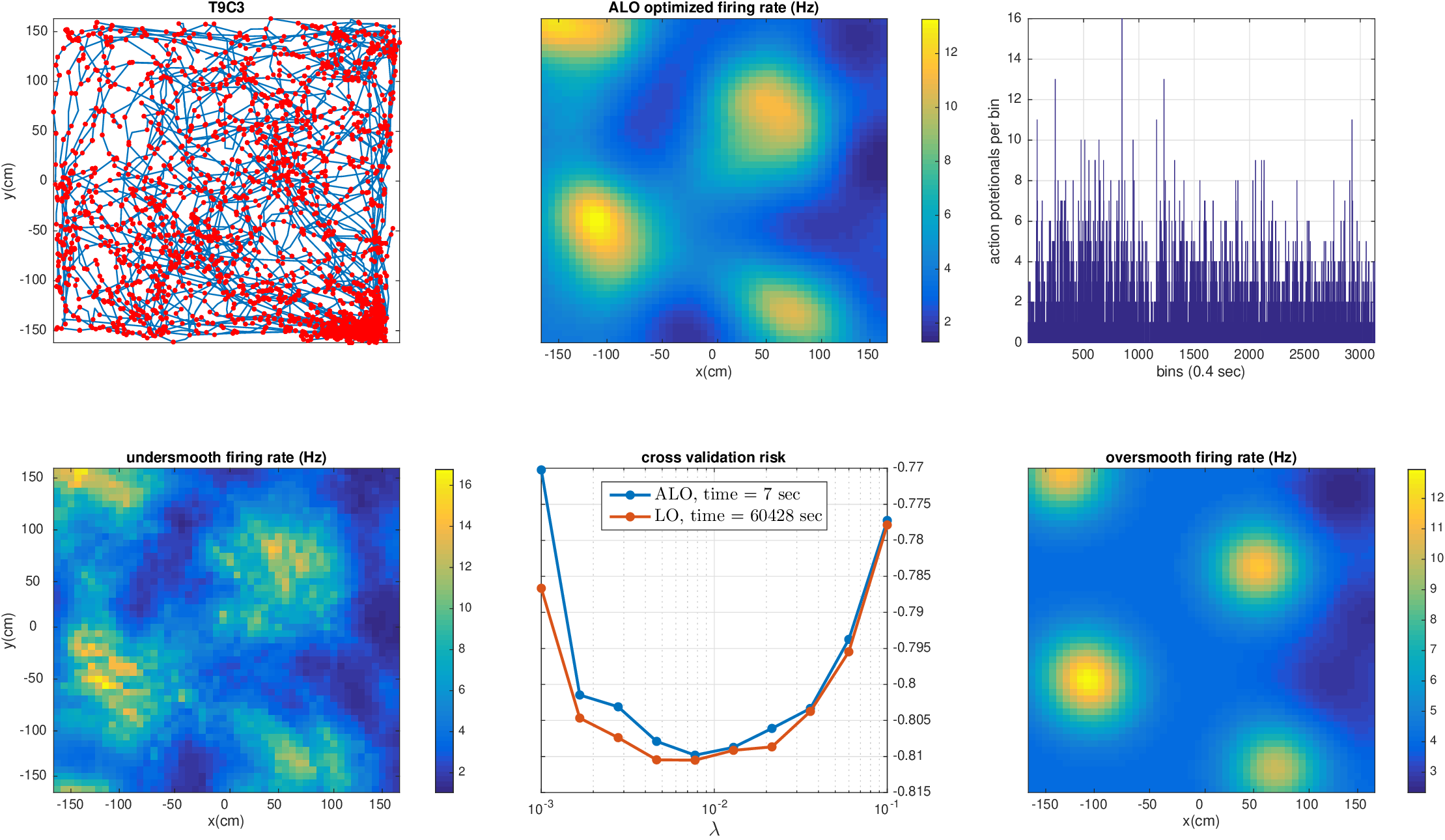}        
        \caption{Top left: Spike locations (red) are superimposed on the animal's trajectory(black). Firing fields are
areas covered by a cluster of action potentials. The firing fields of a grid cell form a periodic triangular matrix tiling
the entire environment available to the animal. Top middle: $\alo$-based firing rate. Top right: $\lo$-based firing rate.  Bottom left: $\lambda=0.001$-based firing rate. Bottom middle: $\alo$ and $\lo$ over a wide range of $\lambda$s. Bottom right: $\lambda=0.1$-based firing rate. }\label{fig:T9C3}
\end{figure}

\section{Concluding Remarks}\label{sec:conc}
Leave-one-out cross validation ($\lo$) is an intuitive and conceptually simple risk estimation technique. Despite its low bias in estimating the extra-sample prediction error, the high computational complexity of $\lo$ has limited its applications for high-dimensional problems. In this paper, by combining a single step of the Newton method  with  low-rank matrix identities, we obtained an approximate formula for $\lo$, called $\alo$. We  showed how $\alo$ can be applied to popular non-differentiable regularizers, such as LASSO. With the aid of theoretical results and numerical experiments, we showed that $\alo$  offers a computationally efficient and statistically accurate estimate of the extra-sample prediction error in high-dimensions.

Important directions for future work involve various approximations that further reduce the computational complexity. The computational bottleneck of $\alo$ is the inversion of the large generalized hat matrix $\bm{H}$. This can make the application of $\alo$ to ultra high dimensional problems computationally challenging. Since the diagonals of our $\bm{H}$ matrix can be represented as leverage scores of an augmented $\bm{X}$ matrix,  scalable methods to approximately compute the  leverage score may offer a promising avenue for future work. For example \cite{DMMW12} offers a randomized method to estimate the leverage scores.  However, the randomized algorithm presented in \cite{DMMW12}  applies to the $p \ll n $ case, making it challenging to apply these methods to high-dimensional settings where
$p$ is also very large. Nevertheless this is certainly a promising direction for speeding up $\alo$. 

In another line of work, the generalized cross-validation approach \cite{CW79,GHW79}  approximates the diagonal elements of $\bm{H}$  with $\tr(\bm{H})/n$. Computationally efficient randomized estimates of $\tr(\bm{H})$ can be produced without having any explicit calculations of this matrix \cite{DG91,WJGG95,G98,LWXGKK00}. The theoretical study of the additional errors introduced by these randomized approximations, and the scalable implementations of them is another promising avenue for future work.

\small
\bibliographystyle{apalike}
\bibliography{myrefs}


\newpage
\appendix

\section{Proofs (FOR ON-LINE PUBLICATION ONLY)}\label{sec:proof}

\subsection{Several concentration results for Gaussian random vectors and matrices}

In this section, we mention a few concentration results that will be used multiple times in the proofs of our main results. We standard and well-known Gaussian tail bound:

\begin{lemma}\label{lem:GaussianTail}
Let $Z\sim N(0,1)$. Further assume that $t>1$. Then,
\[
 P(Z >t) \leq \frac{1}{\sqrt{2 \pi}} {\rm e}^{- \frac{t^2}{2}}. 
 \]
\end{lemma}

Our next lemma obtains a tail bound for the magnitude of a Gaussian random vector and the maximum eigenvalue of a Gaussian matrix. 

\begin{lemma}[Due to \cite{BLM13}] \label{lem:massart}
Let $\bm{x} \sim \N (0,\bm{\Sigma})$ with $\rho_{\max} \triangleq \sigma_{\max} \left( \bm{\Sigma} \right)$, where $\bm{\Sigma} \in \R^{p \times p}$ then
\begin{eqnarray}
\Pr \left [   \left \|  \bm{x} \right \|_2^2  >  5 p\rho_{\max}  \right] &\leq&  e^{-p}.
\end{eqnarray}
Furthermore, if $\bm{X} \in \R^{n \times p}$ is composed of independently distributed $\N (0,\frac{1}{n})$ entries, then
\begin{eqnarray}
\Pr \left [  \sqrt{\sigma_{\max}\left( \bm{X}^\top   \bm{X} \right)} \geq 1+ \sqrt{\frac{p}{n}} +t \right] &\leq& e^{-\frac{nt^2}{2}}.
\end{eqnarray}
\end{lemma}


The above lemma shows how we can find a tail bound for the maximum singular value of an iid Gaussian matrix. Below we extend the result to Gaussian matrices whose columns are dependent on each other. Note that Lemma \ref{lem:7} is the same as Lemma \ref{lem:X} with  $ n \rho_{\max} =c$ and $\sqrt{\frac{p}{n}}=\frac{1}{\sqrt{\delta_0}}$. Hence we present the proof Lemma \ref{lem:7} which can be easily used to prove Lemma \ref{lem:X}.

\begin{lemma} \label{lem:7}
$\bm{X} \in \R^{n \times p}$ is composed of independently distributed $\N (0,\bm{\Sigma})$ rows, with $\rho_{\max} \triangleq \sigma_{\max} \left( \bm{\Sigma} \right)$, where $\bm{\Sigma} \in \R^{p \times p}$  then
\begin{eqnarray}
\Pr \left [  \sigma_{\max}\left( \bm{X   X}^\top\right) \geq \left(\sqrt{n} + 3\sqrt{p}\right)^2 \rho_{\max}  \right ] \leq e^{-p }.
\end{eqnarray}
\end{lemma}
\begin{proof}
Since $\bm{X} \in \R^{n \times p}$ is composed of independently distributed $\N (0,\bm{\Sigma})$ rows, then
\begin{eqnarray}
\lefteqn{\Pr \left [  \sigma_{\max}\left( \bm{X   X}^\top\right) \geq \sigma_0 \right] = \Pr \left [ \sigma_{\max}\left( \bm{X}^\top   \bm{X}\right) \geq \sigma_0    \right]} \nonumber \\
&=& \Pr \left [ \max_{\| \bm{u} \|_2^2 \leq 1} \left \|\bm{ X u} \right \|_2^2 \geq \sigma_0 \right] = \Pr \left [ \max_{\|\bm{ u} \|_2^2 \leq 1} \left \| \bm{Z \Sigma}^{1/2} \bm{u} \right \|_2^2 \geq \sigma_0 \right] \nonumber
\\
&=& \Pr \left [ \max_{\| \bm{\Sigma}^{-1/2} \bm{u} \|_2^2 \leq 1} \left \| \bm{Z u} \right \|_2^2 \geq \sigma_0 \right] \leq \Pr \left [ \max_{   \| \frac{\bm{u}}{\sqrt{\rho_{\max} } }  \|_2^2  \leq 1  } \left \| \bm{Z u} \right \|_2^2 \geq \sigma_0 \right] \nonumber
\\
&=& \Pr \left [ \max_{   \| \bm{u} \|_2^2  \leq 1  } \left \| \bm{Z u} \right \|_2^2 \geq \frac{\sigma_0}{\rho_{\max}} \right]
= \Pr \left [\sqrt{ \sigma_{\max} \left ( \frac{\bm{Z}^\top\bm{ Z}}{n} \right)}  \geq \sqrt{\frac{\sigma_0}{n\rho_{\max}}} \right],
\end{eqnarray}
where $\bm{Z} \in \R^{n \times p}$ is composed of independently distributed $\N(0,1)$ entries. As a consequence of Lemma \ref{lem:massart}, and letting $\sigma_0 = n \rho_{\max} \left(1 + \sqrt{\frac{p}{n}} + t \right)^2$, we get
\begin{eqnarray}
\Pr \left [  \sigma_{\max}\left( \bm{X   X}^\top\right) \geq n \rho_{\max} \left(1 + \sqrt{\frac{p}{n}} + t \right)^2 \right ] \leq e^{-nt^2/2 }.\label{eq:123}
\end{eqnarray}
By substituting $t=\sqrt{\frac{2p}{n}}$ in \eqref{eq:123}, and noting that $3 > 1 + \sqrt{2}$, we get
\begin{eqnarray}
\Pr \left [  \sigma_{\max}\left( \bm{X   X}^\top\right) \geq n \rho_{\max} \left(1 + 3\sqrt{\frac{p}{n}} \right)^2 \right ] \leq e^{-p }. \nonumber
\end{eqnarray}
\end{proof}

\subsection{Proof of Theorem \ref{thm:lasso_approx} }\label{ssec:proofthmnonsmooth}

\subsubsection{Roadmap of the proof}
We first remind the reader that $r_\alpha(z) \triangleq \frac{1}{\alpha} ( \log(1+ e^{-\alpha z})+ \log (1+ e^{\alpha z}))$. Before we discuss the proof, let us mention the following definitions:
\begin{eqnarray}
h_\alpha(\bm{\beta}) &\triangleq& \sum_{i=1}^n \ell(y_i| \bm{x_i}^{\top} \bm{\beta}) + \lambda \sum_{i=1}^p r_\alpha ({\beta}_i), \ \ \ \ \ \ \ \ h(\bm{\beta}) \triangleq \sum_{i=1}^n \ell(y_i| \bm{x_i}^{\top} \bm{\beta}) + \lambda \sum_{i=1}^p |{\beta}_i|, \\
\bl^\alpha &\triangleq& \arg\min_{\bm{\beta}} h_\alpha(\bm{\beta}),  \ \ \ \ \ \ \ \ \ \ \ \ \  \ \ \ \ \ \ \ \ \ \ \ \ \ \ \ \ \hat{\bm{\beta}} \triangleq \arg\min_{\bm{\beta}} h(\bm{\beta}).
\end{eqnarray}

Note that according to Assumptions \ref{A1} and \ref{A3}, $\bl^\alpha$ and $\bl$ are unique.  We first mention a few structural properties of $r_\alpha(z)$ that will be used throughout our proof. Since the proofs of these results are straightforward, we skip them. 

\begin{lemma}\label{lem:1:uniform}
For any $\alpha>0$ we have $r_\alpha (z) \geq |z|$, and 
\[
\sup_z |r_\alpha (z) - |z| | \leq \frac{2 \log 2}{\alpha}. 
\]
In particular, as $\alpha \rightarrow \infty$, $r_\alpha(z)$ uniformly converges to $|z|$.
\end{lemma}

\begin{lemma}\label{lem:ralpha:limits}
$r_\alpha(z)$ is infinitely many times differentiable, and 
\begin{eqnarray}
\rd_\alpha(z) &=& \frac{{\rm e}^{\alpha z} - {\rm e}^{-\alpha z} }{{\rm e}^{\alpha z} + {\rm e}^{-\alpha z}+2} \nonumber \\
\rdd_\alpha(z) &=&   \frac{2 \alpha}{({\rm e}^{\alpha z} + {\rm e}^{-\alpha z}+2)}. 
\end{eqnarray}
Furthermore, if $|z_\alpha| < \frac{\zeta_1}{\alpha}$ for a constant $\zeta_1>0$, then $\lim_{\alpha \rightarrow \infty} \rdd_\alpha(z_\alpha) = +\infty$. Finally, if $|z_\alpha|> \zeta_2$ for a constant $\zeta_2>0$, then $\lim_{\alpha \rightarrow \infty} \rdd_\alpha(z_\alpha) =0$ and $\lim_{\alpha \rightarrow \infty} \rd_\alpha(z_\alpha) = 1$. 
\end{lemma}

Now, we show the main steps for finding the following limit
\begin{eqnarray}
 \lim_{\alpha \rightarrow \infty} \bm{H}^{\alpha} &\triangleq&  \lim_{\alpha \rightarrow \infty}\bm{X} \left (\lambda \diag[\bm{\rdd}_\alpha(\bl^{\alpha})] + \bm{X}^\top \diag[\bm{\ldd}(\bl^{\alpha})] \bm{X} \right)^{-1}  \bm{X}^\top \diag[\bm{\ldd}(\bl^{\alpha})]. \nonumber 
 \end{eqnarray}
 In that vein, let
\begin{eqnarray}
\bm{A} &\triangleq& \bm{X}_{S^c}^\top {\rm diag} [\bm{\ldd} (\bl^{\alpha}) ]\bm{X}_{S^c} + {\rm diag}[ \bm{\rdd}^{\alpha}_{S^c}(\bl^{\alpha}) ], \ \ \ \ \ \bm{B} \triangleq \bm{X}_{S^c}^\top  {\rm diag} [\bm{\ldd}(\bl^{\alpha})] \bm{X}_{S}, \nonumber \\
\bm{C} &\triangleq&  \bm{X}_{S}^\top {\rm diag}[ \bm{\ldd}(\bl^{\alpha})] \bm{X}_{S} + {\rm diag} [\bm{\rdd}^{\alpha}_{S} (\bl^{\alpha})], \  \ \ \ \ \ \ \bm{D}\triangleq (\bm{C}- \bm{B}^\top \bm{A}^{-1}\bm{B})^{-1},
\end{eqnarray}
where $S=\{i: |\hat \beta_i| \neq 0 \}$. Based on Theorem \ref{thm:boundingcoeffs} in Section \ref{ssec:upregLASSO}, for large enough $\alpha$, there exist  fixed numbers $\zeta_1, \zeta_2 >0$ such that
\begin{eqnarray*}
\max_{i \in S^c} |\hat \beta_i^\alpha| < \frac{\zeta_1}{\alpha}, \text{ and }
\min_{i \in S} |\hat \beta_i^\alpha| > \zeta_2, 
\end{eqnarray*}
which with Lemma \ref{lem:ralpha:limits} implies  $\rdd_\alpha(\hat \beta_i^\alpha ) \rightarrow \infty$ for $i \in S^c$ and $\rdd_\alpha(\hat \beta_i^\alpha ) \rightarrow 0$ for $i \in S$,   as $\alpha \rightarrow \infty$. Since the diagonal elements of ${\rm diag}[ \bm{\rdd}^{\alpha}_{S^c}(\bl^{\alpha}) ]$ go off to infinity, $\bm{A}^{-1} \rightarrow 0$, as $\alpha \rightarrow \infty$. Furthermore, since the diagonal elements of ${\rm diag} [\bm{\rdd}^{\alpha}_{S} (\bl^{\alpha})]$ converge to zero, $\lim_{\alpha \rightarrow \infty} \bm{D} = (\bm{X}_S^\top {\rm diag}[ \bm{\ldd}(\bl^{\alpha})] \bm{X}_S)^{-1}$. Therefore, by using the following identity 
\begin{eqnarray}
\begin{bmatrix}
\bm{A} & \bm{B}\\
\bm{B}^\top & \bm{C} 
\end{bmatrix}^{-1} = \begin{bmatrix}
\bm{A}^{-1} + \bm{A}^{-1}\bm{B DB}^\top \bm{A}^{-1} & - \bm{A}^{-1}\bm{B D} \\
-\bm{D B}^\top \bm{A}^{-1} & \bm{D}
\end{bmatrix},
\end{eqnarray}
and noting that $\lim_{\alpha \rightarrow \infty} \bm{A}^{-1} + \bm{A}^{-1}\bm{B DB}^\top \bm{A}^{-1} =0$, $ \lim_{\alpha \rightarrow \infty}- \bm{A}^{-1}\bm{B D} =0$,  we obtain
\begin{eqnarray*}
\lim_{\alpha \rightarrow \infty }  \bm{H}^{\alpha} &=& \lim_{\alpha \rightarrow \infty}  \bm{X} \left (\lambda \diag[\bm{\rdd}(\bl^{\alpha})] + \bm{X}^\top \diag[\bm{\ldd}(\bl^{\alpha})] \bm{X} \right)^{-1}  \bm{X}^\top \diag[\bm{\ldd}(\bl^{\alpha})] \nonumber \\
&=& \bm{X}_S \left ( \bm{X}_S^\top \diag[\bm{\ldd} (\bl)]   \bm{X}_S  \right)^{-1} \bm{X}_S^\top \diag[\bm{\ldd} (\bl)].
\end{eqnarray*}

Note that in Lemma \ref{thm:firstbdlasso} in Section \ref{ssec:upperbetadiff} we prove that $\|\hat{\bm{\beta}}^\alpha - \hat{\bm{\beta}}\|_2 \rightarrow 0$ as $\alpha \rightarrow \infty$. Hence, from the continuity of the second derivative of $\ell$ (Assumption \ref{A4}) we have $\bm{\ldd}(\bl^{\alpha}) \rightarrow \bm{\ldd}(\bl)$ as $\alpha \rightarrow \infty$.



\subsubsection{Proof of $\|\hat{\bm{\beta}}^\alpha - \hat{\bm{\beta}}\|_2 \rightarrow 0$}\label{ssec:upperbetadiff}
\begin{lemma}\label{thm:firstbdlasso}
If Assumptions \ref{A1} and \ref{A3} hold, i.e. uniqueness of $\bl$ and $\bl^\alpha$, then $\lim_{\alpha \rightarrow \infty } \|\hat{\bm{\beta}}^\alpha - \hat{\bm{\beta}}\|_2 = 0$.
\end{lemma}
\textbf{Proof.}
First note that according to Lemma \ref{lem:1:uniform}, we have
\[
|h(\bm{\beta})-h_\alpha(\bm{\beta})| \leq \frac{2 p \log 2}{\alpha}. 
\]
Hence, we have
\begin{equation}\label{eq:lb_gga}
h_\alpha ({\bm{\hat \beta}}^\alpha) \geq h({\bm{\hat \beta}}^\alpha) - \frac{2 p \log 2}{\alpha} \geq h({\bm{\hat \beta}}) - \frac{2 p \log 2}{\alpha},
\end{equation}
and
\begin{equation}\label{eq:up_gga}
h_\alpha({\bm{\hat \beta}}) \leq h({\bm{\hat \beta}}) + \frac{2 p \log 2}{\alpha}.
\end{equation}

Suppose that $\|\bl^\alpha- \bl\|_2$ does not go to zero as $\alpha \rightarrow \infty$. Then, there exists an $\epsilon>0$ for which we can find a sequence $\alpha_1, \alpha_2, \ldots, $ such that
\begin{equation}\label{eq:assumpseq}
\|\bl^{\alpha_i}- \bl\|_2 > \epsilon. 
\end{equation}
According to Lemma \ref{lem:1:uniform}, we have 
\begin{equation}\label{eq:compactnesseq}
\lambda \|\bl^{\alpha_i} \|_1 \overset{(a)}{\leq} \lambda \sum_{j=1}^p r_{\alpha_i} (\hat \beta^{\alpha_i}_j) \overset{(b)}{\leq}  h_{\alpha_i} (\bl^{\alpha_i}) \overset{(c)}{\leq} h_{\alpha_i} (0) {=} \sum_{j=1}^n \ell(y_j|0)+ \frac{2p \log2}{\alpha_i}. 
\end{equation}
Note that Inequality (a) uses Lemma \ref{lem:1:uniform} which proves $|\hat \beta^{\alpha_i}_j| \leq r_{\alpha_i} (\hat \beta^{\alpha_i}_j)$. Inequality (b) is due to the fact that $h_\alpha(\bm{\beta}) = \sum_{i=1}^n \ell(y_i| \bm{x_i}^{\top} \bm{\beta}) + \sum_{i=1}^p r_\alpha (\beta_i)$ and we assume that the loss function returns positive numbers. Inequality (c) is due to the fact that $\bl^{\alpha_i}$ is the minimizer of $h_{\alpha_i}(\bm{\beta})$. 

According to \eqref{eq:compactnesseq} the sequence $\bl^{\alpha_1}, \bl^{\alpha_2}, \ldots$ belongs to a compact set, and hence has a converging subsequence,  called $\bl^{\tilde{\alpha}_1}, \bl^{\tilde{\alpha}_2}, \ldots$. Suppose that $\bl^{\tilde{\alpha}_1}, \bl^{\tilde{\alpha}_2}, \ldots$ converges to ${\bm{\tilde \beta}}$. Therefore, 
\begin{equation}\label{eq:convergalm}
h(\bl^{\tilde{\alpha}_j}) \overset{(d)}{\leq} h_{\tilde{\alpha}_j} (\bl^{\tilde{\alpha}_j}) + \frac{2p \log 2}{\tilde{\alpha}_j} \overset{(e)}{\leq} h_{\tilde{\alpha}_j} (\bl) + \frac{2p \log 2}{\tilde{\alpha}_j} \overset{(f)}{\leq} h (\bl) +  \frac{4p \log 2}{\tilde{\alpha}_j}. 
\end{equation}
Inequality (d) is due to \eqref{eq:lb_gga}. Inequality (e) is true because $\bl^{\tilde{\alpha}_j}$ is the minimizer of $h_{\tilde{\alpha}_j}(\bm{\beta})$, and finally Inequality (f) is due to \eqref{eq:up_gga}. By taking the limit $j \rightarrow \infty$ from both sides of \eqref{eq:convergalm}, we have 
\[
h({\bm{\tilde \beta}}) \leq h(\bl). 
\]
But ${\bm{\tilde \beta}}$ is different from $\bl$, according to \eqref{eq:assumpseq}, contradicting the uniqueness of $\bl$ in Assumption \ref{A1}. $\hfill \Box$. 
\subsubsection{Bounds for regression coefficients in smoothed LASSO}\label{ssec:upregLASSO}

\begin{theorem}\label{thm:boundingcoeffs}
Let $S$ denote the active set of $\bl$, i.e., the location of its non-zero coefficients. Under assumptions \ref{A1}, \ref{A3}, \ref{A4}, and \ref{A2}, there exists a fixed numbers $\zeta_1, \zeta_2>0$, such that for $\alpha$ large enough, we have
\begin{eqnarray*}
\max_{i \in S^c} | \hat \beta^{\alpha}_i| &<& \frac{\zeta_1}{\alpha} \\
\min_{i \in S} | \hat \beta^{\alpha}_i| &>& \zeta_2.
\end{eqnarray*}
 \end{theorem}
 \textbf{Proof.} The optimality conditions
 \begin{eqnarray}
 \sum_{i=1}^n \bm{x_i} \ld(y_i| \bm{x_i}^\top \bl^\alpha) + \lambda \bm{\rd}_\alpha(\bl^\alpha) &=&0,
\\
 \sum_{i=1}^n \bm{x_i} \ld(y_i| \bm{x_i}^\top \bl^\alpha) + \lambda \bm{\hat g} &=&0,
 \end{eqnarray}
lead to
  \begin{eqnarray}\label{eq:boundingsubgrad}
\left  \| \lambda \bm{\rd}_\alpha({\bm{\hat \beta}}^\alpha)- \lambda \bm{\hat g} \right \|_2 = \left \|-\sum_{i=1}^n \bm{x_i} \ld(y_i| \bm{x_i}^\top {\bm{\hat \beta}}^\alpha) + \sum_{i=1}^n \bm{x_i} \ld(y_i| \bm{x_i}^\top {\bm{\hat \beta}}) \right \|_2. 
 \end{eqnarray}
 We know $\|\bl^\alpha-\bl\|_2 \rightarrow 0$ from Lemma \ref{thm:firstbdlasso}. And since $\ell$ is twice differentiable (Assumption \ref{A4}), we can argue that $  \|-\sum \bm{x_i} \ld(y_i| \bm{x_i}^\top {\bm{\hat \beta}}^\alpha) + \sum \bm{x_i} \ld(y_i| \bm{x_i}^\top {\bm{\hat \beta}}) \|_2 \rightarrow 0$ as $\alpha \rightarrow \infty$. Hence, 
 \begin{equation}\label{eq:inftynorm}
  \| \lambda \bm{\rd}_\alpha({\bm{\hat \beta}}^\alpha)- \lambda \bm{\hat g} \|_\infty \leq  \| \lambda \bm{\rd}_\alpha({\bm{\hat \beta}}^\alpha)- \lambda \bm{\hat g}\|_2 \rightarrow 0,
 \end{equation}
 as $\alpha \rightarrow \infty$. This shows that for every $i \in S^c$, $|\alpha \hat \beta^{\alpha}_i|$ should remain bounded as $\alpha \rightarrow \infty$. Suppose that this is not true. Then we find a subsequence that $\alpha_j \hat \beta^{\alpha_j}_i \rightarrow \infty$ as $j \rightarrow \infty$. Then 
 \[
\lim_{j \rightarrow \infty} \rd_{\alpha_j}(\hat \beta^{\alpha_j}_i)= \lim_{j \rightarrow \infty} \frac{e^{\alpha_j \hat \beta^{\alpha_j}_i} -  e^{-\alpha_j \hat \beta^{\alpha_j}_i }}{e^{\alpha_j \hat \beta^{\alpha_j}_i} +  e^{-\alpha_j \hat \beta^{\alpha_j}_i }+2}=1. 
 \]
  If we combine this with Assumption \ref{A2}, we conclude that $\| \lambda \bm{\rd}_\alpha({\bm{\hat \beta}}^\alpha)- \lambda \bm{\hat g} \|_\infty$ will be a constant due to the assumption  $\sup_{i \in S^c} |\hat g_i| < 1$. This is in contradiction with \eqref{eq:inftynorm}. Hence, we have proved that for every $i \in S^c$, $|\alpha \hat \beta^{\alpha}_i| $ remains bounded.

Next, we show that $\min_{i \in S} | \hat \beta^{\alpha}_i|$ is bounded away from zero in the limit $\alpha \rightarrow \infty$. Define $\min_{i \in S} | \hat \beta_i|= \gamma>0$. Lemma \ref{thm:firstbdlasso}  implies $\max_{i \in S} |\hat \beta_i^\alpha- \hat \beta_i | \rightarrow 0$, and therefore, for $\alpha$ large enough, we have
\[
\max_{i \in S} |\hat \beta_i^\alpha - \hat \beta_i| < \gamma/2, 
\]
leading to
\[
\min_{i \in S} | \hat \beta_i^\alpha | > \min_{i \in S} | \hat \beta_i |- \max_{i \in S} | \hat \beta_i^\alpha - \hat \beta_i | > \zeta_2 \triangleq \gamma/2.
\] 

  \hfill $\Box$

\subsection{Proof of Theorem \ref{thm:lassobounderror}}\label{ssec:proofthmapprox}

 The following lemma plays a critical role in our proof of Theorem \ref{thm:lassobounderror}. 
\begin{lemma}\label{lem:matrix_inversion}
Consider a class of symmetric positive definite matrices of the form
\begin{eqnarray}
\bm{\Gamma}_\delta \triangleq 
\begin{bmatrix}
a+ \delta & \bm{b}^\top \\
\bm{b} & \bm{C}
\end{bmatrix},
\end{eqnarray}
where $a>0$, $\delta \geq 0$ and $\bm{C} \in \mathbb{R}^{n-1 \times n-1}$. Then, for any vector $\bm{v} \in \mathbb{R}^n$ we have
\[
\lim_{\delta \rightarrow \infty}\bm{v}^\top \bm{\Gamma}_\delta^{-1} \bm{v}  \leq \bm{v}^\top \bm{\Gamma}_\delta^{-1} \bm{v} \leq \bm{v}^\top \bm{\Gamma}_0^{-1} \bm{v}.
\]
Furthermore, if we define $\bm{v}_{\slash 1} \triangleq (v_2, v_3, \ldots, v_n)^{\top}$, then 
$\lim_{\delta \rightarrow \infty}\bm{v}^\top \bm{\Gamma}_\delta^{-1} \bm{v} = \bm{v}_{\slash 1}^\top \bm{C}^{-1} \bm{v}_{\slash 1}$.
\end{lemma}
\textbf{Proof:}
Define $\kappa \triangleq a+\delta - \bm{b}^\top \bm{C}^{-1}\bm{b}$. Note that since the matrix $\bm{\Gamma}_\delta$ is always positive definite, for any value of $\delta$, $\kappa >0$. 
By using the formulas for the inverse of block matrices we have
\begin{eqnarray}
\bm{\Gamma}^{-1}_\delta = 
\begin{bmatrix}
\frac{1}{\kappa} & -\frac{\bm{b}^\top \bm{C}^{-1}}{k}\\
 - \frac{\bm{C}^{-1} \bm{b}}{\kappa} &  \frac{\bm{C}^{-1} \bm{bb}^\top \bm{C}^{-1}}{\kappa}+ \bm{C}^{-1}
\end{bmatrix}.
\end{eqnarray}
Define $\bm{v}_{\slash 1} \triangleq (v_2, v_3, \ldots, v_n)^{\top}$. 
\begin{eqnarray}
\bm{v}^\top \bm{\Gamma}^{-1}_\delta \bm{v}  &=& \frac{v_1^2}{\kappa} + \bm{v}_{\slash 1}^\top \bm{C}^{-1} \bm{v}_{\slash 1} + \frac{\bm{v}_{\slash 1}^\top \bm{C}^{-1} \bm{bb}^\top \bm{C}^{-1}\bm{v}_{\slash 1}}{\kappa} - 2\frac{\bm{v}_{\slash 1}^\top \bm{C}^{-1}\bm{b} v_1}{\kappa}  \nonumber \\
&=&  \bm{v}_{\slash 1}^\top \bm{C}^{-1} \bm{v}_{\slash 1} + \frac{1}{\kappa} (v_{1} -\bm{b}^\top \bm{C}^{-1}\bm{v}_{\slash 1} )^2.
\end{eqnarray}
Lemma \ref{lem:matrix_inversion} follows from the monotonicity of $\bm{v}^\top \bm{\Gamma}^{-1}_\delta \bm{v}$ in terms of $\kappa$. 
\hfill $\Box$

\begin{proof}[Proof of Theorem \ref{thm:lassobounderror}]
Before we start the proof, let us emphasize on the following facts that will be used later in the proof. 
\begin{enumerate}
\item Consider an index $i \in (S \cup T)^{c}$. We know that $\hat \beta_i =0$ and the subgradient $|\hat g_i| < 1$. Hence, according to the proof of Theorem \ref{thm:boundingcoeffs} we have $\alpha \hat \beta^{\alpha}_i < \zeta$. Therefore, according to Lemma \ref{lem:ralpha:limits} we have  $\rdd ({{\hat \beta}}^{\alpha}_i) \rightarrow \infty$ as $\alpha \rightarrow \infty$.

\item Consider $i \in S$. Then by definition $\hat \beta_i \neq 0$. Similar to the proof of Theorem \ref{thm:lasso_approx}, we have $\rdd ({{\hat \beta}}^{\alpha}_i) \rightarrow 0$ as $\alpha \rightarrow \infty$.
\end{enumerate}
Hence, we already know the limiting behavior of  $\rdd ({{\hat \beta}}^{\alpha}_i)$ for $i \in S$ and $i \in (S \cup T)^c$ as $\alpha \rightarrow \infty$. The only remaining index set is $T$. Unfortunately, for $i \in T$ we can not specify the limiting behavior of $\rdd ({{\hat \beta}}^{\alpha}_i)$. Hence, our goal is to use Lemma \ref{lem:matrix_inversion} to get around this issue.  
Set $U\triangleq (S \cup T)^c$ and define the matrices
\begin{eqnarray}
\bm{\tilde{A}}^{\alpha} &\triangleq& \bm{X}_{S}^\top {\rm diag} [\bm{\ldd} (\bl^{\alpha}) ]\bm{X}_{S} + {\rm diag}[ \bm{\rdd}^{\alpha}_{S}(\bl^{\alpha}) ],  \ \ \ \ \ 
\bm{\tilde{B}}^{\alpha} \triangleq \bm{X}_{T}^\top {\rm diag} [\bm{\ldd} (\bl^{\alpha}) ]\bm{X}_{T} , \nonumber \\
\bm{\tilde{C}}^\alpha &\triangleq& \bm{X}_{U}^\top {\rm diag} [\bm{\ldd} (\bl^{\alpha}) ]\bm{X}_{U} + {\rm diag}[ \bm{\rdd}^{\alpha}_{U}(\bl^{\alpha}) ],  \ \ \ \ \ 
\bm{\tilde{D}}^\alpha \triangleq \bm{X}_{S}^\top  {\rm diag} [\bm{\ldd}(\bl^{\alpha})] \bm{X}_{T}, \nonumber \\
\bm{\tilde{E}}^\alpha &\triangleq& \bm{X}_{S}^\top  {\rm diag} [\bm{\ldd}(\bl^{\alpha})] \bm{X}_{U},  \ \ \ \ \ \ \ \ \  \ \ \ \ \ \ \ \  \ \ \ \ \ \ \ \ 
\bm{\tilde{F}}^\alpha \triangleq \bm{X}_{T}^\top  {\rm diag} [\bm{\ldd}(\bl^{\alpha})] \bm{X}_{U}.
\end{eqnarray}
Given this notation we have
\begin{eqnarray}
H_{ii}^{\alpha} = \bm{x}_{i}^\top\begin{bmatrix}
\bm{\tilde{A}}^\alpha & \bm{\tilde{D}}^\alpha  & \bm{\tilde{E}}^\alpha \\
{(\bm{\tilde{D}}^\alpha) }^\top & \bm{\tilde{B}}^\alpha + {\rm diag}[ \bm{\rdd}^{\alpha}_{T}(\bl^{\alpha}) ] & \bm{\tilde{F}}^\alpha  \\
(\bm{\tilde{E}}^{\alpha})^\top  & (\bm{\tilde{F}}^{\alpha})^\top  & \bm{\tilde{C}}^{\alpha} 
\end{bmatrix}^{-1}
\bm{x}_{i} \ldd_i(\bl^{\alpha}).
\end{eqnarray}
Here each element of ${\rm diag}[ \bm{\rdd}^{\alpha}_{T}(\bl^{\alpha})]$ may converge to any number in the range $[0, \infty]$. Hence we use Lemma \ref{lem:matrix_inversion} to find upper and lower bounds for $H_{ii}^{\alpha}$. According to Lemma \ref{lem:matrix_inversion} we have 
\begin{eqnarray}\label{eq:upperbd:lasso}
H_{ii}^{\alpha} \leq \bm{x}_{i}^\top\begin{bmatrix}
\bm{\tilde{A}}^\alpha & \bm{\tilde{D}}^\alpha  & \bm{\tilde{E}}^\alpha \\
{(\bm{\tilde{D}}^\alpha) }^\top & \bm{\tilde{B}}^\alpha & \bm{\tilde{F}}^\alpha  \\
(\bm{\tilde{E}}^{\alpha})^\top  & (\bm{\tilde{F}}^{\alpha})^\top  & \bm{\tilde{C}}^{\alpha} 
\end{bmatrix}^{-1}
\bm{x}_{i} \ldd_i(\bl^{\alpha}),
\end{eqnarray}
and 
\begin{eqnarray}\label{eq:lowerbd:lasso}
H_{ii}^{\alpha} &\geq& \lim_{\delta_{|T| \rightarrow \infty}}\ldots \lim_{\delta_1 \rightarrow \infty} \bm{x}_{i}^\top\begin{bmatrix}
\bm{\tilde{A}}^\alpha & \bm{\tilde{D}}^\alpha  & \bm{\tilde{E}}^\alpha \\
{(\bm{\tilde{D}}^\alpha) }^\top & \bm{\tilde{B}}^\alpha + {\rm diag}[ \delta_1, \delta_2, \ldots, \delta_{|T|} ] & \bm{\tilde{F}}^\alpha  \\
(\bm{\tilde{E}}^{\alpha})^\top  & (\bm{\tilde{F}}^{\alpha})^\top  & \bm{\tilde{C}}^{\alpha} 
\end{bmatrix}^{-1}
\bm{x}_{i} \ldd_i(\bl^{\alpha}) \nonumber \\
&=& \bm{x}_{i, , S\cup U}^\top\begin{bmatrix}
\bm{\tilde{A}}^{\alpha} & \bm{\tilde{E}}^{\alpha} \\
(\bm{\tilde{E}}^{\alpha})^\top  & \bm{\tilde{C}}^{\alpha} 
\end{bmatrix}^{-1}
\bm{x}_{i, S \cup U} \ldd_i(\bl^{\alpha}).
\end{eqnarray}
The rest of the proof is similar to the proof of Theorem \ref{thm:lasso_approx}; we take the limit $\alpha \rightarrow \infty$ from both sides of \eqref{eq:upperbd:lasso} and \eqref{eq:lowerbd:lasso}, and then use the block matrix inversion formulas (similar to those used in the proof of Theorem \ref{thm:lasso_approx}) and the fact that $(\bm{\tilde{A}}^{\alpha})^{-1} \rightarrow 0$ as $\alpha \rightarrow \infty$ to complete the proof.
\end{proof}

\subsection{Derivation of \eqref{eq:bridgeHalo}} \label{ssec:bridgederivation}

\subsubsection{Roadmap of the derivations}\label{ssec:bridgeroadmap}
The goal of this section is to derive the $\alo$ formula, presented in \eqref{eq:bridgeHalo}, for the following class of bridge estimators:
 \begin{eqnarray}\label{eq:ori_optalpha}
\bl \triangleq  \underset{\bm{\beta} \in \R^p}{\argmin}  \Bigl \{ \sum_{i=1}^n  \ell ( y_i|\bm{x_i}^\top \bm{\beta} ) + \lambda \|\bm{\beta}\|_q^q  \Bigr \}, 
\end{eqnarray}
where $q \in (1,2)$.  Since $\|\bm{\beta}\|_q^q$ is not twice differentiable at zero, similar to what we did for LASSO, we first consider a smoothed version of the bridge regularizer:
\begin{equation}\label{eq:smoothingffunc}
    r^q_\gamma(z) = \frac{1}{\gamma}\int |u|^q \psi((z - u)/\gamma) du,
\end{equation}
where $\psi$ satisfies the following conditions:
\begin{description}
    \item[(i)] $\psi$ has  a compact support, i.e.,  $\mathrm{supp}(\psi) = [-1, 1]$. Also, $\psi(w) \geq 0$ for every $w$. 
    \item[(ii)] $\int \psi(w)dw = 1$ and $\psi(0) > 0$;
    \item[(iii)] $\psi$ is infinitely many times smooth and symmetric around 0 on $\mathbb{R}$;
\end{description}
The two important properties of $r^q_\gamma(z)$ are
\begin{enumerate}
\item $r^q_\gamma(z)$ is infinitely many times differentiable for any nonzero value of $\gamma$. 
\item $|r^q_\gamma(z)-|z|^q | \rightarrow 0$ as $\gamma \rightarrow 0$.  This claim will be proved in Lemma \ref{lem:approxerrorbridge} below. 
\end{enumerate}
Hence, instead of finding the $\alo$ formula directly for \eqref{eq:ori_optalpha}, we start with 
\begin{equation}\label{eq:bridgesmoothed}
 \bl^\gamma  \triangleq \arg\min_{\bm{\beta}}\sum_{i=1}^n \ell(y_i| \bm{x_i}^\top \bm{\beta}) + \lambda \sum_{i=1}^p r^q_\gamma ({\beta}_i). 
\end{equation}
Given that both the loss function and the regularizer are smooth in \eqref{eq:bridgesmoothed}, we can use \eqref{eq:aloformulafinal} to obtain the following formula as the estimate of the out-of-sample prediction error of $ \bl^\gamma$:
\begin{eqnarray}\label{eq:aloformulagamma}
\alo^{\gamma} &\triangleq& \frac{1}{n} \sum_{i=1}^n \phi \left (y_i,  \bm{x_i}^\top \bl^\gamma +\left(\frac{\ld_i(\bl^\gamma)}{\ldd_i(\bl^\gamma)}\right) \left(  \frac{H^{\gamma}_{ii}}{1 - H^{\gamma}_{ii}} \right)     \right), 
\end{eqnarray}
where 
\begin{eqnarray}
 \bm{H}^{\gamma} &\triangleq&  \bm{X} \left (\lambda \diag[\bm{\rdd}^q_\gamma(\bl^{\gamma})] + \bm{X}^\top \diag[\bm{\ldd}(\bl^{\gamma})] \bm{X} \right)^{-1}  \bm{X}^\top \diag[\bm{\ldd}(\bl^{\gamma})].
  \end{eqnarray}

  Note that we are interested in $\alo^{\gamma}$ for large values of $\gamma$. Hence, as suggested for the LASSO problem in Section \ref{ssec:non-diff},  we calculate $\lim_{\gamma \rightarrow 0} \alo^{\gamma}$. In Section \ref{proof:thmbridge} we prove the following theorem:
 
 \begin{theorem}\label{thm:aloapproxbridge}
 If the loss function is twice continuously differentiable with respect to its second argument, and the optimization problem in \eqref{eq:bridgesmoothed} has a unique solution for every $\gamma$, then 
 \begin{eqnarray}\label{eq:aloformulagamma}
\lim_{\gamma \rightarrow 0} \alo^{\gamma} &\triangleq& \frac{1}{n} \sum_{i=1}^n \phi \left (y_i,  \bm{x_i}^\top \bl +\left(\frac{\ld_i(\bl)}{\ldd_i(\bl)}\right) \left(  \frac{H_{ii}}{1 - H_{ii}} \right)     \right),
\end{eqnarray}
where
\begin{eqnarray}
\bm{H} \triangleq \bm{X}_S \left ( \bm{X}_S^\top \diag[\bm{\ldd} (\bl)]   \bm{X}_S + \lambda \diag[\bm{\rdd}^q_{S} (\bm{\hat{\beta}})]\right)^{-1} \bm{X}_S^\top \diag[\bm{\ldd} (\bl)],
  \end{eqnarray}
 $S$ is the active set of $\bl$, and ${\rdd}^q (u) = q (q-1) |u|^{q-2}$.

 \end{theorem}
  
  Proof of this Theorem is presented in Section \ref{proof:thmbridge}. We will first prove in Lemma \ref{lem:convbridge} that 
\[
\|\bl^\gamma-\bl\|_2 \rightarrow 0. 
\]
Since, $\|\bl^\gamma- \bl\|_2 \rightarrow 0$, and $\ld$ and $\ldd$ functions are continuous, it is straightforward to prove that as $\gamma \rightarrow 0$ 
\[
\ld_i(\bl^\gamma) \rightarrow \ld_i(\bl), \ \ \ \ \ \ \ldd_i(\bl^\gamma) \rightarrow \ldd_i(\bl).
\]
Hence, the final remaining challenge in proving Theorem \ref{thm:aloapproxbridge}  is to calculate $\lim_{\gamma \rightarrow 0} H^{\gamma}_{ii}$. In Section \ref{proof:thmbridge} we prove that
  \begin{eqnarray*}
\lim_{\gamma \rightarrow 0 }  \bm{H}^{\gamma} &=&  \bm{X}_S \left ( \bm{X}_S^\top \diag[\bm{\ldd} (\bl)]   \bm{X}_S + \lambda \diag[\bm{\rdd}^q_S (\bm{\hat{\beta}})\right)^{-1} \bm{X}_S^\top \diag[\bm{\ldd} (\bl)].
\end{eqnarray*}

\subsubsection{Basic properties of $r^q_\gamma(\cdot)$}

\begin{lemma}\label{lem:approxerrorbridge}
The smoothed regularizer $r^q_\gamma(\cdot)$ satisfies
\[
\sup_{|z|<M}|r^q_\gamma(z)-|z|^q | \leq q(M+\gamma)^{q-1} \gamma. 
\]
\end{lemma}
\begin{proof}
According to the symmetry, we only consider $z\geq 0$. We have
\begin{eqnarray}
|r^q_\gamma(z)- |z|^q| &=& \frac{1}{\gamma} \left| \int_{-\gamma}^\gamma (|z-u|^q - |z|^q)\psi\left(\frac{u}{\gamma}\right) du \right| \nonumber \\
&\overset{(a)}{\leq}& (z+\gamma)^q - z^q \overset{(b)}{=} q \tilde{z}^{q-1} \gamma \leq q(M+\gamma)^{q-1} \gamma. 
\end{eqnarray}
Note that inequality (a) is due to the fact that since $z>0$, the difference between $|z-u|^q - |z|^q$ is maximized when $u= - \gamma$. In other words,
\[
||z-u|^q - |z|^q| \leq (z+\gamma)^q - z^q,  \ \ \ \forall u \in [-\gamma, \gamma].
\]
Furthermore, equality (b) is a result of the mean value theorem and $\tilde{z} \in (z, z+\gamma)$.  
\end{proof}
  
  \subsubsection{Proof of Theorem \ref{thm:aloapproxbridge}}\label{proof:thmbridge}

Consider the following definitions:
\begin{eqnarray}\label{eq:defsmoothellq}
h^q_\gamma(\bm{\beta}) &\triangleq& \sum_{i=1}^n \ell(y_i| \bm{x_i}^{\top} \bm{\beta}) + \lambda \sum_{i=1}^p r^q_\gamma ({\beta}_i), \ \ \ \ \ \ \ \ h^q(\bm{\beta}) \triangleq \sum_{i=1}^n \ell(y_i| \bm{x_i}^{\top} \bm{\beta}) + \lambda \sum_{i=1}^p |{\beta}_i|^q, 
\end{eqnarray}

As discussed in Lemma \ref{lem:approxerrorbridge}, the difference $|r^q_\gamma(z)- |z|^q| $ is bounded by the maximum value that $z$ takes. Our first lemma shows that $\sup_{\gamma \in (0,1]}\| \bl^\gamma\|_\infty <M$. Hence, according to Lemma \ref{lem:approxerrorbridge} the discrepancy between $|r^q_\gamma(\hat \beta^\gamma_{i})- |\hat \beta^\gamma_i|^q|$ goes to zero as $\gamma \rightarrow 0$. 

\begin{lemma}
There exists an $M< \infty$ such that $\sup_{\gamma \in [0,1]}\| \bl^\gamma\|_\infty <M$ and $\|\bl\|_\infty<M$. 
\end{lemma}
 Here we only present the sketch of the proof, and skip the straightforward details. If $\| \bl^\gamma\|_\infty \rightarrow \infty$, then $h^q_\gamma(\bl^\gamma) \rightarrow \infty$. Hence, since $h^q_\gamma(\bm{0})$ is bounded, $\| \bl^\gamma\|_\infty$ cannot go off to infinity.   
 We can now use this lemma to prove that as $\gamma \rightarrow 0$,  $\|\bl^\gamma-\bl\|_2 \rightarrow 0$.

\begin{lemma}\label{lem:convbridge}
If the optimization problems in \eqref{eq:defsmoothellq} have unique solutions, then as $\gamma \rightarrow 0$
\[
\|\bl^\gamma-\bl\|_2 \rightarrow 0. 
\]
\end{lemma}
The proof of this lemma is similar to the proof of Lemma \ref{thm:firstbdlasso}, and is hence skipped here. As mentioned in Section \ref{ssec:bridgeroadmap}, the main step in proving Theorem \ref{thm:aloapproxbridge} is to find the limit of $\lim_{\gamma \rightarrow 0} \bm{H_\gamma}$. The main step in this calculation is to calculate $\lim_{\gamma \rightarrow 0} \bm{\rdd}(\bl^{\gamma})$. The following lemma shows how this limit can be calculated.

\begin{lemma}\label{lem:behaviorinfinity} Let $z_\gamma$ denote a function of $\gamma$. 
If  $z_\gamma \rightarrow 0$ as $\gamma \rightarrow 0$, then
\[
\lim_{\gamma \rightarrow 0} \rdd_\gamma^q (z_\gamma) =\infty. 
\]
\end{lemma}
\begin{proof}
Without loss of generality we consider the case $z_\gamma \geq 0$. We consider three different cases. Each case has a slightly different proof strategy.  
\begin{enumerate}
\item Case I: $\frac{z_\gamma}{\gamma} \rightarrow \infty$ or $\frac{z_\gamma}{\gamma} \rightarrow c \geq 1$.
\item Case II: $\frac{z_\gamma}{\gamma} \rightarrow c$, where $c \in (0,1)$.  
\item Case III: $\frac{z_\gamma}{\gamma} \rightarrow 0$. 
\end{enumerate}

It is straightforward to show that
\begin{equation}\label{eq:rddellqder}
\rdd_\gamma^q(z_\gamma) = \int_{-\infty}^\infty q|z_\gamma -u|^{q-1}  {\rm sign}(z_\gamma -u) \frac{1}{\gamma^2}\psd \left({\frac{u}{\gamma}} \right) du.   
\end{equation}
Note that $\psd \left({\frac{u}{\gamma}} \right) =0$ for $u$ outside the interval $[-\gamma, \gamma]$. Now we consider the three cases we described above.  

Case I: We assume that for large enough values of $\gamma$, $z_\gamma > \gamma$. Clearly, this holds when $z_\gamma /\gamma \rightarrow c>1$. However, it may be violated when $\frac{z_\gamma}{\gamma} \rightarrow 1$. But, this special case can be handled with a similar approach and is hence skipped. We have
\begin{eqnarray}
|\rdd_\gamma^q(z_\gamma)| &=&  |\int_{-\gamma}^0 q(z_\gamma -u)^{q-1}  \frac{1}{\gamma^2}\psd \left({\frac{u}{\gamma}} \right) du  + \int_{0}^\gamma q(z_\gamma -u)^{q-1}  \frac{1}{\gamma^2}\psd \left({\frac{u}{\gamma}} \right) du| \nonumber \\
&=& |\int_{0}^\gamma q[(z_\gamma -u)^{q-1}-(z_\gamma +u)^{q-1} ] \frac{1}{\gamma^2}\psd \left({\frac{u}{\gamma}} \right) dt |\nonumber \\
&\overset{(a)}{=}& 2q(q-1) |\int_{0}^\gamma  \tilde{z}_u^{q-2} u  \frac{1}{\gamma^2}\psd \left({\frac{u}{\gamma}} \right) du| \nonumber \\
&\overset{(b)}{\geq}&  2q(q-1) (z_\gamma+\gamma) ^{q-2} |\int_{0}^\gamma  u  \frac{1}{\gamma^2}\psd \left({\frac{u}{\gamma}} \right) du | \nonumber \\
&\overset{(c)}{=}& q(q-1) (z_\gamma+\gamma) ^{q-2}. 
\end{eqnarray}
Equality (a) is due to the mean-value theorem. To obtain (b) we used the fact that  $\tilde{z}_u \in [z_\gamma-\gamma, z_\gamma+\gamma]$, and that $z_\gamma - \gamma >0$ (hence $\tilde{z}_u >0$). The last equality is the result of integration by parts. Note that since $z_\gamma \rightarrow 0$ as $\gamma \rightarrow 0$, $q(q-1) (z_\gamma+\gamma) ^{q-2} \rightarrow \infty$.  \\

Case II:  $\frac{z_\gamma}{\gamma} \rightarrow c$, where $c \in (0,1)$. For large enough values of $\gamma$, we know that $z_\gamma< \gamma$. Hence, according to \eqref{eq:rddellqder} we have
\begin{eqnarray}\label{eq:caseIIell_q}
\rdd_\gamma^q(z_\gamma)&=& -\int_0^\gamma q(z_\gamma +u)^{q-1}  \frac{1}{\gamma^2}\psd \left({\frac{u}{\gamma}} \right) du + \int_{0}^{z_\gamma} q(z_\gamma -u)^{q-1}  \frac{1}{\gamma^2}\psd \left({\frac{u}{\gamma}} \right) du  \nonumber \\
&&- \int_{z_\gamma}^\gamma q|z_\gamma -u|^{q-1}  \frac{1}{\gamma^2}\psd \left({\frac{u}{\gamma}} \right) du \nonumber \\
&\geq& -\int_{z_\gamma}^\gamma q(z_\gamma +u)^{q-1}  \frac{1}{\gamma^2}\psd \left({\frac{u}{\gamma}} \right) du \nonumber \\
&\geq& - q(2z_\gamma )^{q-1} \int_{z_\gamma}^\gamma   \frac{1}{\gamma^2}\psd \left({\frac{u}{\gamma}} \right) du = \frac{q(2z_\gamma )^{q-1}}{\gamma} \psi\left(\frac{z_\gamma}{\gamma}\right). 
\end{eqnarray}
It is straightforward to confirm that $\frac{q(2z_\gamma )^{q-1}}{\gamma} \psi(\frac{z_\gamma}{\gamma}) \rightarrow \infty$. \\

Case III: First note that since $z_\gamma/\gamma \rightarrow 0$, for large enough $\gamma$, $z_\gamma<\gamma/2$. Similar to the derivation in \eqref{eq:caseIIell_q}, we have
\begin{eqnarray}
\rdd_\gamma^q(z_\gamma)&\geq& - \int_{z_\gamma}^\gamma q|z_\gamma -u|^{q-1}  \frac{1}{\gamma^2}\psd \left({\frac{u}{\gamma}} \right) du \nonumber \\
&\geq& -\int_{z_\gamma}^\gamma q(z_\gamma +u)^{q-1}  \frac{1}{\gamma^2}\psd \left({\frac{u}{\gamma}} \right) du \nonumber \\
&\geq& -\int_{\gamma/2}^\gamma q(z_\gamma +u)^{q-1}  \frac{1}{\gamma^2}\psd \left({\frac{u}{\gamma}} \right) du \nonumber \\
&\geq& q(z_\gamma +\frac{\gamma}{2})^{q-1} \frac{r(0.5)}{\gamma}. 
\end{eqnarray}
Again, it is straightforward to see that the last expression goes to $\infty$ as $\gamma \rightarrow 0$. 
\end{proof}

  We remind the reader that our goal is to show that 
    \begin{eqnarray*}
\lim_{\gamma \rightarrow 0 }  \bm{H}^{\gamma} &=&  \bm{X}_S \left ( \bm{X}_S^\top \diag[\bm{\ldd} (\bl)]   \bm{X}_S + \lambda \diag [\bm{\rdd}^q_S (\bm{\hat{\beta}})]\right)^{-1} \bm{X}_S^\top \diag[\bm{\ldd} (\bl)],
\end{eqnarray*}
where $S$ denotes the active set of $\bl$. Since, $\|\bl^\gamma- \bl\|_2 \rightarrow 0$, and $\ld$ and $\ldd$ functions are continuous, it is straightforward to prove that 
\[
\ld_i(\bl^\gamma) \rightarrow \ld_i(\bl), \ \ \ \ \ \ \ldd_i(\bl^\gamma) \rightarrow \ldd_i(\bl). 
\]
Let $S$ denote the set of indices of the non-zero elements of $\bl$, and define
\begin{eqnarray}
\bm{A} &=& \bm{X}_{S^c}^\top {\rm diag} [\bm{\ldd} (\bl^{\gamma}) ]\bm{X}_{S^c} + {\rm diag}[ \bm{\rdd}^{q}_{\gamma,S^c}(\bl^{\gamma}) ], \ \ \ \ \ \bm{B} = \bm{X}_{S^c}^\top  {\rm diag} [\bm{\ldd}(\bl^{\gamma})] \bm{X}_{S}, \nonumber \\
\bm{C} &=&  \bm{X}_{S}^\top {\rm diag}[ \bm{\ldd}(\bl^{\gamma})] \bm{X}_{S} + {\rm diag} [\bm{\rdd}^{q}_{\gamma,S} (\bl^{\gamma})].
\end{eqnarray}
Also define $\bm{D}\triangleq (\bm{C}- \bm{B}^\top \bm{A}^{-1}\bm{B})^{-1}$. According to Lemmas \ref{lem:convbridge} and \ref{lem:behaviorinfinity}, the diagonal elements of ${\rm diag}[ \bm{\rdd}^{q}_{\gamma,S^c}(\bl^{\gamma}) ]$ go off to infinity. Hence, it is straightforward to show that $\bm{A}^{-1} \rightarrow 0$, as $\gamma \rightarrow 0$. By using the following identity 
\begin{eqnarray}
\begin{bmatrix}
\bm{A} & \bm{B}\\
\bm{B}^\top & \bm{C} 
\end{bmatrix}^{-1} = \begin{bmatrix}
\bm{A}^{-1} + \bm{A}^{-1}\bm{B DB}^\top \bm{A}^{-1} & - \bm{A}^{-1}\bm{B D} \\
-\bm{D B}^\top \bm{A}^{-1} & \bm{D}
\end{bmatrix},
\end{eqnarray}
and noting that $\lim_{\gamma \rightarrow 0} \bm{A}^{-1} + \bm{A}^{-1}\bm{B DB}^\top \bm{A}^{-1} =0$, $ \lim_{\gamma \rightarrow 0}- \bm{A}^{-1}\bm{B D} =0$, and $\lim_{\gamma \rightarrow 0} \bm{D} = (\bm{X}_S^\top {\rm diag}[ \bm{\ldd}(\bl)] \bm{X}_S)^{-1}$ we obtain
\begin{eqnarray*}
\lim_{\gamma \rightarrow 0 }  \bm{H}^{\gamma} &=&  \bm{X}_S \left ( \bm{X}_S^\top \diag[\bm{\ldd} (\bl)]   \bm{X}_S + \lambda \diag[ \bm{\rdd}^q_{S} (\bm{\hat{\beta}}_{\lambda})\right)^{-1} \bm{X}_S^\top \diag[\bm{\ldd} (\bl)].
\end{eqnarray*}


\subsection{Proofs of the Lemmas of Section \ref{sec:res}}

\subsubsection{Proof of Lemma \ref{lem:logistic1} } \label{ssec:prooflogistic1}

Since 
\begin{eqnarray*}
\ld_i(\bm{\beta}) &=& -y_i + \frac{e^{\bm{x_i}^\top \bm{ \beta}}}{1+e^{\bm{x_i}^\top \bm{ \beta}}}, \quad \ldd_i(\bm{\beta}) = \frac{e^{\bm{x_i}^\top \bm{ \beta}}}{(1+e^{\bm{x_i}^\top \bm{ \beta}} )^2},  \quad \lddd_i(\bm{\beta}) = \frac{e^{\bm{x_i}^\top \bm{ \beta}} (1-e^{\bm{x_i}^\top \bm{ \beta}} ) }{(1+e^{\bm{x_i}^\top \bm{ \beta}} )^3}  
\end{eqnarray*}
using simple algebra it is straightforward to show that for any $\bm{\beta}$, we have
\begin{eqnarray*}
 \| \ld(\bm{\beta}) \|_{\infty}  \leq 1, \quad \| \ldd(\bm{\beta}) \|_{\infty}  \leq 1/4, \quad  \| \lddd(\bm{\beta}) \|_{\infty}  \leq 1/10
\end{eqnarray*}
Therefore,
\begin{eqnarray*}
 \|  \bm{\ldd}_{/i}( \bm{\beta + \delta}) - \bm{\ldd}_{/ i}(\bm{\beta}  ) \|_2 &\leq&  \|  \bm{\ldd}( \bm{\beta+\delta}) - \bm{\ldd}(\bm{\beta}  ) \|_2 = \sqrt{\sum_i \left(\ldd(\beta_i+\delta_i) - \ldd(\beta_i) \right)^2 } \\
 &=&\sqrt{\sum_i \lddd(\beta_i+\epsilon_i)^2 (\bm{x}_i^T\bm{\delta})^2 } \quad \mbox{using the mean-value Theorem where} \quad \epsilon_i \in [0,\delta_i] \\
 &\leq& \sqrt{\bm{\delta}^\top \bm{X}^\top \bm{X} \bm{\delta}} \leq \sqrt{\sigma_{max}({\bm{X}^\top \bm{X}})} \|\bm{\delta}\|_2. 
\end{eqnarray*}
Finally, based on the inequality above, we have
\begin{eqnarray*}
 \sup_{t \in [0,1]} \frac{  \|  \bm{\ldd}_{/ i}( (1-t) \bli + t \bl ) - \bm{\ldd}_{/ i}(\bl  ) \|_2}{ \|   \bli - \bl   \|_2} &\leq&  \sup_{t \in [0,1]} \frac{  (1-t) \| \bli -\bl  \|_2  \sqrt{\sigma_{max}({\bm{X}^\top \bm{X}})}}{ \|   \bli - \bl   \|_2} \leq \sqrt{\sigma_{max}({\bm{X}^\top \bm{X}})}.
\end{eqnarray*}
The last statement of the Theorem is a direct result of Lemma \ref{lem:7}.

\subsubsection{Proof of Lemma \ref{lem:poisson1}}\label{ssec:proofpoisson1}

Since 
\begin{eqnarray*}
\ld_i(\bm{\beta}) &=& f'(\bm{x}_i^\top \bm{\beta})- y_i f'(\bm{x}_i^\top \bm{\beta})/f(\bm{x}_i^\top \bm{\beta}), \\ 
\ldd_i(\bm{\beta}) &=& f''(\bm{x}_i^\top \bm{\beta})- y_i (f'/f)'(\bm{x}_i^\top \bm{\beta}), 
\\ 
\lddd_i(\bm{\beta}) &=&  f'''(\bm{x}_i^\top \bm{\beta})- y_i (f'/f)''(\bm{x}_i^\top \bm{\beta}), 
\end{eqnarray*}
where
\begin{eqnarray}
f'(z)=\frac{e^z}{1+e^z}, \quad f''(z)=\frac{e^z}{(1+e^z)^2}, \quad f'''(z)=\frac{e^z(1-e^z)}{(1+e^z)^3}. \label{eq:f}
\end{eqnarray}
Concerning equations \eqref{eq:f}, the following inequalities hold
\begin{eqnarray*}
f'(z) \leq 1, \quad f''(z)\leq 1/4, \quad f'''(z)\leq 1/10.
\end{eqnarray*}
For any  $x>0$, consider the function
\begin{eqnarray*}
h(x):= \frac{x}{(1+x)\log(1+x)}.
\end{eqnarray*}
It is straightforward to check that $h(x)$ is a decreasing function of $x>0$ and that $\lim_{x \rightarrow 0} h(x) =1$. Hence, by simply using $x=e^z$, for any $z$, we have 
\begin{eqnarray*}
f'(z)/f(z) \leq 1 \quad \mbox{ leading to } \quad \| \ld(\bm{\beta}) \|_{\infty}  \leq 1 +\| \bm{y} \|_{\infty}.
\end{eqnarray*}
Moreover, 
\begin{eqnarray*}
f''/f &=& f'/f \times 1/(1+e^z) \leq 1\\
f'''/f &=&  f''/f \times (1-e^z)/(1+e^z)\leq 1.
\end{eqnarray*}
Since 
\begin{eqnarray*}
(f'/f)'' &=& f'''/f + 2 f'^3/f^3 - 3 f' f''/f^2 \quad \mbox{ leading to } \quad  |(f'/f)''| \leq 6.
\end{eqnarray*}
Therefore, $\| \lddd(\bm{\beta}) \|_{\infty}  \leq 1 + 6 \| \bm{y} \|_{\infty}$, leading to
\begin{eqnarray*}
 \|  \bm{\ldd}_{/i}( \bm{\beta + \delta}) - \bm{\ldd}_{/ i}(\bm{\beta}  ) \|_2 &\leq&  \|  \bm{\ldd}( \bm{\beta+\delta}) - \bm{\ldd}(\bm{\beta}  ) \|_2 = \sqrt{\sum_i \left(\ldd(\beta_i+\delta_i) - \ldd(\beta_i) \right)^2 } \\
 &=&\sqrt{\sum_i \lddd(\beta_i+\epsilon_i)^2 \delta_i^2 } \quad \mbox{using the mean-value Theorem where} \quad \epsilon_i \in [0,\delta_i] \\
 &\leq& (1 + 6 \| \bm{y} \|_{\infty}) \sqrt{\sigma_{\max} (\bm{X}^\top \bm{X})}\| \bm{\delta} \|_2.
\end{eqnarray*}
Finally, based on the inequality above, we have
\begin{eqnarray*}
 \sup_{t \in [0,1]} \frac{  \|  \bm{\ldd}_{/ i}( (1-t) \bli + t \bl ) - \bm{\ldd}_{/ i}(\bl  ) \|_2}{ \|   \bli - \bl   \|_2} &\leq& (1 + 6 \| \bm{y} \|_{\infty}) \sqrt{ \sigma_{\max} (\bm{X}^\top \bm{X})} \sup_{t \in [0,1]} \frac{  (1-t) \| \bli -\bl  \|_2}{ \|   \bli - \bl   \|_2} \nonumber \\
 &\leq& (1 + 6 \| \bm{y} \|_{\infty} ) \sqrt{ \sigma_{\max} (\bm{X}^\top \bm{X})}.
\end{eqnarray*}

\subsubsection{Proof of Lemma \ref{lem:poisson2}}\label{ssec:proofpoisson2}
 We prove Lemma \ref{lem:poisson2} using the following inequality (for large enough $n$ such that $\sqrt{\tilde{c}} \log^{3/2}n > 1$) 
\begin{eqnarray*}
\Pr \left( (1 + 6 \| \bm{y} \|_{\infty}) \sqrt{\sigma_{max} ({\bm{X}^\top \bm{X}})} \geq \zeta_1 \log^{3/2} n  \right)  &\leq&
\Pr \left( \|\bm{y}\|_\infty \geq (2 \log n) \sqrt{9\tilde{c} \log n}  \right) \\
&+&
\Pr \left( \sigma_{\max} (\bm{X}^\top \bm{X}) \geq c (1+  \frac{3}{\sqrt{\delta_0}})^2  \right),
\end{eqnarray*}
where $\zeta_1=37 \sqrt{c \tilde{c}} \Big(1+ \frac{3}{\sqrt{\delta_0}} \Big).$ Note that according to Lemma \ref{lem:7}, we have 
\[
\Pr \left[\sigma_{\max} (\bm{X}^\top \bm{X}) \geq c (1+  \frac{3}{\sqrt{\delta_0}})^2\right] \leq {\rm e}^{-p}. 
\]
In the next step, we bound $\|\bm{y}\|_\infty$. We have
\begin{eqnarray}
\Pr (\|\bm{y}\|_\infty \geq t \ | \ \bm{X}) \leq \sum_{i=1}^n \Pr(y_i \geq t \ | \ \bm{X}).
\end{eqnarray}
For $t> \|\bm{\lambda}\|_\infty$, set $\gamma_i = \log \left( \frac{t}{\lambda_i} \right)$. We have
\[
\Pr (y_i \geq t  \ | \ \bm{X}) = \Pr ({\rm e}^{\gamma_i y_i} \geq {\rm e}^{\gamma_i t} \ | \bm{X} ) \overset{(a)}{\leq} {\rm e}^{- \gamma_i t} \mathbb{E} ({\rm e}^{\gamma_i y_i}) \overset{(b)}{=}  {\rm e}^{- \gamma_i t} {\rm e}^{\lambda_i({\rm e}^\gamma_i -1)} = {\rm e}^{-t \log \frac{t}{\lambda_i} +t - \lambda_i},
\]
where (a) is a result of Markov's inequality, and (b) uses the formula for the moment generating function of a Poisson random variable. It is straightforward to see that ${\rm e}^{-t \log \frac{t}{\lambda_i} +t - \lambda_i} \leq {\rm e}^{-t \log \frac{t}{\|\bm{\lambda}\|_\infty} +t - \|\bm{\lambda\|_\infty}}$. Hence,
\[
\Pr (\|\bm{y}\|_\infty \geq t \ | \ \bm{X}) \leq n {\rm e}^{-t \log \frac{t}{\|\bm{\lambda}\|_\infty} +t - \|\bm{\lambda\|_\infty}}. 
\]
If we set $t =2 \|\bm{\lambda}\|_\infty \log n$, then (for large enough $n$ such that $\log \log n >2$) we will have
\[
\Pr (\|\bm{y}\|_\infty \geq 2 \|\bm{\lambda}\|_\infty \log n  \ | \ \bm{X}) \leq n^{1- \|\bm{\lambda}\|_\infty \log\log n }. 
\]
Define the event
\[
\mathcal{E}_p = \{ 1 \leq \|\bm{\lambda} \|_\infty \leq \sqrt{9 \tilde{c} \log n} \},
\]
and the set 
\[
\mathcal{X}_p = \{ \bm{X}\ | \  1 \leq \|\bm{\lambda} \|_\infty \leq \sqrt{9 \tilde{c} \log n} \}.
\]
First note that 
\[
\lambda_i = \log (1+ {\rm e}^{\bm{x}_i^\top \bm{\beta^*}}) \leq \log (1+ {\rm e}^{|\bm{x}_i^\top \bm{\beta^*}| }) \leq \log ({\rm e}^{|\bm{x}_i^\top \bm{\beta^*}| }+ {\rm e}^{|\bm{x}_i^\top \bm{\beta^*}| }) \leq \log 2 + |\bm{x}_i^ \top \bm{\beta^*}|. 
\]
If we define the event
\[
\tilde{\mathcal{E}}_p = \{ \log 2+ \max_i  |\bm{x}_i^ \top \bm{\beta^*}| \leq \sqrt{9 \tilde{c} \log n} \}  \cap \{ \max_i  |\bm{x}_i^ \top \bm{\beta^*}| \geq 1\},
\]
then $\tilde{\mathcal{E}}_p \subset \mathcal{E}_p$, leading to 
\begin{eqnarray*}
\Pr \left( \mathcal{E}_p^c \right) &\leq& \Pr \left( \tilde{ \mathcal{E}}_p^c \right)
 \leq \Pr \left( \{ \log 2+ \max_i  |\bm{x}_i^ \top \bm{\beta^*}| \leq \sqrt{9 \tilde{c} \log n} \}^{c}  \right)+\Pr (\{ \max_i  |\bm{x}_i^ \top \bm{\beta^*}| \geq 1\}^c).
\end{eqnarray*}
Note that (for large enough $n$)
\begin{eqnarray}
\lefteqn{\Pr \left( \{ \log 2+ \max_i  |\bm{x}_i^ \top \bm{\beta^*}| \leq \sqrt{9 \tilde{c} \log n} \}^{c}  \right)
\leq n \Pr \left(   |\bm{x}_i^ \top \bm{\beta^*}| \geq \sqrt{9 \tilde{c} \log n} -1   \right) } \nonumber \\
&\leq& n \Pr \left(   |\bm{x}_i^ \top \bm{\beta^*}| \geq 2\sqrt{ \tilde{c} \log n}    \right) \overset{(c)}{\leq} 2 n {\rm e}^{-2 \log n} \leq \frac{2}{n},
\end{eqnarray}
where (c) is a direct consequence of the Gaussian tail bound in Lemma \ref{lem:GaussianTail} and the fact that $\bm{x}_i^ \top \bm{\beta^*} \sim \N(0, \tilde{c})$. Furthermore, if $Z \sim N(0, \tilde{c})$, then
\begin{eqnarray}
\Pr (\{ \max_i  |\bm{x}_i^ \top \bm{\beta^*}| \geq 1\}^c) = \left(  \Pr(|\bm{x}_i^ \top \bm{\beta^*}| < 1) \right)^n = (P(Z < 1))^n. 
\end{eqnarray}

We now have
\begin{eqnarray}
\lefteqn{\Pr (\|\bm{y}\|_\infty \geq (2 \log n) \sqrt{9\tilde{c} \log n} ) } \nonumber \\
&=& \int_{\bm{X} \in \mathcal{X}_p} \Pr (\|\bm{y}\|_\infty \geq (2 \log n) \sqrt{9\tilde{c} \log n}  \ | \  \bm{X}) dp_{\bm{X}} +  \int_{\bm{X} \in \mathcal{X}^c_p} \Pr (\|\bm{y}\|_\infty \geq (2 \log n) \sqrt{9 \tilde{c} \log n}  \ | \  \bm{X}) dp_{\bm{X}} \nonumber \\
&\leq&  \int_{\bm{X} \in \mathcal{X}_p} \Pr (\|\bm{y}\|_\infty \geq 2 \log n \|\bm{\lambda}\|_\infty  \ | \  \bm{X}) dp_{\bm{X}} + \Pr(\mathcal{E}^c_p)  \nonumber \\
&\leq&  \int_{\bm{X} \in \mathcal{X}_p}  n^{1- \|\bm{\lambda}\|_\infty \log\log n } dp_{\bm{X}} + \Pr \left( \tilde{ \mathcal{E}}_p^c \right)  \\
&\leq&  \int_{\bm{X} \in \mathcal{X}_p}  n^{1- \log\log n } dp_{\bm{X}} + \Pr \left( \{ \log 2+ \max_i  |\bm{x}_i^ \top \bm{\beta^*}| \leq \sqrt{9 \tilde{c} \log n} \}^{c}  \right)+\Pr (\{ \max_i  |\bm{x}_i^ \top \bm{\beta^*}| \geq 1\}^c)
  \nonumber \\
&\leq& n^{1- \log\log n } + \frac{2}{n} + {\rm e}^{-n \log (\frac{1}{P(Z \leq 1)}) }, \nonumber
\end{eqnarray}
where $Z \sim N(0, \tilde{c})$.

\subsubsection{Proof of Lemma \ref{lem:psudoHuber}} \label{ssec:proofpseudoHuber}

Note that
\begin{eqnarray}
f'(a) &=& \frac{a}{\sqrt{1+ \left(\frac{a}{\gamma} \right)^2}} \leq \gamma, \nonumber \\
f''(a) &=& \left({1+ \left(\frac{a}{\gamma} \right)^2}\right)^{\frac{3}{2}}, \nonumber \\
f'''(a) &=& \frac{-3a\gamma^3}{(a^2+\gamma^2)^{\frac{5}{2}}}.
\end{eqnarray}
Note that $|f'''(a)| \leq \frac{3}{\gamma}$. To see this consider the following two cases:
\begin{itemize}
\item Case I, $|a|  \leq \gamma$: 
\[
|f'''(a)| \leq \frac{3 \gamma^4}{\gamma^5} \leq \frac{3}{\gamma}.
\]

\item Case II, $|a| > \gamma$:
\[
|f'''(a)| \leq \frac{3 |a| \gamma^3}{|a|^5} \leq \frac{3\gamma^3}{a^4|} \leq \frac{3}{\gamma}.  
\]
\end{itemize}

Therefore,
\begin{eqnarray*}
 \|  \bm{\ldd}_{/i}( \bm{\beta + \delta}) - \bm{\ldd}_{/ i}(\bm{\beta}  ) \|_2 &\leq&  \|  \bm{\ldd}( \bm{\beta+\delta}) - \bm{\ldd}(\bm{\beta}  ) \|_2 = \sqrt{\sum_i \left(\ldd(\beta_i+\delta_i) - \ldd(\beta_i) \right)^2 } \\
 &=&\sqrt{\sum_i \lddd(\beta_i+\epsilon_i)^2 (\bm{x}_i^\top\bm{\delta})^2 } \quad \mbox{using the mean-value Theorem where} \quad \epsilon_i \in [0,\delta_i] \\
 &\leq&\frac{3}{\gamma} \sqrt{\bm{\delta}^\top \bm{X}^\top \bm{X} \bm{\delta}} \leq \frac{3}{\gamma} \sqrt{\sigma_{max}({\bm{X}^\top \bm{X}})} \|\bm{\delta}\|_2. 
\end{eqnarray*}
Finally, based on the inequality above, we have
\begin{eqnarray*}
 \sup_{t \in [0,1]} \frac{  \|  \bm{\ldd}_{/ i}( (1-t) \bli + t \bl ) - \bm{\ldd}_{/ i}(\bl  ) \|_2 }{ \|   \bli - \bl   \|_2} &\leq&  \sup_{t \in [0,1]} \frac{  (1-t) \| \bli -\bl  \|_2 \sqrt{\sigma_{max}({\bm{X}^\top \bm{X}})}}{ \|   \bli - \bl   \|_2} \leq \sqrt{\sigma_{max}({\bm{X}^\top \bm{X}})}.
\end{eqnarray*}

\subsubsection{Proof of Lemma \ref{cor:ridge}} \label{ssec:proofcorridge}

Define
\begin{eqnarray}
\bl = \arg\min_{\bm{\beta}} f(\bm{\beta}) =  \arg\min_{\bm{\beta}}  \sum_{j=1}^n \frac{(y_j - \bm{x}_j^\top \bm{\beta})^2}{{2}} + \lambda \sum_{j=1}^p r(\beta_i), \nonumber \\
\bli = \arg\min_{\bm{\beta}} f_{\slash i}(\bm{\beta}) =  \arg\min_{\bm{\beta}}  \sum_{j=1, j \neq i}^n \frac{(y_j - \bm{x}_j^\top \bm{\beta})^2}{{2}} + \lambda \sum_{j=1}^p r(\beta_i)
\end{eqnarray}
Furthermore, define $r_{0.5} (\beta) = \frac{\gamma}{2} \beta^2 + (1-\gamma) r^\alpha (\beta)$. Since $\bm{y}=\bm{X\beta^*} + \bm{\e}$, where $\bm{\e} \sim \N(0,\bm{I} \sigma_{\e}^2)$, the optimality conditions yield
\begin{eqnarray}
\bm{y}- \bm{X}\bl &=&  (\bm{I}- \bm{X}(\bm{X}^\top \bm{X}+ \frac{\lambda \gamma}{2} \bm{I})^{-1} \bm{X}^{\top})\bm{y} {+} \lambda   \bm{X}(\bm{X}^\top \bm{X}+ \frac{\lambda \gamma}{2} \bm{I})^{-1}\bm{\rd}_{0.5}(\bl) \nonumber  \\
&=&  (\bm{I}- \bm{X}(\bm{X}^\top \bm{X}+ \frac{\lambda \gamma}{2} \bm{I})^{-1} \bm{X}^{\top})\bm{X \beta^*} +  (\bm{I}- \bm{X}(\bm{X}^\top \bm{X}+ \frac{\lambda \gamma}{2} \bm{I})^{-1} \bm{X}^{\top})\bm{\epsilon} \nonumber \\
&& + {\lambda}   \bm{X}(\bm{X}^\top \bm{X}+ \frac{\lambda \gamma}{2} \bm{I})^{-1}\bm{\rd}_{0.5}(\bl).\label{eq:three}
\end{eqnarray}
We bound $\|\bm{y}- \bm{X}\bl \|_\infty$ and complete the proof of Lemma \ref{cor:ridge}, by separately bounding the infinity norm of each of the three terms in \ref{eq:three} using  Lemma \ref{lem:l1}, \ref{lem:l2} and \ref{lem:l3}, and defining 
\begin{eqnarray}  
\tilde \zeta = 2 \sqrt{c \tilde{c}} + 2 \sigma_\epsilon + \lambda \bar{\zeta} \label{eq:zz}, 
\end{eqnarray}
where $\bar{\zeta}$ is introduced in Lemma \ref{lem:l3}.

\begin{lemma} \label{lem:l1}
Under the assumptions of Lemma \ref{cor:ridge} we have 
\[
\Pr ( \|(\bm{I}- \bm{X}(\bm{X}^\top \bm{X}+ \frac{\lambda \gamma}{2} \bm{I})^{-1} \bm{X}^{\top}) \bm{X}\bm{\beta^*} \|_\infty > 2\sqrt{c \tilde{c} { \log n}}) \leq \frac{2}{n}. 
\]

\end{lemma}
\begin{proof}
First note that
\begin{equation}\label{eq:biasfirstexp1}
 (\bm{I}- \bm{X}(\bm{X}^\top \bm{X}+ \frac{\lambda \gamma}{2} \bm{I})^{-1} \bm{X}^{\top}) \bm{X}\bm{\beta^*} =  \frac{\lambda \gamma}{2} \bm{X}(\bm{X}^\top \bm{X}+ \frac{\lambda \gamma}{2} \bm{I})^{-1} \bm{\beta^*}. 
\end{equation}
Define $\bm{D}_i = (\XI^\top \XI +\frac{\lambda \gamma}{2} \bm{I})^{-1}$. According to the matrix inversion lemma we have
\begin{eqnarray}\label{eq:upperbias1}
\bm{x}_i^\top (\bm{X}^\top \bm{X}+ \frac{\lambda \gamma}{2} \bm{I})^{-1} \bm{\beta^*} = \bm{x}_i^\top \bm{D}_i \bm{\beta^*} - \frac{{\bm x}_i^\top \bm{D}_i \bm{x}_i \bm{x}^{\top}_i \bm{D}_i \bm{\beta^*} }{1+ \bm{x}_i^\top \bm{D}_i \bm{x_i} } = \frac{ \bm{x}^{\top}_i \bm{D}_i \bm{\beta^*} }{1+ \bm{x}_i^\top \bm{D}_i \bm{x_i} }.
\end{eqnarray}
Note that conditioned on $\XI$ the distribution of $\bm{x}_i^\top \bm{D}_i \bm{\beta^*}$ is a zero mean Gaussian random variable with variance $v_i =  \| \bm{\Sigma}^{1/2} \bm{D}_i \bm{\beta^*}\|_2^2 \leq \frac{4 \rho_{\max}}{ \lambda^2 \gamma^2} \|\bm{\beta^*}\|_2^2$. Hence, \eqref{eq:upperbias1} and the Gaussian tail bound, i.e. Lemma \ref{lem:GaussianTail}, lead to 
\begin{eqnarray}
\Pr (|\bm{x}_i^\top (\bm{X}^\top \bm{X}+\frac{ \lambda \gamma}{2} \bm{I})^{-1} \bm{\beta^*} | > t \ |  \ \XI) \leq \Pr (|\bm{x}^{\top}_i \bm{D}_i \bm{\beta^*}| > t \ | \ \XI) \leq 2 {\rm e}^{-\frac{t^2}{2 \| \Sigma^{1/2} \bm{D}_i \bm{\beta^*}\|_2^2}} \leq 2 {\rm e}^{- \frac{\lambda^2 \gamma^2 t^2}{8 \rho_{\max} \|\beta^*\|_2^2 }}.
\end{eqnarray}
 Hence, by marginalizing $\XI$, we get
\[
\Pr (|\bm{x}_i^\top (\bm{X}^\top \bm{X}+ \frac{\lambda \gamma}{2} \bm{I})^{-1} \bm{\beta^*} | > t ) \leq 2 {\rm e}^{- \frac{\lambda^2 \gamma^2 t^2}{8 \rho_{\max} \|\beta\|_2^2 }} = 2 {\rm e}^{- \frac{\lambda^2 \gamma^2 t^2}{8 c \tilde{c} }}.
\]
By setting $t= \frac{4\sqrt{c \tilde{c} { \log n}}}{\lambda \gamma}$ we have
\[
\Pr (|\bm{x}_i^\top (\bm{X}^\top \bm{X}+ \frac{\lambda \gamma}{2} \bm{I})^{-1} \bm{\beta^*} | > \frac{4\sqrt{c \tilde{c} { \log n}}}{\lambda \gamma} ) \leq \frac{2}{n^2}. 
\]
This combined with a union bound and \eqref{eq:biasfirstexp1} proves that
\begin{equation*}\label{eq:ridge_case1}
\Pr \left(\| (\bm{I}- \bm{X}(\bm{X}^\top \bm{X}+ \frac{\lambda \gamma}{2}\bm{I})^{-1} \bm{X}^{\top}) \bm{X}\bm{\beta^*} )\|_\infty > 2\sqrt{c \tilde{c}  { \log n}}\right) \leq \frac{2}{n}. 
\end{equation*}

\end{proof}

\begin{lemma}\label{lem:l2}
If $\bm{\e} \sim \N(0,\bm{I} \sigma_{\e}^2)$, then
\[
\Pr \left[ \| (\bm{I}- \bm{X}(\bm{X}^\top \bm{X}+ \frac{\lambda \gamma}{2} \bm{I})^{-1} \bm{X}^{\top}) \bm{\epsilon}\|_\infty \geq 2 \sigma_\epsilon \sqrt{\log n}  \right] \leq \frac{2}{n}.
\]
\end{lemma}
\begin{proof}
Note that conditioned on $\bm{X}$, the distribution of  $\bm{v} = (\bm{I}- \bm{X}(\bm{X}^\top \bm{X}+ \frac{\lambda \gamma}{2} \bm{I})^{-1} \bm{X}^{\top}) \bm{\epsilon}$  is multivariate Gaussian with mean zero and covariance matrix ${ \sigma_{\e}^2}(\bm{I}- \bm{X}(\bm{X}^\top \bm{X}+ \frac{\lambda \gamma}{2} \bm{I})^{-1} \bm{X}^{\top})^2$. We have
 \begin{eqnarray}
 (\bm{I}- \bm{X}(\bm{X}^\top \bm{X}+ \frac{\lambda \gamma}{2}\bm{I})^{-1} \bm{X}^{\top})^2 = \bm{I} -  \bm{X}(\bm{X}^\top \bm{X}+ \frac{\lambda \gamma}{2} \bm{I})^{-1} \bm{X}^{\top}- \frac{\lambda \gamma}{2}  \bm{X}(\bm{X}^\top \bm{X}+ \frac{\lambda \gamma}{2} \bm{I})^{-2} \bm{X}^{\top}.
 \end{eqnarray}
We define $\sigma_i^2(\bm{X})= \left(1 - \bm{x}_i^\top (\bm{X}^\top \bm{X}+ \frac{\lambda \gamma}{2}\bm{I})^{-1}\bm{x}_i - \frac{\lambda \gamma}{2} \bm{x}_i^\top (\bm{X}^\top \bm{X}+ \frac{\lambda \gamma}{2} \bm{I})^{-2}\bm{x}_i\right) \sigma_\epsilon^2$. Clearly $\sigma_i^2 (\bm{X}) \leq \sigma_\epsilon^2$, hence,
\begin{eqnarray}
\Pr (\|\bm{v} \|_\infty > t \ | \ \bm{X}) \leq \sum_{i=1}^n \Pr (|v_i| > t  \ | \ \bm{X}) \leq \sum_{i=1}^n 2 {\rm e}^{-\frac{t^2}{2  \sigma_i^2(X)}}= 2n {\rm e}^{-\frac{t^2}{2  \sigma_{\epsilon}^2}}. 
\end{eqnarray}
Hence, by setting $t = 2 \sigma_\epsilon \sqrt{\log n}$, we have 
\begin{equation*}\label{eq:ridge_case2}
\Pr (\|\bm{v} \|_\infty > t \ | \ \bm{X}) \leq \frac{2}{n}. 
\end{equation*}
\end{proof}

\begin{lemma} \label{lem:l3}
Under the assumptions of Lemma \ref{cor:ridge} we have
\begin{eqnarray}
\Pr \left[ \|\bm{X} (\bm{X}^\top \bm{X}+ \frac{\lambda \gamma}{2} \bm{I})^{-1} \bm{\rd}_{0.5}(\bl)\|_\infty > \bar{\zeta} \sqrt{\log n} \right] \leq \frac{6}{n} + 2n{\rm e}^{-n +1}+n {\rm e}^{-p},
\end{eqnarray}
where 
\begin{eqnarray}\label{eq:zeta_barzeta}
\bar{\zeta} &=&  \frac{5c}{\lambda^2 \gamma \delta_0} \Big(1+ \frac{\alpha (1- \gamma)}{\gamma}\Big) \left( 2\sqrt{ (c \tilde{c} + \sigma_\epsilon^2)} + \sqrt{\frac{10c (c \tilde{c} + \sigma_\epsilon^2) }{\lambda \gamma}}\right)  +  \sqrt{20 \zeta (c \tilde{c} + \sigma_\epsilon^2)}, \nonumber \\
\zeta &=& \frac{2 c}{\lambda^3 \gamma^2} \left(1 + \frac{\alpha (1- \gamma)}{2\gamma}\right).
\end{eqnarray}
\end{lemma}
\begin{proof}
Since $f_{\slash i} (\bli) \leq f_{\slash i} (\bm{0})$, we have
\begin{equation}\label{eq:upperbli}
{2}\lambda \gamma \|\bli\|_2^2 \leq \|\bm{y}_{\slash i}\|_2^2. 
\end{equation}
Furthermore, due to $\rdd_{0.5} (\beta) \leq  \gamma + \frac{\alpha  (1- \gamma)}{2}$,  $\rd_{0.5}(0)=0$ (see Lemma \ref{lem:ralpha:limits}), and  \eqref{eq:upperbli}, we have
\begin{equation}\label{eq:rdBound1}
\| \bm{\rd}_{0.5} (\bli)\|_2^2 \leq  \left( \gamma + \frac{\alpha  (1- \gamma)}{2}\right) \|\bli\|_2^2 \leq \left(\frac{1}{{2\lambda}} + \frac{\alpha (1- \gamma)}{{4\lambda} \gamma}\right) \|\bm{y}_{\slash i}\|_2^2.
\end{equation}
The first order optimality condition yields
\[
\bm{X}^\top \bm{X} (\bli - \bl) + \lambda \bm{\rd} (\bli) - \lambda \bm{\rd} (\bl) = - \bm{x}_i (y_i - \bm{x}_i^\top \bli).  
\]
Since the minimum eigenvalue of the Hessian of $\bm{r}(\bm{\beta})$ is $2 \gamma$, {therefore the minimum eigenvalue of $\bm{X}^\top \bm{X} + \lambda \diag[\bm{\rdd}(\bm{\beta})]$ (for all $\bm{\beta}$) is greater than $2\lambda \gamma$,} leading to  
\[
\|\bli- \bl\|_2 \leq \frac{|y_i -  \bm{x}_i^\top \bli|}{{ 2 }\lambda \gamma} \|\bm{x}_i\|_2. 
\]
This together with $\rdd_{0.5} (\beta) \leq  \gamma + \frac{\alpha  (1- \gamma)}{2}$ yields 
\[
\| \bm{\rd}_{0.5}(\bli) - \bm{\rd}_{0.5} (\bl)\|_2 \leq \left( \gamma+ \frac{\alpha  (1- \gamma)}{2}\right) \|\bli - \bl\|_2 \leq \left(\frac{1}{{2\lambda}} + \frac{\alpha (1- \gamma)}{{4\lambda} \gamma}\right) |y_i -  \bm{x}_i^\top \bli| \|\bm{x}_i\|_2.
\]
Define $\bm{D}_i = (\XI^\top \XI + \frac{\lambda \gamma}{2} \bm{I})^{-1}$. According to the matrix inversion lemma we have 
\begin{eqnarray}\label{eq:secondmainterm1}
\bm{x}_i^\top (\bm{X}^\top \bm{X}+ \frac{\lambda \gamma}{2} \bm{I})^{-1} \bm{\rd}_{0.5}(\bl) = \bm{x}_i^\top \bm{D}_i \bm{\rd}_{0.5} (\bl) - \frac{\bm{x}^\top_i \bm{D}_i \bm{x}_i \bm{x}_i^\top \bm{D}_i \bm{\rd}{}_{0.5} (\bl)  }{1+ \bm{x}_i^\top \bm{D}_i \bm{x}_i} = \frac{ \bm{x}_i^\top \bm{D}_i \bm{\rd}_{0.5} (\bl)}{1+ \bm{x}_i^\top \bm{D}_i \bm{x}_i}. 
\end{eqnarray}
Furthermore, we have
\begin{equation}\label{eq:x_iD_rdbreak}
 |\bm{x}_i^\top \bm{D}_i \bm{\rd}_{0.5} (\bl)| \leq |\bm{x}_i^\top \bm{D}_i \bm{\rd}_{0.5}(\bli)|+ |\bm{x}_i^\top \bm{D}_i (\bm{\rd}_{0.5}(\bl)- \bm{\rd}_{0.5}(\bli))|.
\end{equation}
First note that, since the maximum eigenvalue of $\bm{D}_i $ is $\frac{\lambda \gamma}{2}$ we have
\begin{eqnarray}\label{eq:DiTwoDifferenceUpper1}
\lefteqn{ |\bm{x}_i^\top \bm{D}_i (\bm{\rd}_{0.5}(\bl)- \bm{\rd}_{0.5}(\bli))|} \nonumber \\&\leq& \frac{2}{\lambda \gamma} \|\bm{x}_i \|_2 \|\bm{\rd}_{0.5}(\bl)- \bm{\rd}_{0.5}(\bli)\|_2 
 \leq  \frac{1}{\lambda^{ 2} \gamma} \|\bm{x}_i \|^2_2 \left(1 + \frac{\alpha (1- \gamma)}{2 \gamma}\right)  |y_i -  \bm{x}_i^\top \bli| \nonumber \\
 &\leq& \frac{1}{\lambda^2 \gamma} \left(1 + \frac{\alpha (1- \gamma)}{ 2\gamma} \right)   \|\bm{x}_i \|^2_2 (|y_i| + |\bm{x}_i^\top \bli| ).
\end{eqnarray}
Furthermore, we have
\begin{enumerate}
\item Due to Lemma \ref{lem:massart}, $\Pr (\|\bm{x}_i\|_2^2 > 5 p \rho_{\max}) \leq {\rm e}^{-p}$, leading to
\begin{equation}\label{eq:boundx_i}
\Pr (\|\bm{x}_i\|_2^2 > \frac{5c}{\delta_0} ) \leq {\rm e}^{-p}. 
\end{equation}

\item Note that $y_i \sim N(0, \bm{\beta}^\top \bm{\Sigma} \bm{\beta} + \sigma_\epsilon^2)$. Furthermore, $\bm{\beta}^\top \bm{\Sigma} \bm{\beta} + \sigma_\epsilon^2 \leq \rho_{\max} \bm{\beta}^\top \bm{\beta}+ \sigma_\epsilon^2 \leq c \tilde{c} + \sigma_\epsilon^2$. Hence, using the Gaussian tail bound, i.e. Lemma \ref{lem:GaussianTail}, we have  
\begin{equation}
\Pr(|y_i|> t) \leq 2 {\rm e}^{- \frac{t^2}{2 (c \tilde{c} + \sigma_\epsilon^2)}}. 
\end{equation}
Hence,
\begin{equation}\label{eq:upperyinfty}
\Pr(|y_i|> 2 \sqrt{(c\tilde{c} + \sigma_\epsilon^2) \log n} ) \leq \frac{2}{n^2}. 
\end{equation}
\item Given  $\XI, \bm{y}_{\slash i}$, the distribution of $\bm{x}_i^\top \bli$ is $N(0, \bli^\top \bm{\Sigma} \bli)$. Furthermore, $\bli^\top \bm{\Sigma} \bli \leq \frac{c \bli^\top \bli}{n} \leq \frac{c \|\bm{y}_{\slash i}\|_2^2}{{\color{magenta}2} n\lambda \gamma}$, where the last inequality is due to \eqref{eq:upperbli}. Hence, we have
\begin{eqnarray}\label{boundx_ibli}
 \Pr (| \bm{x}_i^\top \bli| > t | \XI, \bm{y}_{\slash i} ) \leq 2 {\rm e}^{- \frac{n \lambda \gamma t^2}{ c \|\bm{y}_{\slash i}\|_2^2}}.  
 \end{eqnarray}
According to Lemma  \ref{lem:massart} since $y_i \overset{i.i.d.}{\sim} N(0, \bm{\beta}^\top \bm{\Sigma} \bm{\beta} + \sigma_\epsilon^2)$, and $\bm{\beta}^\top  \bm{\Sigma} \bm{\beta} + \sigma_\epsilon^2 \leq c \tilde{c}+ \sigma_\epsilon^2$,  we have
\begin{equation}\label{eq:boundonnormy}
\Pr ( \|\bm{y}_{\slash i}\|_2^2 > 5 (n-1) (c\tilde{c} + \sigma_\epsilon^2)) \leq {\rm e}^{-n+1}. 
\end{equation}
Let $B$ denote the event that $\|\bm{y}_{\slash i}\|_2^2 \leq 5 (n-1) (c\tilde{c} + \sigma_\epsilon^2)$. Then, combining \eqref{boundx_ibli} and \eqref{eq:boundonnormy}, we have 
\[
 \Pr (| \bm{x}_i^\top \bli| > t) \leq  \Pr (| \bm{x}_i^\top \bli| > t \ | \ B) + \Pr(B^c) \leq 2 {\rm e}^{-\frac{\lambda \gamma t^2}{5c (c \tilde{c} + \sigma_\epsilon^2)}}+ {\rm e}^{-n+1}. 
\]
Hence,
\begin{equation}\label{eq:xblithirdterm}
 \Pr \left[ | \bm{x}_i^\top \bli| > \sqrt{\frac{10 c(c \tilde{c}+ \sigma_\epsilon^2 ) }{\lambda \gamma} \log n} \right] \leq \frac{2}{n^2}+{\rm e}^{-n+1}. 
\end{equation}

\end{enumerate}
By combining \eqref{eq:boundx_i}, \eqref{eq:upperyinfty}, \eqref{eq:xblithirdterm}, and \eqref{eq:DiTwoDifferenceUpper1} we conclude that

\begin{eqnarray}\label{eq:x_iDidiff1}
&\Pr  & \left[ |\bm{x}_i^\top D_i (\bm{\rd}_{0.5}(\bl)- \bm{\rd}_{0.5}(\bli))| > \frac{5c}{\lambda^{ 2} \gamma \delta_0} \Big(1+ \frac{\alpha (1- \gamma)}{{ 2}\gamma}\Big) \left( 2\sqrt{ (c \tilde{c} + \sigma_\epsilon^2)} + \sqrt{\frac{10c (c \tilde{c} + \sigma_\epsilon^2) }{\lambda \gamma}} \right)  \sqrt{ \log n}\right] \nonumber \\
&\leq & \frac{4}{n^2}+ {\rm e}^{-n +1} + {\rm e}^{-p}. 
\end{eqnarray}

Next, we compute an upper bound on $ |\bm{x}_i^\top \bm{D}_i \bm{\rd}_{0.5}(\bli)|$. Since $\bm{x}_i$ is independent of $\bm{y}_{\slash i}$ and $\XI$, we conclude that given  $\XI$ and $\bm{y}_{\slash i}$, $\bm{x}_i^\top \bm{D}_i \bm{\rd}_{0.5}(\bli)$ is a Gaussian random variable with mean zero and variance 
$$\|\bm{\Sigma}^{1/2} \bm{D}_i \bm{\rd}_{0.5}(\bli)\|^2_2 \leq \frac{4 \rho_{\max}}{\lambda^2 \gamma^2}\| \bm{\rd}_{0.5}(\bli)\|^2_2 \leq  \frac{{2} \rho_{\max}}{\lambda^{3} \gamma^2} \left(1 + \frac{\alpha (1- \gamma)}{2\gamma}\right) \|\bm{y}_{\slash i}\|_2^2 = \frac{\zeta \|\bm{y}_{\slash i}\|_2^2}{n},$$
where $\zeta = \frac{2 c}{\lambda^3 \gamma^2} \left(1 + \frac{\alpha (1- \gamma)}{2\gamma}\right)$, and the second inequality is due to \eqref{eq:rdBound1}. Hence,
\[
\mathbb{P} (|\bm{x}_i^\top \bm{D}_i \bm{\rd}_{0.5}(\bli)| > t  \ | \ \XI, \bm{y}_{\slash i} ) \leq 2 {\rm e}^{-\frac{nt^2}{2\zeta \|\bm{y}_{\slash i}\|_2^2}}. 
\]
Considering the event $B$ of $\|\bm{y}_{\slash i}\|_2^2 \leq 5 (n-1) (c\tilde{c} + \sigma_\epsilon^2)$, we have
\begin{eqnarray}
\Pr \left (|\bm{x}_i^\top \bm{D}_i \bm{\rd}_{0.5}(\bli)| > t \right) &\leq& \Pr \left (|\bm{x}_i^\top \bm{D}_i \bm{\rd}_{0.5}(\bli)| > t  \big | B \right)  + \Pr(B^c) \leq
2 {\rm e}^{-\frac{t^2}{10\zeta (c\tilde{c} + \sigma_\epsilon^2)}} +   {\rm e}^{-n+1}.
\end{eqnarray}
Hence,
\begin{eqnarray}\label{eq:xiDirdbli1}
\Pr \left (|\bm{x}_i^\top \bm{D}_i \bm{\rd}_{0.5}(\bli)| >   \sqrt{ 20 \zeta (c \tilde{c} + \sigma_\epsilon^2) \log n} \right) 
& \leq& \frac{2}{n^2}+   {\rm e}^{-n+1}. 
\end{eqnarray}
By combining \eqref{eq:secondmainterm1}, \eqref{eq:x_iD_rdbreak}, \eqref{eq:x_iDidiff1}, and \ref{eq:xiDirdbli1} we conclude that if 
\[
\bar{\zeta} = \frac{5c}{\lambda^2 \gamma \delta_0} \Big(1+ \frac{\alpha (1- \gamma)}{\gamma}\Big) \left( 2\sqrt{ (c \tilde{c} + \sigma_\epsilon^2)} + \sqrt{\frac{10c (c \tilde{c} + \sigma_\epsilon^2) }{\lambda \gamma}}\right)  +  \sqrt{20 \zeta (c \tilde{c} + \sigma_\epsilon^2)},
\]
then
\begin{eqnarray}
\lefteqn{\Pr \left[ |\bm{x}_i^\top (\bm{X}^\top \bm{X}+ \frac{\lambda \gamma}{2} \bm{I})^{-1} \bm{\rd}_{0.5}(\bl)| > \bar{\zeta} \sqrt{\log n} \right]} \nonumber\\ 
&\leq& \Pr \left[ |\bm{x}_i^\top {\bm D}_i \bm{\rd}_{0.5} (\bl) |\geq   \bar{\zeta} \sqrt{\log n} \right] \leq \frac{6}{n^2} + 2{\rm e}^{-n +1}+ {\rm e}^{-p}.
\end{eqnarray}
Hence,
\[
\Pr \left[ \|\bm{X} (\bm{X}^\top \bm{X}+ \frac{\lambda \gamma}{2} \bm{I})^{-1} \bm{\rd}_{0.5}(\bl\|_\infty| > \bar{\zeta} \sqrt{\log n} \right] \leq \frac{6}{n} + 2n{\rm e}^{-n +1}+n  {\rm e}^{-p}.
\]
\end{proof}

\subsubsection{Proof of Lemma \ref{lem:smoothLASSO}}\label{ssec:lem:smoothLASSO1}

Since $r^\alpha(z)=\alpha^{-1} \log \left(e^{\alpha z} + e^{-\alpha z} + 2 \right)$, we have $e^{\alpha r^\alpha(z)} = e^{\alpha z} + e^{-\alpha z} + 2$, and because of Lemma \ref{lem:1:uniform}, $ e^{\alpha z} + e^{-\alpha z} + 2 \geq e^{\alpha |z|}$. Moreover, $\rddd^{\alpha}(z) = 2 \alpha^2 (e^{-\alpha z} - e^{\alpha z})/(e^{\alpha z} + e^{-\alpha z} + 2)^2 \leq 2 \alpha^2 (e^{-\alpha z} - e^{\alpha z})/e^{2 \alpha |z|} \leq 4 \alpha^2 e^{-\alpha |z|}$. The next step is
\begin{eqnarray*}
\frac{ \|  \bm{\rdd}^\alpha(  \bm{\beta + \delta} ) - \bm{\rdd}^\alpha(\bm{\beta}) \|_2}{\| \bm{\delta} \|_2} &=& \frac{\sqrt{\sum_i \left(\rdd^\alpha(\beta_i+\delta_i) - \rdd^\alpha(\beta_i) \right)^2 } }{\| \bm{\delta} \|_2}  \\
 &=&\frac{\sqrt{\sum_i \rddd^\alpha(\beta_i+\epsilon_i)^2 \delta_i^2 } }{\| \bm{\delta} \|_2}  \quad \mbox{using the mean-value Theorem where} \quad \epsilon_i \in [0,\delta_i] \\
  &=&\frac{4\alpha^2 \sqrt{\sum_i  \delta_i^2 e^{-2\alpha |\beta_i+\epsilon_i|} }  }{\| \bm{\delta} \|_2} 
  \\
 &\leq&4\alpha^2.
\end{eqnarray*}

\subsection{Proof of Theorem  \ref{theo:main}} \label{sec:theo-1}
We first present lemmas necessary for the proof of Theorem \ref{theo:main}. Lemmas are proved in section \ref{sec:lem:proof}.

\begin{lemma} \label{lem:5}
Let $\bm{X} \in \R^{m \times p}$ be a matrix with $m > p =\text{rank}(\bm{X})$. Moreover, let $\bm{D} \in \R^{m \times m}$ and $\bm{D+\Gamma} \in \R^{m \times m}$ be  diagonal matrices with positive elements, then
\begin{eqnarray*}
\lefteqn{\left ( \bm{X}^\top \bm{D X} \right)^{-1} - \left ( \bm{X}^\top (\bm{D+\Gamma}) \bm{X}  \right)^{-1}}\nonumber \\ &=& \bm{A}^{-1}  \bm{X}^\top  \bm{\Gamma}  \bm{X A}^{-1}-
\bm{A}^{-1} \bm{X}^\top \bm{\Gamma X}   \left( \bm{X}^\top (\bm{D} + \bm{\Gamma}) \bm{X} \right)^{-1} \bm{X}^\top \bm{\Gamma   X A}^{-1},
\end{eqnarray*}
where $\bm{A} \triangleq \bm{X}^\top \bm{D X}  $.
\end{lemma}
\begin{lemma} \label{lem:9}
Assume that $\bm{X}^\top (\bm{D + \Gamma}) \bm{X} $ and $\bm{X}^\top \bm{D X} $ are positive definite, and 
define:
\begin{eqnarray}
\bm{\Gamma} &\triangleq& \diag(\bm{\gamma}), \\
\bar \omega_{\max} &\triangleq& \sigma_{\max}\left( \bm{X   X}^\top\right), \label{lem:9-1}   \\  
\nu_{\min} &\triangleq& \sigma_{\min}\left( \bm{X}^\top (\bm{D} + \bm{\Gamma}) \bm{X}  \right), \label{lem:9-2} \\
\bm{A} &\triangleq& \bm{X}^\top \bm{D X}. 
\end{eqnarray}
Then,
\begin{eqnarray}
\left | \bm{z}^\top \left ( \bm{X}^\top (\bm{D}+\bm{\Gamma}) \bm{X} \right)^{-1} \bm{z}- \bm{z}^\top \left ( \bm{X}^\top\bm{ D X}  \right)^{-1} \bm{z} \right| \leq 
 \left(  \| \bm{\gamma} \|_2   +\left( \frac{\bar \omega_{\max}}{\nu_{\min}} \right) \| \bm{\gamma} \|_4^2    \right) \left \| \bm{X A}^{-1} \bm{z} \right  \|_4^2.
\end{eqnarray}
\end{lemma}
\begin{lemma}\label{lem:3} {Let $S$ denote the event that  \eqref{eq:c1def}, \eqref{eq:c2def}, \eqref{eq:c3def}, and \eqref{eq:minEigComb} hold.}  If $S$ holds, then
\begin{eqnarray}
\left | \bm{x_i}^\top\bm{ \Delta}^*_{/i} -  \left( \frac{\ld_i(\bl)}{\ldd_i(\bl)} \right) \frac{  H_{ii}}{1 -  H_{ii}}  \right | \leq \bar C_i \left ( \left \|   \XI \bm{J}_{/i}^{-1}  \bm{x_i}  \right\|_4^2 + \left \|   \bm{J}_{/i}^{-1}  \bm{x_i}  \right\|_4^2  \right),   \nonumber
\end{eqnarray}
where
\begin{eqnarray*}
\bm{ \Delta}^*_{/ i} &\triangleq& \bli-  \bl, \\
\bm{H} &\triangleq& \bm{X} \left( \lambda \diag[\bm{\rdd}(\bl)] +\bm{X}^\top \diag[\bm{\ldd}(\bl)] \bm{X} \right)^{-1} \bm{X}^\top \diag[\bm{\ldd(\bl)}],
\\
\bm{J}_{/i} &\triangleq& \lambda \diag[\bm{ \rdd}(\bli)] +\XI^\top \diag[\bm{\ldd}_{/ i}(\bli)] \XI,
\\
\bar C_i &\triangleq& 4\left \| \bm{x_i} \right \|_2 \left( \frac{c_1^2(n) c_2(n) }{\nu} \right)  \left( 1 +  \frac{2 c_1(n)c_2(n) (1+\omega_{\max,i})}{\nu^2}  \left \| \bm{x_i} \right \|_2     \right),
\end{eqnarray*}
and $c_1(n)$ and $c_2(n)$ are defined in Assumption \ref{ass:1},  $\nu$ is defined in Assumption \ref{ass:2}, and $\omega_{\max,i} \triangleq \sigma_{\max} \left(\XI \XI^\top \right) $.
\end{lemma}
\begin{lemma} \label{lem:6}
Let $\bm{x} \sim \N (0,\bm{\Sigma})$ with $\rho_{\max} \triangleq \sigma_{\max} \left( \bm{\Sigma} \right)$, where $\bm{\Sigma} \in \R^{p \times p}$ then
\begin{eqnarray}
\Pr \left [   \left \|  \bm{x} \right \|_4^2  >  2(1+c) \rho_{\max} \sqrt{p} \log p \right] &\leq&  \frac{2}{p^c}.
\end{eqnarray}
Moreover, if 
\begin{eqnarray}
\omega_{\max} &\triangleq& \sigma_{\max}\left( \bm{X   X}^\top\right), \label{lem:9-1}   \\  
\nu_{\min} &\triangleq& \sigma_{\min}\left( \bm{J}  \right)\label{lem:9-2},
\end{eqnarray}
where $\bm{x}$ is independent of the symmetric matrix $\bm{J} \in \R^{p \times p}$ and $\bm{X} \in \R^{m \times p}$, then
\begin{eqnarray}
\Pr \left [  \left \|   \bm{J}^{-1} \bm{x} \right \|_4^2  >  2(1+c) \left(\frac{\rho_{\max}}{  \nu_{\min}^2 } \right)  \sqrt{p} \log p \right]&<&\frac{2}{p^c},\\
\Pr \left [  \left \|  \bm{X J}^{-1} \bm{x} \right \|_4^2  >  2(1+c) \left(\frac{\rho_{\max}}{  \nu_{\min}^2 } \omega_{\max}\right)  \sqrt{m} \log m \right]&<&\frac{2}{m^c}.\label{eq:lem:6}
\end{eqnarray}
%
%
\end{lemma}



\begin{proof} [Proof of Theorem \ref{theo:main}]
{Let $S$ denote the event that  \eqref{eq:c1def}, \eqref{eq:c2def}, \eqref{eq:c3def}, and \eqref{eq:minEigComb} hold.}  Furthermore, define the following events:
\begin{eqnarray}
G &\triangleq& \left \{ \max_{1 \leq i \leq n}\left | \bm{x_i}^\top \bm{\Delta}_{/i}^*  -\left( \frac{\ld_i(\bl)}{\ldd_i(\bl)} \right) \frac{  H_{ii}}{1 -  H_{ii}} \right | >  C \frac{\log p}{\sqrt{p} }  \right\}, \label{eq:G} \\
E_i &\triangleq&   \left \{ \left | \bm{x_i}^\top \bm{\Delta}_{/i}^*  -\left( \frac{\ld_i(\bl)}{\ldd_i(\bl)} \right) \frac{  H_{ii}}{1 -  H_{ii}} \right | >  C \frac{\log p}{\sqrt{p} }  \right\}, \label{eq:C} \\
\tilde{E}_i &\triangleq&  \left \{ \bar C_i \left( \left \|   \XI \bm{J}_{/i}^{-1}  \bm{x_i}  \right\|_4^2 + \left \|   \bm{J}_{/i}^{-1}  \bm{x_i}  \right\|_4^2 \right )   >   C \frac{\log p}{\sqrt{p} }  \right\}, \\
F_i &\triangleq&      \left \{ \bar C_i \left( \left \|   \XI \bm{J}_{/i}^{-1}  \bm{x_i}  \right\|_4^2 + \left \|   \bm{J}_{/i}^{-1}  \bm{x_i}  \right\|_4^2 \right )   >  \bar C_i C_i \sqrt{p} \log p  \right\}, \\
K_i &\triangleq& \left \{ \frac{C}{\sqrt{p}} \geq  \bar C_i C_i \sqrt{p} \right \},
\\
W_i &\triangleq& \left \{ \| \bm{x_i} \|_2^2 > 5 p \rho_{\max} \right \}  \cup  \left \{  \omega_{\max}  > \left(\sqrt{n} + 3 \sqrt{p} \right)^2\rho_{\max}  \right \},
\end{eqnarray}
where $C$ in \eqref{eq:C} is a positive constant (defined later in \eqref{eq:CC}), and
\begin{eqnarray}
C_i &\triangleq& 2 (1+c) \left(\frac{\rho_{\max}}{\nu^2} \right) \left(1 + \omega_{\max} \sqrt{\frac{n-1}{p}} \frac{\log (n-1)}{\log p} \right), \label{eq:Ci}\\
\bar C_i &\triangleq& 4 \left \| \bm{x_i} \right \|_2 \left( \frac{c_1^2(n) c_2(n) }{\nu} \right)  \left( 1 +  \frac{ 2c_1(n)c_2(n) (1+\omega_{\max})}{\nu^2}  \left \| \bm{x_i} \right \|_2     \right), \\
\omega_{\max} &\triangleq& \sigma_{\max} \left( \bm{X X}^\top \right).
\end{eqnarray}
The variable $c$ in \eqref{eq:Ci} is later set to 3, but for now all we need to know is that it is a positive constant. Due to Lemma \ref{lem:3}, {if the event $S$ holds}, then for every $i$ we have
$$
\bar C_i \left( \left \|   \XI \bm{J}_{/i}^{-1}  \bm{x_i}  \right\|_4^2 + \left \|   \bm{J}_{/i}^{-1}  \bm{x_i}  \right\|_4^2 \right ) \geq\left | \bm{x_i}^\top \bm{\Delta}_{/i}^*  -\left( \frac{\ld_i(\bl)}{\ldd_i(\bl)} \right) \frac{  H_{ii}}{1 -  H_{ii}} \right |.
$$
{Since $\Pr[S^c] \leq q_n+ \tilde{q}_n$, we have}
\begin{eqnarray}\label{eq:mainEevent}
\Pr[G] &\leq& \Pr [G |  S] + \Pr[S^c] \leq  \Pr \left[ \max_{1 \leq i \leq n} \bar C_i \left( \left \|   \XI \bm{J}_{/i}^{-1}  \bm{x_i}  \right\|_4^2 + \left \|   \bm{J}_{/i}^{-1}  \bm{x_i}  \right\|_4^2 \right ) >  C \frac{\log p}{\sqrt{p} }  \ | \ S \right] + q_n + \tilde{q}_n \nonumber \\
&\leq& \frac{1}{1- q_n- \tilde{q}_n}  \Pr \left[ \max_{1 \leq i \leq n} \bar C_i \left( \left \|   \XI \bm{J}_{/i}^{-1}  \bm{x_i}  \right\|_4^2 + \left \|   \bm{J}_{/i}^{-1}  \bm{x_i}  \right\|_4^2 \right ) >  C \frac{\log p}{\sqrt{p} }  \right] + q_n + \tilde{q}_n \nonumber \\
&\leq& \frac{1}{1- q_n- \tilde{q}_n} \sum_{i=1}^n \Pr[\tilde{E}_i] + q_n + \tilde{q}_n. 
\end{eqnarray}
Hence, we now obtain an upper bound for $\Pr[\tilde{E}_i]$;
\begin{eqnarray}
\Pr[\tilde{E}_i]  &\stackrel{}{\leq}&
 \Pr [\tilde{E}_i | K_i  ] + \Pr [ K_i^c  ], \nonumber
 \\
&{\leq}& \Pr [F_i | K_i]  + \Pr [ K_i^c  ] \leq  \frac{\Pr(F_i)}{\Pr(K_i)} + \Pr [ K_i^c  ] \nonumber \\
&\leq& \frac{\Pr \left [ \left \|   \XI \bm{J}_{/i}^{-1}  \bm{x_i}  \right\|_4^2 > 2 (1+c) \left(\frac{\rho_{\max}   }{\nu^2} \omega_{\max} \right) \sqrt{n-1} \log (n-1) \right] }{\Pr(K_i)}\nonumber \\&+& \frac{\Pr \left [ \left \|    \bm{J}_{/i}^{-1}  \bm{x_i}  \right\|_4^2 > 2 (1+c) \left(\frac{\rho_{\max}}{\nu^2} \right) \sqrt{p} \log p \right] }{\Pr(K_i)} + \Pr [ K_i^c  ]
\nonumber \\
&\stackrel{1}{\leq}&  \left(\frac{2}{(n-1)^c}+ \frac{2}{p^c}\right) \frac{1}{\Pr(K_i)}+ \Pr [ K_i^c  ],
\end{eqnarray}
where $\stackrel{1}{\leq}$ is due to  Inequality \eqref{eq:lem:6} from Lemma \ref{lem:6}. To bound  $\Pr [ K_i^c  ]$ we define
\begin{eqnarray}
C &\triangleq&
32 \sqrt{5 } \left( \frac{c_1^2(n) c_2(n) (p\rho_{\max})^{3/2}}{\nu^3} \right) \left(1 + \left(\sqrt{\frac{n}{p}} + 3 \right)^2p \rho_{\max} \sqrt{\frac{n-1}{p}} \frac{\log (n-1)}{\log p}\right)\nonumber \\ &\times&
 \Bigg( 1 +  \frac{2 c_1(n)c_2(n)\sqrt{5 } \left (1+ \Big(\sqrt{\frac{n}{p}} + 3  \Big)^2 p \rho_{\max}  \right)  \sqrt{ p \rho_{\max}} }{\nu^2 }       \Bigg) \label{eq:CC}
\end{eqnarray}
obtained by setting $c=3$, and computing $p \bar C_i C_i  $ after putting $\sqrt{5 p \rho_{\max}}$ and $ \left(\sqrt{n} + 3 \sqrt{p} \right)^2\rho_{\max}$, bounds in event $W_i$, into  $\left \| \bm{x_i} \right \|_2 $ and $\omega_{\max}$, respectively. Next,
\begin{eqnarray}
\Pr [ K_i^c  ] &=& \Pr \left[ \frac{C}{p} <  \bar C_i C_i \right ] \leq \Pr \left[ C <  p \bar C_i C_i \right  | W_i^c] + \Pr[W_i] \nonumber
\\
&=& \Pr \left [ C < C \right] +\Pr[W_i] = \Pr[W_i]. \nonumber
\end{eqnarray}
The term $\Pr[W_i]$ is exponentially small because $\bm{x_i}$ is $\N(0,\bm{\Sigma})$ with $\rho_{\max}=\sigma_{\max} \left(\bm{\Sigma} \right)$, leading to
\begin{eqnarray}
\Pr[W_i] &\leq& \Pr \left[ \| \bm{x_i} \|_2^2 > 5 p \rho_{\max}   \right] + \Pr \left[  \sigma_{\max} \left( \bm{X X}^\top \right)  > \left(\sqrt{n} + 3 \sqrt{p} \right)^2\rho_{\max}  \right] \leq 2e^{-p} ,
\end{eqnarray}
due to Lemma \ref{lem:massart} and Lemma \ref{lem:7}. In summary, since for $p\geq 1$ we have $\frac{1}{1- {\rm e}^{-p}} < 2$, for $c=3$ we obtain
\[
\Pr[\tilde{E}_i] \leq \frac{4}{(n-1)^3}+\frac{4}{p^3}+2e^{-p}.
\]
This combined with \eqref{eq:mainEevent} leas to 
\begin{eqnarray}
\Pr \left [ \max_{1 \leq i \leq n} \left | \bm{x_i}^\top \bm{\Delta}_{/i}^*  -\left( \frac{\ld_i(\bl)}{\ldd_i(\bl)} \right) \frac{  H_{ii}}{1 -  H_{ii}} \right | >  C \frac{\log p}{\sqrt{p} }  \right ] &\leq & \left(\frac{4n}{(n-1)^3}+\frac{4n}{p^3}+2ne^{-p} \right) \frac{1}{1-q_n - \tilde{q}_n} + q_n + \tilde{q}_n \nonumber \\
&\leq& \frac{8n}{(n-1)^3}+\frac{8n}{p^3}+4ne^{-p}  + q_n + \tilde{q}_n,
\end{eqnarray}
where the last inequality is due to the assumption that $q_n+ \tilde{q}_n \leq 0.5$. Hence, Inequality \eqref{eq:loo_approximation} in Theorem \ref{theo:main} follows. Note that in the presentation of Theorem \ref{theo:main}, we replaced respectively $32\sqrt{5}$ and $2\sqrt{3}$ with the upper-bounds 72 and 5, and we replaced $\sqrt{\frac{n-1}{p}}\frac{\log(n-1)}{\log p}$ with the upper bound $\sqrt{\frac{n}{p}}\frac{\log n}{\log p}$. We also used $\delta_0$ to denote $n/p$.

\end{proof}

\subsection{Proofs of lemmas    \ref{lem:5}, \ref{lem:9}, \ref{lem:3}, \ref{lem:6}  and \ref{lem:7}  \label{sec:lem:proof}}

\begin{proof} [Proof of Lemma \ref{lem:5}]
Let $Q \triangleq \{i: \Gamma_{ii} \neq 0  \}$. Moreover, let $\bm{X}_{Q,:} \in \R^{|Q| \times p}$ stand for the sub-matrix of $\bm{X}$ restricted to the rows indexed by $Q$, and let $\bm{\tilde \Gamma} \in \R^{|Q| \times |Q|}$ be the diagonal matrix with the diagonal elements of $\bm{\Gamma}$ indexed by $Q$. Then, $\bm{X}^\top \bm{\Gamma} \bm{X} = \bm{X}_{Q,:}^\top \bm{\tilde \Gamma} \bm{X}_{Q,:} $, and in turn,  the Woodbury inversion lemma yields
\begin{eqnarray} 
(\bm{X}^\top \bm{D X} + \bm{X}^\top \bm{\Gamma X})^{-1} 
&=& (\bm{A}  + \bm{X}_{Q,:}^\top \bm{\tilde \Gamma} \bm{X}_{Q,:})^{-1} \nonumber \\
&=& \bm{A}^{-1} - {\bm A}^{-1} \bm{X}_{Q,:}^\top (\bm{\tilde \Gamma}^{-1} + \bm{X}_{Q,:} \bm{A}^{-1} \bm{X}^\top_{Q,:} )^{-1} \bm{X}_{Q,:}{\bm A}^{-1}. \label{eq:firstlemmafirststep}
\end{eqnarray}
Using the Woodbury  lemma again we obtain
\begin{eqnarray}
(\bm{\tilde \Gamma}^{-1} + \bm{X}_{Q,:} \bm{A}^{-1} \bm{X}_{Q,:}^\top )^{-1} 
&=& 
\bm{ \tilde \Gamma} - \bm{ \tilde \Gamma} \bm{X}_{Q,:} (\bm{ A+ X}_{Q,:}^\top \bm{\tilde \Gamma X}_{Q,:}    )^{-1} \bm{ X}_{Q,:}^\top \bm{\tilde \Gamma}\nonumber
\\
&=& \bm{ \tilde \Gamma} - \bm{ \tilde \Gamma} \bm{X}_{Q,:} ({\bm  X}^\top (\bm{ \Gamma+ D) X}    )^{-1} \bm{ X}_{Q,:}^\top \bm{\tilde \Gamma}.\label{eq:firstlemmasecondstep}
\end{eqnarray}
Hence, by using \eqref{eq:firstlemmafirststep} and \eqref{eq:firstlemmasecondstep} we have
\begin{eqnarray*}
(\bm{X}^\top \bm{D X})^{-1} - (\bm{X}^\top \bm{D X} + \bm{X}^\top \bm{\Gamma X})^{-1} 
=  {\bm A^{-1}} \bm{X}^\top \left( \bm{  \Gamma} - \bm{  \Gamma X} (\bm{  X^\top ( \Gamma+ D) X    })^{-1} {\bm X^\top \Gamma} \right) \bm{X}{\bm A^{-1}}. \ \ \ \ \ \ \ \ \ \ \ \ \ \ \ \ \ \ \ \ \  \ \ \ \ \ 
\end{eqnarray*}
\end{proof}

\begin{proof} [Proof of Lemma \ref{lem:9}]
Let $\bm{A} \triangleq \bm{X}^\top \bm{D X}$, then
\begin{eqnarray}
\lefteqn{\left |\bm{z}^\top \left ( \bm{X}^\top (\bm{D+\Gamma}) \bm{X} \right)^{-1} \bm{z}- \bm{z}^\top \left ( \bm{X}^\top \bm{D X}  \right)^{-1} \bm{z} \right| } \nonumber \\
&\stackrel{1}{=}& 
\left |\bm{z}^\top \left ( \bm{A}^{-1}  \bm{X}^\top  \bm{\Gamma  X A}^{-1}- \bm{A}^{-1} \bm{X}^\top\bm{ \Gamma X }  \left( \bm{X}^\top (\bm{D + \Gamma}) \bm{X} \right)^{-1} \bm{X}^\top \bm{\Gamma   X A}^{-1} \right)  \bm{z} \right| \nonumber
\\
&\stackrel{2}{\leq}& \left |\bm{z}^\top \bm{ A}^{-1}  \bm{X}^\top  \bm{\Gamma  X A}^{-1} \bm{z} \right  | + 
\bm{z}^\top \bm{A}^{-1} \bm{X}^\top \bm{\Gamma X}   \left( \bm{X}^\top (\bm{D + \Gamma}) \bm{X} \right)^{-1}\bm{ X}^\top \bm{\Gamma}   \bm{X A}^{-1} \bm{ z} \nonumber 
\\&\stackrel{3}{\leq}&  \| \bm{\gamma} \|_2   \left \| \bm{X A}^{-1} \bm{z} \right  \|_4^2 + 
\bm{z}^\top\bm{ A}^{-1} \bm{X}^\top \bm{\Gamma X }  \left( \bm{X}^\top (\bm{D + \Gamma}) \bm{X} \right)^{-1} \bm{X}^\top \bm{\Gamma   X A}^{-1}  \bm{z}  \nonumber
\\&\stackrel{4}{\leq}&  \| \bm{\gamma} \|_2   \left \| \bm{X A}^{-1} \bm{z} \right  \|_4^2 +\left( \frac{\bar \omega_{\max}}{\nu_{\min}} \right)
\bm{z}^\top \bm{A}^{-1} \bm{X}^\top \bm{\Gamma}^2   \bm{X A}^{-1} \bm{ z } \nonumber
\\&\stackrel{5}{\leq}&  \| \bm{\gamma} \|_2   \left \| \bm{X A}^{-1} \bm{z} \right  \|_4^2 +\left( \frac{\bar \omega_{\max}}{\nu_{\min}} \right)
\| \bm{\gamma} \|_4^2   \left \| \bm{X A}^{-1} z \right  \|_4^2  \nonumber
\\&=& \left(  \| \bm{\gamma} \|_2   +\left( \frac{\bar \omega_{\max}}{\nu_{\min}} \right)
\| \bm{\gamma} \|_4^2    \right) \left \| \bm{X A}^{-1} \bm{z} \right  \|_4^2, 
\end{eqnarray}
where $\stackrel{1}{=}$ is due to Lemma \ref{lem:5}, $\stackrel{2}{\leq}$ is due to the triangle inequality, and the fact that $\bm{X}^\top (\bm{D} + \bm{\Gamma}) \bm{X} $ is positive definite, and $\stackrel{3}{\leq}$ and $\stackrel{5}{\leq}$ are due to Cauchy-Schwartz inequality:
\begin{eqnarray*}
\bm{x}^\top \diag[\bm{\gamma}] \bm{x} &=& \sum_{i=1}^n x_i^2 \gamma_i \leq \sqrt{ \| \bm{x} \|_4^4 \| \bm{\gamma} \|_2^2 } = \| \bm{x} \|_4^2 \| \bm{\gamma} \|_2.
\end{eqnarray*}
Finally, $\stackrel{4}{\leq}$ is due to \eqref{lem:9-1} and \eqref{lem:9-2}.
\end{proof}

\begin{proof} [Proof of Lemma \ref{lem:3}]
Define the approximate leave-$i$-out perturbation vector as
\begin{eqnarray}
\bm{\hat \Delta}_{/i} &\triangleq& \ld_i(\bl)  [ \bm{J}_{/i} (\bli - \bm{\Delta^*}_{/i})]^{-1}  \bm{x_i},   \label{eq:D1} %
\end{eqnarray}
where the exact leave-$i$-out perturbation vector is given by
\begin{eqnarray}
\bm{\Delta^*}_{/i} &\triangleq& \bli -\bl.
\end{eqnarray}
Woodbury lemma yields:
\begin{eqnarray}
\bm{x_i}^\top \bm{\hat \Delta}_{/i}  &=& \ld_i(\bl) \bm{x_i}^\top\Bigl (\lambda \diag[\bm{ \rdd}(\bli -  \bm{\Delta}_{/i}^* )] +\XI^\top \diag[\ldd_{/ i}(\bli - \bm{ \Delta}_{/i}^* )] \XI  \Bigr )^{-1} \bm{x_i} \nonumber
\\
&=& \ld_i(\bl) \bm{x_i}^\top\Bigl (\lambda \diag[\bm{\rdd}(\bl )] +\XI^\top \diag[\bm{\ldd}_{/ i}(\bl)] \XI  \Bigr )^{-1} \bm{x_i}  \nonumber \\
&=& \ld_i(\bl) \bm{x_i}^\top\Bigl (\lambda \diag[\bm{\rdd}(\bl )] +\bm{X}^\top \diag[\bm{\ldd}(\bl)] \bm{X} - \bm{x_i} \bm{x_i}^\top   \ldd_i(\bl) \Bigr )^{-1} \bm{x_i}  \nonumber
\\&=& \left( \frac{\ld_i(\bl)}{\ldd_i(\bl)} \right) \frac{ \ldd_i(\bl)    \bm{x_i}^\top\Bigl (\lambda \diag[\bm{ \rdd}(\bl )] +\bm{X}^\top \diag[\bm{\ldd}(\bl)] \bm{X}  \Bigr )^{-1} \bm{x_i}    }{1 -  \ldd_i(\bl)    \bm{x_i}^\top\Bigl (\lambda \diag[\bm{\rdd}(\bl )] +\bm{X}^\top \diag[\bm{\ldd}(\bl)] \bm{X} \Bigr )^{-1} \bm{x_i}  } 
\nonumber \\&=& \left( \frac{\ld_i(\bl)}{\ldd_i(\bl)} \right) \frac{  H_{ii}}{1 -  H_{ii}},
\end{eqnarray}
where $\bm{H} \triangleq \bm{X} \left ( \lambda \diag [\bm{ \rdd}(\bm{\bl})] + \bm{X}^\top \diag[\bm{\ldd}(\bl)] \bm{X} \right)^{-1} \bm{X}^\top \diag[\bm{\ldd}(\bl)]$. Define 
\begin{eqnarray}
\bm{f}_{/i} (\bm{\theta}) &\triangleq&  \lambda \bm{ \rd}( \bm{\theta} )+ \XI^\top \bm{\ld}_{/ i}(\bm{\theta}). 
\end{eqnarray}
The leave-one-out estimate, $\bli = \bl + \bm{\Delta^*}_{/i}$,  satisfies $\bm{f}_{/i} (\bm{\Delta^*}_{/ i}) = 0$. The multivariate mean-value Theorem yields
\begin{eqnarray}
0 &=& \bm{f}_{/i} (\bl + \bm{\Delta^*}_{/i})  = \bm{f}_{/i} (\bl) + \left(\int_0^1 \bm{J}_{/i} (\bl + t \bm{\Delta^*}_{/i}) dt \right)\bm{\Delta^*}_{/i}
\end{eqnarray}
where the Jacobean is
\begin{eqnarray}
\bm{J}_{/i} (\bm{\theta}) &=&  \lambda \diag[\bm{ \rdd}(\bm{\theta})]  +  \XI^\top \diag[\bm{\ldd}_{/ i}( \bm{\theta})] \XI.  
\end{eqnarray}
Moreover, $\bl$ satisfies  
\begin{eqnarray*}
0 &=& \lambda\bm{  \rd}(\bl)+ \bm{X}^\top\bm{\ld} (\bl)  = \bm{f}_{/i} (\bl) + \ld_i(\bl) \bm{x_i}.
\end{eqnarray*}
We get
\begin{eqnarray*}
\ld_i(\bl) \bm{x_i} &=&- \left(\int_0^1 \bm{J}_{/i} (\bl + t \bm{\Delta^*}_{/i}) dt \right)\bm{\Delta^*}_{/i},
\end{eqnarray*}
so that
\begin{eqnarray}
\bm{\Delta^*}_{/i} &=&- \ld_i(\bl)   \left( \int_0^1 \bm{J}_{/i} (\bl + t \bm{\Delta^*}_{/i}) dt \right)^{-1} \bm{x_i},   \label{eq:D} 
\end{eqnarray}
leading to the following inequality
\begin{eqnarray}
\left \| \bm{ \Delta}^*_{/ i} \right \|_2 \leq  \left( \frac{|\ld_i(\bl)|}{\nu} \right) \left \| \bm{x_i} \right \|_2, \label{eq:bounded}
\end{eqnarray} 
as a consequence of Assumption \ref{ass:2}. Next, we look at the part of $\bm{\Delta^*}_{i,\lambda}$ dependent on $\bm{x_i}$,  so we rewrite 
\eqref{eq:D} as
\begin{eqnarray}
\bm{\Delta^*}_{/ i} &=&  - \ld_i(\bl)   \left( \int_0^1 \bm{J}_{/i} (\bli - (1-t) \bm{\Delta^*}_{/ i}) dt \right)^{-1} \bm{x_i}.   \label{eq:D2} %
\end{eqnarray}
Let us  rewrite the Jacobean in a more compact form:
\begin{eqnarray}
\bm{J}_{/i}(\bm{\theta}) = \tXI^\top \bm{D}_{/i}(\bm{\theta}) \tXI,
\end{eqnarray}
where
\begin{eqnarray} \label{eq:tXI}
 \tXI \triangleq \begin{bmatrix} \XI \\ \bm{I} 
\end{bmatrix} \in \R^{(n-1+p) \times p},  \ \ \ \ \ 
\bm{D}_{/i}(\bm{\theta}) \triangleq 
\diag\begin{bmatrix} 
\bm{\ldd}_{/ i}(\bm{\theta} )  
\\ 
\lambda \bm{\rdd}(\bm{\theta} )
\end{bmatrix}
\in \R^{(n-1+p) \times (n-1+p)}.
\end{eqnarray}
Define
\begin{eqnarray}
\bm{\gamma}_{\bm{\delta}/i}(\bm{\theta}) &\triangleq&  
\begin{bmatrix} 
\bm{\ldd}_{/ i}(\bm{\theta} + \bm{\delta} )  - \bm{\ldd}_{/ i}(\bm{\theta}  )  
\\ 
\lambda (\bm{\rdd}(\bm{\theta} + \bm{\delta}) - \bm{\rdd}(\bm{\theta} ))
\end{bmatrix}
\end{eqnarray}
so that 
\begin{eqnarray}
\bm{J}_{/i}(\bm{\theta} + \bm{\delta}) &=& \bm{J}_{/i}(\bm{\theta} ) + \tXI^\top \diag\left[\bm{\gamma}_{\bm{\delta}/i}(\bm{\theta} )\right] \tXI   , \label{eq:A}
\end{eqnarray}
Note that  $\bm{J}_{/i}(\bm{\theta + \delta})$ is  positive definite for all $t \in[0,1]$,  $\bm{\theta} = \bli$ and $\bm{\delta}=- (1-t) \bm{\Delta}_{/ i}^*$, due to Assumption \ref{ass:2}. The last  steps of the proof are as follows:
\begin{eqnarray}
\lefteqn{\left | \bm{x_i}^\top \bm{\Delta}^*_{/ i}  - \bm{x_i}^\top \bm{\hat \Delta}_{/ i}  \right | = |\ld_i(\bl)|
\left | \bm{x_i}^\top \left( \int_0^1 \bm{J}_{/i}(\bli - (1-t) \bm{\Delta}_{/ i}^*) dt \right)^{-1} \bm{x_i}   - \bm{x_i}^\top \bm{J}_{/i}^{-1}(\bli - \bm{\Delta}_{/ i}^*) \bm{x_i}  \right  |} \nonumber
\\
&\leq&
|\ld_i(\bl)| \left | \bm{x_i}^\top \left( \int_0^1 \bm{J}_{/i}(\bli - (1-t) \bm{\Delta}_{/ i}^*) dt \right)^{-1} \bm{x_i}   - \bm{x_i}^\top \bm{J}_{/i}^{-1}(\bli) \bm{x_i}  \right  | \nonumber \\ &+&
|\ld_i(\bl)| \left | \bm{x_i}^\top \bm{J}_{/i}^{-1}(\bli ) \bm{x_i}   - \bm{x_i}^\top \bm{J}_{/i}^{-1}(\bli - \bm{\Delta}_{/ i}^*) \bm{x_i}  \right  | \nonumber
\\
&\stackrel{0}{\leq}&|\ld_i(\bl)| \left | \bm{x_i}^\top  \bm{J}_{/i}^{-1}(\bli) \bm{x_i}- \bm{x_i}^\top  \left ( \bm{J}_{/i}(\bli) + \tXI^\top 
 \diag\left[ \int_0^1 \bm{\gamma}_{-(1-t)\bm{\Delta}_{/ i}^* / i}(\bli) dt \right] \tXI   \right)^{-1} \bm{x_i}   \right  | \nonumber
\\
&+&
|\ld_i(\bl)| \left |\bm{x_i}^\top  \bm{J}_{/i}^{-1}(\bli) \bm{x_i} -  \bm{x_i}^\top  \left ( \bm{J}_{/i}(\bli) +  \tXI^\top  \diag\left[\bm{\gamma}_{-\bm{\Delta}_{/ i}^*/i}(\bli) \right]\tXI   \right)^{-1} \bm{x_i}  \right  |  \nonumber
\\
&\stackrel{1}{\leq}&
|\ld_i(\bl)| \left(\left \|  \int_0^1 \bm{\gamma}_{-(1-t)\bm{\Delta}_{/ i}^* / i}(\bli) dt \right \|_2 
+ \left(  \frac{\bar \omega_{\max,i}}{\nu} \right)  \left \|  \int_0^1 \bm{\gamma}_{-(1-t)\bm{\Delta}_{/ i}^* / i}(\bli) dt   \right \|_4^2\right)
\left \| \tXI \bm{J}_{/ i}^{-1}(\bli) \bm{x_i} \right\|_4^2 \nonumber \\
&+&
|\ld_i(\bl)| \left(\left \| \bm{\gamma}_{-\bm{\Delta}_{/ i}^*/i}(\bli)   \right \|_2 
+ \left(  \frac{\bar \omega_{\max,i}}{\nu} \right)  \left \| \bm{\gamma}_{-\bm{\Delta}_{/ i}^*/i}(\bli)    \right \|_4^2\right)
\left \| \tXI \bm{J}_{/ i}^{-1}(\bli) \bm{x_i} \right\|_4^2 \nonumber \\
&\stackrel{2}{\leq}&
|\ld_i(\bl)| \left(\left \|  \int_0^1 \bm{\gamma}_{-(1-t)\bm{\Delta}_{i,\lambda}^* / i}(\bli) dt \right \|_2 
+ \left(  \frac{\bar \omega_{\max,i}}{\nu} \right)  \left \|  \int_0^1 \bm{\gamma}_{-(1-t)\bm{\Delta}_{/ i}^* / i}(\bli) dt   \right \|_2^2\right)
\left \| \tXI \bm{J}_{/ i}^{-1}(\bli) \bm{x_i} \right\|_4^2 \nonumber \\
&+&
|\ld_i(\bl)| \left(\left \| \bm{\gamma}_{-\bm{\Delta}_{i,\lambda}^*/i}(\bli)   \right \|_2 
+ \left(  \frac{\bar \omega_{\max,i}}{\nu} \right)  \left \| \bm{\gamma}_{-\bm{\Delta}_{/ i}^*/i}(\bli)    \right \|_2^2\right)
\left \| \tXI \bm{J}_{/ i}^{-1}(\bli) \bm{x_i} \right\|_4^2 \nonumber \\
&\stackrel{3}{\leq}&  4 c_1(n)c_2(n)  \left \| \bm{\Delta}^*_{/ i} \right \|_2 \left(  1 + 2 c_2(n) \left  \|  \bm{\Delta}^*_{/ i} \right \|_2 \left(  \frac{\bar \omega_{\max,i}}{\nu} \right)    \right) 
\left \| \tXI \bm{J}_{/i}^{-1}(\bli) \bm{x_i} \right\|_4^2 \nonumber \\
&\stackrel{4}{\leq}& 4\left \| \bm{x_i} \right \|_2 \left( \frac{c_1^2(n) c_2(n) }{\nu} \right)  \left( 1 +  \frac{2 c_1(n)c_2(n) \bar\omega_{\max,i}  }{\nu^2}  \left \| \bm{x_i} \right \|_2     \right) 
\left \| \tXI \bm{J}_{/i}^{-1}(\bli) \bm{x_i} \right\|_4^2
\nonumber \\
&\stackrel{5}{\leq}& \underbrace{4\left \| \bm{x_i} \right \|_2 \left( \frac{c_1^2(n) c_2(n) }{\nu} \right)  \left( 1 +  \frac{2c_1(n)c_2(n) (1+  \omega_{\max,i} ) }{\nu^2}  \left \| \bm{x_i} \right \|_2     \right) }_{\triangleq \bar C_i} 
\sqrt{ \left \| \XI \bm{J}_{/i}^{-1}(\bli) \bm{x_i} \right\|_4^4 + \left \|  \bm{J}_{/i}^{-1}(\bli) \bm{x_i} \right\|_4^4}, \nonumber
\end{eqnarray}
where
\begin{itemize} 
\item $\stackrel{0}{\leq}$  is  due  \eqref{eq:A}. 
\item $\stackrel{1}{\leq}$ is due to   Assumption \ref{ass:2}, and Lemma \ref{lem:9}, where $\bar \omega_{\max,i} \triangleq \sigma_{\max} \left( \tXI \tXI^\top \right)$.
\item $\stackrel{2}{\leq}$ is due the fact that for any $\bm{\gamma}$ we have $\| \bm{\gamma} \|_4^2 \leq \| \bm{\gamma} \|_2^2$,
\item $\stackrel{3}{\leq}$ is due to Assumption \ref{ass:1} as illustrated below
\begin{eqnarray*}
\left \| \bm{\gamma}_{-\bm{\Delta}_{/i}^*/i}(\bli)   \right \|_2 &\leq& 
\left \| \bm{\ldd}_{/ i}(\bli-\bm{\Delta}_{/i}^*)  - \bm{\ldd}_{/ i}(\bli  )   \right\|_2 + \left \| \lambda (\bm{\rdd}(\bli -\bm{\Delta}_{/i}^* ) - \bm{\rdd}(\bli ))  \right\|_2 
\\
&\leq& 2 c_2(n) \left \|   \bm{\Delta}_{/i}^* \right \|_2.  
\end{eqnarray*}
Likewise,
\begin{eqnarray}\label{eq:referencefromproof}
\left \|  \int_0^1 \bm{\gamma}_{-(1-t)\bm{\Delta}_{/i}^* / i}(\bli) dt   \right \|_2 &\leq&
 \int_0^1 \left \|  \bm{\gamma}_{-(1-t)\bm{\Delta}_{/i}^* / i}(\bli)   \right \|_2 dt \nonumber
\\ &\leq& 
 \int_0^1 
\left \| \bm{\ldd}_{/ i}(\bli-(1-t)\bm{\Delta}_{/i}^*)  - \bm{\ldd}_{/ i}(\bli  )   \right\|_2  dt \nonumber
\\
&+&
 \int_0^1   \left \| \lambda (\bm{\rdd}(\bli -(1-t)\bm{\Delta}_{/i}^* ) - \bm{\rdd}(\bli ))  \right\|_2  dt \nonumber
\\
&\leq& 2 c_2(n) \left \|   \bm{\Delta}_{/i}^* \right \|_2. 
\end{eqnarray}
Here we should emphasize that this is the main place in which we have used the smoothness of second derivatives of the loss and regularizer in Assumption \ref{ass:1}. \footnote{Note that by checking the derivation, it is clear that we can replace Assumption \ref{ass:1} with the following weaker assumptions: 
\begin{eqnarray}
 c_2(n) &>&  \sup_{t \in [0,1]} \frac{  \|  \bm{\ldd}_{/ i}( (1-t) \bli + t \bl ) - \bm{\ldd}_{/ i}(\bl  ) \|_2}{\left \|   \bli - \bl  \right \|_2^\zeta}  \label{eq:footnote}
 \\
 c_2(n) &>&   \sup_{t \in [0,1]} \frac{  \|  \bm{\rdd}( (1-t) \bli + t \bl ) - \bm{\rdd}(\bl  ) \|_2}{\left \|   \bli - \bl  \right \|_2^\zeta}
\end{eqnarray}
for some $\zeta>0$, and still find an (weaker) upper bound for $\left | \bm{x_i}^\top \bm{\Delta}^*_{/i}  - \bm{x_i}^\top \bm{\hat \Delta}_{/i}  \right |$ that converges to zero as $n, p \rightarrow \infty$. }
\item   $\stackrel{4}{\leq}$ is due to inequality \eqref{eq:bounded}, and Assumption \ref{ass:1}.
\item   $\stackrel{5}{\leq}$ is due to \eqref{eq:tXI}, and
\begin{eqnarray}
\bar \omega_{\max,i} &=& \sigma_{\max} \left( \tXI \tXI^\top \right)=\sigma_{\max} \left( \tXI^\top \tXI \right)= \sigma_{\max} \left( 
 \begin{bmatrix} \XI  \nonumber \\ \bm{I} 
\end{bmatrix} ^\top
 \begin{bmatrix} \XI \\  \bm{I} 
\end{bmatrix} 
\right) \\
&=& \sigma_{\max} \left( \bm{I} +  \XI^\top \XI  \right) \leq 1+  \sigma_{\max} \left( \XI^\top \XI  \right) = 1+  \omega_{\max,i}.\label{eq:12}
\end{eqnarray}
\end{itemize}
The final result follows the basic inequality: $\sqrt{a^2 + b^2} \leq |a| + |b|$.
\end{proof}

\begin{proof}  [Proof of Lemma \ref{lem:6}]
First, we prove  
\begin{eqnarray}
\Pr\left [ \| \bm{x} \|_{\infty} >  \rho_{\max} \sqrt{2 (1+c)   \log p}   \right] &\leq& \frac{2}{p^c}
\end{eqnarray}
as follows
\begin{eqnarray}
\Pr\left [ \| \bm{x} \|_{\infty} > t   \right] \leq \sum_{i=1}^p\Pr\left [ | x_i | > t  \right]  
\leq  2 \sum_{i=1}^p e^{-\frac{ t^2}{2 \Sigma_{ii}}}} \leq  2 p e^{-\frac{ t^2}{2 \max_{i=1,\ldots,p}\Sigma_{ii}}}
\leq  2 e^{\log p -\frac{  t^2}{2 \rho_{\max}   }  , \nonumber
\end{eqnarray}
where $t= \rho_{\max} \sqrt{2 (1+c)   \log p}$ and $\max_{i=1,\ldots,p}\Sigma_{ii}  \leq \rho_{\max}$. 
Second, we prove  
\begin{eqnarray}
\Pr\left [ \| \bm{x} \|_4^2 >  2 (1+c)  \rho_{\max} \sqrt{p} \log p   \right] &\leq& \frac{2}{p^c}
\end{eqnarray}
in the following way:
\begin{eqnarray}
\Pr \left [ \sqrt{ \sum_{i=1}^p {x_i}^4 } > t \right] &=&\Pr \left [  \sum_{i=1}^p {x_i}^4  > t^2 \right] \leq \Pr \left [  p \max_{i=1,\ldots,p} {x_i}^4  > t^2 \right]  \nonumber 
\\
&\leq& \Pr \left [  \left \|\bm{x}  \right \|_{\infty}  > \left ( \frac{t^2}{p} \right)^{1/4} \right]  \leq  2 e^{\log p -\frac{    t}{2  \rho_{\max} \sqrt{p}}  },
\end{eqnarray}
where $t= 2 (1+c)  \rho_{\max} \sqrt{p} \log p $ yields the desired result. Let $\bm{z} \triangleq \bm{J}^{-1} \bm{x}$ and $\bm{u} \triangleq \bm{X J}^{-1} {\bm x}$, then $\bm{z}$ is zero mean Gaussian with covariance $ \bm{\Sigma_z} = \bm{J}^{-1} \bm{\Sigma J}^{-1}$ and $ \bm{\Sigma_u} = \bm{X J}^{-1}\bm {\Sigma_x J}^{-1} \bm{X}^{\top}$, leading to
\begin{eqnarray}
\sigma_{\max} \left(  \bm{\Sigma_z}  \right) &= & \sigma_{\max} \left( \bm{J}^{-1} \bm{\Sigma J}^{-1} \right)  \leq \frac{\rho_{\max}}{  \nu_{\min}^2 }, \\
\sigma_{\max} \left(  \bm{\Sigma_u}  \right) &= & \sigma_{\max} \left( \bm{X J}^{-1} \bm{\Sigma J}^{-1} \bm{X}^\top\right)  \leq \frac{\rho_{\max}}{  \nu_{\min}^2 }\omega_{\max}.
\end{eqnarray}
Therefore, we have
\begin{eqnarray}
\lefteqn{\Pr \left [  \left \|   \bm{J}^{-1} \bm{x} \right \|_4^2  >  2(1+c) \left(\frac{\rho_{\max}}{  \nu_{\min}^2 } \right)  \sqrt{p} \log p \right]} \nonumber \\
&\leq& \Pr \left [ \left \|   \bm{J}^{-1} \bm{x} \right \|_4^2 >  2(1+c) \sigma_{\max} \left(  \bm{\Sigma_z}  \right)  \sqrt{p} \log p \right] \leq \frac{2}{p^c}, \nonumber
\end{eqnarray}
and
\begin{eqnarray}
\lefteqn{\Pr \left [  \left \|  \bm{ X J}^{-1} x \right \|_4^2  >  2(1+c) \left(\frac{\rho_{\max}}{  \nu_{\min}^2 } \omega_{\max} \right)  \sqrt{m} \log m \right] } \nonumber \\
&\leq& \Pr \left [ \left \| \bm{X  J}^{-1} x \right \|_4^2 >  2(1+c) \sigma_{\max} \left( \bm{ \Sigma_u}  \right)  \sqrt{m} \log m \right] \leq \frac{2}{m^c}. \nonumber
\end{eqnarray}

\end{proof}

\subsection{Proof of Corollary \ref{cor:aloVSlo}} \label{proof:corofMain}

To bound $|\lo-\alo|$ we use the following variable 
\begin{eqnarray*}
\kappa_i \triangleq   \left( \frac{\ld_i(\bl)}{\ldd_i(\bl)} \right) \frac{  H_{ii}}{1 -  H_{ii}} -\bm{x_i}^\top \bm{\Delta}_{/i}^*
\end{eqnarray*}
as follows:
\begin{eqnarray}
\lefteqn{\left | \lo-\alo \right| =
\left| \frac{1}{n} \sum_{i=1}^n \phi \left (y_i,  \bm{x_i}^\top \bli   \right) - \frac{1}{n} \sum_{i=1}^n \phi \left (y_i,  \bm{x_i}^\top \bl +\left(\frac{\ld_i(\bl)}{\ldd_i(\bl)}\right) \left(  \frac{H_{ii}}{1 - H_{ii}} \right)     \right) \right|}
\nonumber
\\
 &=&
\left| \frac{1}{n} \sum_{i=1}^n \phi \left (y_i,  \bm{x_i}^\top \bl + \bm{x_i}^\top \bm{\Delta}_{/i}^*   \right) - \frac{1}{n} \sum_{i=1}^n \phi \left (y_i,  \bm{x_i}^\top \bl +\left(\frac{\ld_i(\bl)}{\ldd_i(\bl)}\right) \left(  \frac{H_{ii}}{1 - H_{ii}} \right)     \right) \right|
\nonumber
\\
 &=&
\left| \frac{1}{n} \sum_{i=1}^n \pd \left (y_i,  \bm{x_i}^\top \bli + a_i   \kappa_i \right) \kappa_i \right|.
\end{eqnarray} 
where $a_1, \ldots, a_n$ denote $n$ numbers between 0 and 1. Note that we have $\kappa_i < \frac{C_o}{\sqrt{p} }$ with probability at least $1- \left(\frac{8n}{(n-1)^3}+\frac{8n}{p^3}+4ne^{-p} \right)- q_n - \tilde{q}_n$. Therefore, with at least the same probability we have
\begin{eqnarray}
\left | \lo-\alo \right|
&\leq& \frac{C_o}{\sqrt{p} } \times  \max_{i=1,\ldots,n}   \sup_{|b_i| <  \frac{C_o}{\sqrt{p} }}  \left| \pd \left (y_i,  \bm{x_i}^\top \bli + b_i   \right)  \right|.
\label{eq:loo_approximation2}
\end{eqnarray} 



\if1\longer
{

\subsection{$\alo$ and $\lo$ in the $p$ fixed and large $n$ regime}\label{sec:largen_asymptot}

As we discussed so far, our main concern in this paper is high-dimensional settings in which $n$ is proportional to $p$. However, to present a complete picture about  $\alo$, in this section, we  study it in the classical asymptotic regime where $n$ is  large and $p$ is fixed.  The assumptions presented here will be used throughout Section \ref{sec:largen_asymptot} only. Let
\[
\hat{\bm{\beta}}_{\lambda_n} \triangleq  \underset{\bm{\beta} \in \R^p}{\argmin}  \Bigl \{ \sum_{i=1}^n  \ell ( y_i|\bm{x_i}^\top \bm{\beta} ) + \lambda_n r(\bm{\beta})  \Bigr \}. 
\] 
We also assume that the samples $\{(\bm{x_i}, y_i)\}_{i=1}^n$ are independent and identically distributed, and that $\frac{\lambda_n}{n} \rightarrow \lambda^*$. Define,
\[
\bm{\beta}^* = \argmin_{\bm{\beta} }\mathbb{E} \ell(y_i|\bm{x_i}^\top \bm{\beta}) + \lambda^* r(\bm{\beta}).
\]
Also, define 
\begin{eqnarray}
\bm{R} &\triangleq& \E \left \{\ld (y|\bm{x}^\top \bm{\beta}^*) \pd(y,\bm{x}^\top \bm{\beta}^*) \bm{x x}^\top \right\}, \nonumber \\
\bm{K} &\triangleq& \E \left \{\ldd(y|\bm{x}^\top \bm{\beta}^*) \bm{xx}^\top +  \lambda^* {\rm diag} [\bm{\rdd} ( \bm{\beta}^*)]  \right \}.
\end{eqnarray}
For the sake of brevity, we follow \cite{stone1977asymptotic} and make the following assumptions that enable us avoid repeating standard asymptotic arguments that can be found elsewhere, e.g. in \cite{van2000asymptotic}:
\begin{itemize}
\item[] (B.1) As $n \rightarrow \infty$ ${\bm{\hat\beta}}_{\lambda_n} \overset{p}{\rightarrow} \bm{\beta}^*$.
\item[] (B.2) $\sup_i \|{\bm{\hat \beta}}_{\lambda_n \slash i} - \bm{\beta}^*\|_2 \overset{p}{\rightarrow} 0$.  
\item [] (B.3) Define $\dli \triangleq {\bm{\hat \beta}}_{\lambda_n} - {\bm{\hat \beta}}_{\lambda_n \slash i}$. Let $b_1, b_2, \ldots, b_n $ and $c_1, c_2, \ldots, c_n$ denote $2n$ numbers between $[0,1]$ that may depend on the dataset $\mathcal{D}$.  Then, we assume that
\[
\frac{\bm{X} \diag\left [  \bm{\ldd} \left( {\bm{\hat \beta}}_{\lambda_n} + b_i \dli \right)  \right] \bm{X}^\top + \lambda_n {\rm diag} [\bm{\rdd} ( {\bm{\hat \beta}}_{\lambda_n} + c_i\dli )]}{n} \overset{p}{\rightarrow}  \bm{K}. 
\]
\item[] (B.4) Let $a_1, \ldots, a_n$ denote $n$ number between $0,1$ that may depend on dataset $\mathcal{D}$. Then, assume that
$$\frac{1}{n} \sum_{i=1}^n \bm{x}_i \bm{x}_i^\top \ld_i({\bm{\hat \beta}}_{\lambda_n \slash i}) \pd \left(y_i,\bm{x}_i^\top {\bm{\hat \beta}}_{\lambda_n} + a_i \bm{x}_i^\top \dli \right)   \overset{p}{\rightarrow}  \bm{R}.$$ 

\item[] (B.5) Note that $$H_{ii}= \bm{x_i}^\top  \left( \bm{X} \diag [ \bm{ \ldd} ( {\bm{ \hat \beta}}_{\lambda_n}  )  ] \bm{X}^\top + \lambda {\rm diag} [\bm{\rdd} ({\bm{\hat \beta}}_{\lambda_n })] \right)^{-1} \bm{x}_i \ldd_i ({\bm{\hat \beta}}_{\lambda_n}).$$ Hence, we also assume that $H_{ii} \overset{p}{\rightarrow} 0$. 

\item[] (B.6) Let $d_1, d_2, \ldots, d_n$ denote $n$ numbers between $[0,1]$. Note that we have already assumed that $\sup_i H_{ii} \overset{p}{\rightarrow} 0$. We further assume that $$ \sum_{i=1}^n \left(\frac{\bm{x}_i \bm{x}_i^\top}{n}\right) \left(\frac{\ld_i({\bm{ \hat \beta}}_{\lambda_n})}{1 - H_{ii}}\right) \pd \left(y_i,\bm{x}_i^\top \bl +d_i \left(\frac{\ld_i({\bm{\hat \beta}}_{\lambda_n})}{\ldd_i({\bm{\hat \beta}}_{\lambda_n})}\right) \left(  \frac{H_{ii}}{1 - H_{ii}} \right) \right) \overset{p}{\rightarrow} \bm{R}. $$
\end{itemize}

It should be clear that all these assumptions can be proved under appropriate regularity conditions on the loss function and the regularizer. Note that according to this theorem the error between $\alo$ and $\lo$ is $o_p(1/n)$.

\begin{theorem}\label{thm:alo_lo_largen}
Under assumptions (B.1), (B.2), \ldots, (B.6), we have $n({\rm ALO}_{\lambda_n}-{\rm LO}_{\lambda_n}) \overset{p}{\rightarrow} 0,$
as $n \rightarrow \infty$. 
\end{theorem}
\begin{proof}
For notational simplicity instead of using $\lambda_n$ we use $\lambda$ in our formulas. However, the reader should note that $\lambda/n \rightarrow \lambda^*$. First note that the gradient condition implies that $\XI \bm{\ld}_{/ i}(\bli) + \lambda \bm{ \rd}(\bli)=0.$ Hence,
\begin{equation}\label{eq:firstgradmin1}
\bm{X} \bm{\ld}(\bli) +\lambda  \bm{\rd}(\bli)= \ld_i(\bli)\bm{x_i}. 
\end{equation}
Furthermore, we can use the fact that $\bm{X \ld}(\bl) + \lambda \bm{\rd}(\bl) = 0$ to obtain  
\[
\bm{X \ld}(\bli) +\lambda \bm{ \rd}(\bli) - \bm{X \ld}(\bl) - \lambda \bm{ \rd}(\bl)= \ld_i(\bli)\bm{x_i}
\]
Since both the loss function and the regularizer are assumed to be twice continuously differentiable, we can use the mean value theorem to simplify this expression to
\begin{eqnarray}\label{eq:finalgradlargen}
\left ( \bm{X} \diag\left [  \bm{\ldd} \left( \bl + b_i \dli \right)  \right] \bm{X}^\top + \lambda \diag \left[\bm{\rdd} \left( \bl + c_i\dli \right)\right]   \right) \dli  = \ld_i(\bli)\bm{x_i},
\end{eqnarray}
 where all $b_i$s and $c_i$s are in $[0,1]$. Furthermore, if $\phi$ is continuously differentiable, then we can again use the mean value theorem to obtain
 \begin{eqnarray}
\lefteqn{ \lo = \frac{1}{n} \sum_{i=1}^n \phi(y_i, \bm{x_i}^\top \bl) + \frac{1}{n} \sum_{i=1}^n \bm{x_i}^\top \dli \pd \left(y_i,\bm{x_i}^\top \bl + a_i \bm{x_i}^\top \dli \right)  = \frac{1}{n} \sum_{i=1}^n \phi(y_i, \bm{x_i}^\top \bl)} \nonumber \\
&+& \frac{1}{n} \sum_{i=1}^n \bm{x_i}^\top  \left( \bm{X} \diag\left [ \bm{ \ldd }\left( \bl + b_i \dli \right)  \right] \bm{X}^\top + \lambda \diag\left [\bm{ \rdd} \left( \bl + c_i\dli \right) \right]  \right)^{-1} \bm{x_i} \nonumber\\
&\times&\ld_i(\bli) \pd \left(y_i,\bm{x_i}^\top \bl + a_i \bm{x_i}^\top \dli \right). \nonumber
 \end{eqnarray}
Hence,
\begin{eqnarray}
\lefteqn{ \lo - \frac{1}{n} \sum_{i=1}^n \phi(y_i, \bm{x_i}^\top \bl) } \nonumber \\
&=& \frac{1}{n} \sum_{i=1}^n {\rm trace} \left[  \left( \bm{X} \diag\left [ \bm{ \ldd }\left( \bl + b_i \dli \right)  \right]\bm{ X}^\top + \lambda \diag[\bm{\rdd} \left( \bl + c_i\dli \right)]  \right)^{-1} \bm{x_i} \bm{x_i}^\top \right]
\nonumber
\\
&\times&
\ld_i(\bli) \pd \left(y_i,\bm{x_i}^\top \bl + a_i \bm{x_i}^\top \dli \right). \nonumber
\end{eqnarray}
It is then straightforward to use Assumptions (B.3) and (B.4) to claim that
\[
n (\lo - \frac{1}{n} \sum_{i=1}^n \phi(y_i, \bm{x_i}^\top \bl)) \overset{p}{\rightarrow} {\rm trace} (\bm{K}^{-1}\bm{R}). 
\]
Similarly, we can use the mean value theorem to argue that 
\begin{eqnarray}
\alo &=& \frac{1}{n} \sum_{i=1}^n \phi(y_i, \bm{x_i}^\top \bl) 
\nonumber
\\
&+& 
\frac{1}{n} \sum_{i=1}^n   \left(  \frac{H_{ii}}{\ldd_i(\bl)} \right) \times  \left(\frac{\ld_i(\bl)}{1 - H_{ii}}\right)  \pd \left (y_i,  \bm{x_i}^\top \bl +d_i \left(\frac{\ld_i(\bl)}{\ldd_i(\bl)}\right) \left(  \frac{H_{ii}}{1 - H_{ii}} \right)     \right)
\nonumber
\\
&=& 
\frac{1}{n} \sum_{i=1}^n \phi(y_i, \bm{x_i}^\top \bl) + 
\frac{1}{n}  \text{trace}\Bigg [
\left( \frac{\bm{X} \diag [ \bm{ \ldd} ( \bl  )  ] \bm{X}^\top + \lambda \diag[\bm{\rdd}( \bl )]}{n}  \right)^{-1} \nonumber
\\
&\times&
 \sum_{i=1}^n      
\left(\frac{\bm{x_i} \bm{x_i}^\top}{n}\right)\left(  \frac{  \ld_i(\bl) }{1 - H_{ii}} \right)  \pd \left (y_i,  \bm{x_i}^\top \bl +d_i \left(\frac{\ld_i(\bl)}{\ldd_i(\bl)}\right) \left(  \frac{H_{ii}}{1 - H_{ii}} \right)     \right) \Bigg] \nonumber
\end{eqnarray}
with  $|d_i| \leq 1$, $i=1,\cdots,n$. Again it is straightforward to use Assumptions (B.3) and (B.6) to show that
\begin{equation}
n(\alo-\frac{1}{n} \sum_{i=1}^n \phi(y_i, \bm{x_i}^\top \bl) )\overset{p}{\rightarrow} {\rm trace} (\bm{K}^{-1}\bm{R}). 
\end{equation}

\end{proof}
} \fi

\end{document}